\documentclass[journal,twoside]{IEEEtran}
\usepackage{color}
\usepackage{subfigure}
\usepackage{graphicx}
\usepackage{booktabs}
\usepackage{longtable}
\usepackage{amssymb}
\usepackage{amsmath}
\usepackage{cite}
\usepackage{multirow}
\usepackage{url}
\usepackage{tabularray}
\usepackage{pgfplots}
\pgfplotsset{compat=1.5}
\usepgfplotslibrary{groupplots}
\usepackage{algorithm}
\usepackage{algpseudocode}
\usepackage{amsmath}

\usepackage{colortbl}
\usepackage{algorithmicx}
\usepackage{nicematrix}
\usepackage{threeparttable}    

\usepackage{xcolor}
\definecolor{darkblue}{rgb}{0, 0, 0.6}
\definecolor{darkgreen}{rgb}{0, 0.7, 0}
\definecolor{darkred}{rgb}{0.8, 0, 0}
\usepackage{pifont}
\newcommand*\colorcheck{%
  \expandafter\newcommand\csname greencheck\endcsname{\textcolor{darkgreen}{\ding{52}}}%
}
\newcommand*\colorcross{%
  \expandafter\newcommand\csname redcross\endcsname{\textcolor{darkred}{\ding{56}}}%
}
\colorcheck
\colorcross
\usepackage{mathrsfs}
\usepackage{listings}
\usepackage{tabularx} 

\usepackage{graphicx} 
\usepackage{epsfig}
\usepackage{epstopdf} 
\usepackage[switch, pagewise]{lineno} 

\usepackage[
    colorlinks=true,  
    linkcolor=red,   
    urlcolor=blue,    
    citecolor=green    
]{hyperref}

\newcommand{\secondbest}[1]{\underline{{#1}}}

\usepackage{paralist}

\usepackage{amsthm} 

\newtheorem{theorem}{Theorem}
\newtheorem{lemma}{Lemma}

\newcommand{\x}{{\mathbf{x}}}

\newcommand{\y}{{\mathbf{y}}}
\newcommand{\z}{{\mathbf{z}}}

\newcommand{\f}{{\mathbf{f}}}

\newcommand{\Rmnum}[1]{\expandafter\@slowromancap\romannumeral #1@}
\begin{document}

\title{TVRN: Invertible Neural Networks for Compression-Aware Temporal Video Rescaling}

\markboth{Accepted by IEEE Transactions on Image Processing}%
{Feng \MakeLowercase{\textit{et al.}}: TVRN: Invertible Neural Networks for Compression-Aware Temporal Video Rescaling}

\author{
Xinmin Feng,
Li Li, \IEEEmembership{Senior Member, IEEE},
Dong Liu, \IEEEmembership{Senior Member, IEEE},
and Feng Wu, \IEEEmembership{Fellow, IEEE}
\thanks{
Date of current version \today. This work was supported by the Natural Science Foundation of China under Grant 62021001. We acknowledge the support of GPU cluster built by MCC Lab of Information Science and Technology Institution, USTC.

The authors are with the MOE Key Laboratory of Brain-Inspired Intelligent Perception and Cognition, University of Science and Technology of China, Hefei 230093, China.
(e-mail: xmfeng2000@mail.ustc.edu.cn;  lil1@ustc.edu.cn; dongeliu@ustc.edu.cn; fengwu@ustc.edu.cn).
}
}

\maketitle

\begin{abstract}
To fit diverse display and bandwidth constraints, high-frame-rate videos are temporally downscaled to low-frame-rate (LFR) and later upscaled, requiring joint optimization for effective frame-rate rescaling.
However, existing methods typically link the two operations via training objectives, without fully exploiting their reciprocal nature, which may cause high-frequency information loss. Moreover, they overlook the impact of lossy codecs on LFR videos, limiting real-world applicability.
In this work, we propose an end-to-end framework for compression-aware frame-rate rescaling, named TVRN. To regularize high-frequency information lost during frame-rate downscaling, TVRN adopts an invertible architecture that combines a Multi-Input Multi-Output Temporal Wavelet Transform with a high-frequency reconstruction module. To enable end-to-end training through non-differentiable lossy codecs, we design a surrogate network that approximates their gradients. Finally, to improve robustness under various compression levels, we extend TVRN to an asymmetric architecture by incorporating compression-aware features learned via a learning-to-rank strategy.
Extensive experiments show that TVRN outperforms existing methods in reconstruction quality under industrial video compression settings. Source code is publicly available at \url{https://github.com/fengxinmin/TVRN_public}.

\end{abstract}


\section{Introduction}
\label{introduction}

\IEEEPARstart{W}ITH the rapid growth of high-frame-rate (HFR) video content, efficient transmission of such videos has become a critical problem. In practical applications  such as video streaming, content providers often downsample HFR videos to low-frame-rate (LFR) formats to alleviate network congestion, later restoring the original frame rate on the client side \cite{james2019beta, palmer2021voxel, wang2023reparo, meng2023enabling, Yin2023SAFR, wang2021enabling}. This raises the question of whether the downscaling process can be optimized to facilitate high-quality video reconstruction. Video frame-rate rescaling, also termed temporal video rescaling, aims to optimize the downscaling and upscaling processes jointly while ensuring visually pleasing LFR videos for efficient transmission and preview.

Early methods typically discard frames and use a video frame interpolation network to recover the original frame rate \cite{wang2023reparo,Yin2023SAFR}, but they treat downscaling and upscaling as separate tasks, ignoring their reciprocal nature.
Recent approaches aim to jointly optimize both processes \cite{xiang2022learning, zhang2025continuous}, where the downscaling network embeds motion information from HFR videos into LFR videos in a visually imperceptible manner, and the upscaling network reconstructs the original HFR video. 
However, these methods only rely on learning objectives for HFR recovery, neglecting the modeling of motion information loss during downsampling, which potentially degrades reconstruction quality.
More importantly, they apply rounding operations to compress  LFR videos, limiting compatibility with lossy video codecs like HEVC \cite{sullivan2012overview}. Thus, \textit{the challenge of achieving effective frame-rate rescaling under non-differentiable lossy codecs remains largely unsolved}.

Recent advances in spatial video rescaling have shown the effectiveness of invertible architectures in mitigating information loss \cite{huang2021video, tian2021self, ding2024towards, ho2022video, tian2023clsa}. While Haar wavelets are widely used for spatial frequency decomposition, temporal Haar wavelets are unsuitable for frame-rate rescaling due to the need for additional motion cues \cite{dong2023temporal}. Besides, invertible architectures typically model discarded high-frequency components as distributions conditioned on low-frequency ones, then reconstruct them using residual networks with downscaled video as input. However, motion-induced spatial misalignment across frames limits the accuracy of high-frequency component reconstruction in temporal video rescaling.

For end-to-end optimization with non-differentiable lossy video codecs, prior works use surrogate networks to simulate compression distortions, enabling gradient-based backpropagation \cite{tian2023clsa}. Built upon Dense3D-T blocks \cite{iandola2014densenet}, these networks struggle to capture inter-frame compression artifacts, particularly those arising from motion compensation. This limitation is critical for temporal video rescaling, where accurately modeling inter-frame distortions is key to robust motion steganography in LFR videos under various compression levels.

To address these limitations, we introduce the \textbf{T}emporal \textbf{V}ideo \textbf{R}escaling \textbf{N}etwork (\textbf{TVRN}). It integrates the proposed Multi-Input Multi-Output Temporal Wavelet Transform (MIMO-TWT) and Video Rescaling Network (MIMO-VRN) to regularize information loss by modeling discarded temporal high-frequency components as distributions conditioned on downscaled videos. 
MIMO-TWT decomposes videos into temporal frequency frames, removing the need for additional motion information. We recover the spatially misaligned high-frequency components using a reconstruction module guided by bidirectional optical flow and contextual features.
Moreover, a surrogate network mimics compression distortions from non-differentiable codecs, enabling end-to-end optimization. This network uses compression-aware invertible blocks to emulate distortions by degrading high-frequency components. We further extend the framework to an asymmetric architecture for enhanced robustness across various compression levels.
Extensive experiments show that our framework outperforms both frame-skipping and learning-based methods in temporal video rescaling. To the best of our knowledge, this is the first learnable frame-rate rescaling method to outperform state-of-the-art VFI methods under modern lossy codecs.

In summary, our main contributions are as follows:
\begin{itemize}

\item To enable video frame-rate rescaling with an invertible architecture, we design the MIMO-TWT and the bidirectional optical flow-guided high-frequency reconstruction module, effectively regularizing and reconstructing high-frequency information lost during downscaling.

\item We introduce a surrogate network with compression-aware invertible networks to estimate the gradient of non-differentiable lossy codecs, extending the framework to compression-aware temporal video rescaling.
\item We identify the challenges of integrating video quality enhancement methods into the proposed approach and extend our framework to an asymmetric architecture by incorporating compression-aware features.
\item Experimental results show that our framework achieves better HFR reconstruction quality compared to previous methods, while maintaining visually friendly LFR videos. 
\end{itemize}

\begin{figure*}
  \centering
  \includegraphics[width=1.0\textwidth]{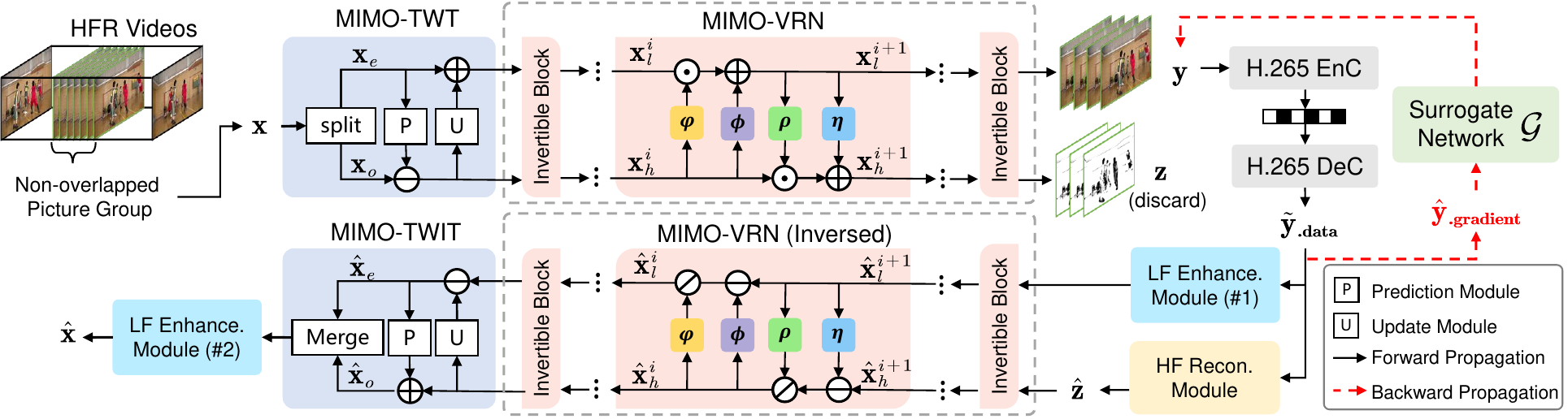}
  \caption{
  \textbf{Overview of the proposed Temporal Video Rescaling Network (TVRN)}.
  The MIMO-TWT followed by MIMO-VRN first decomposes a group of frames $\x$ from high-frame-rate (HFR) videos into visually pleasing low-frame-rate (LFR) videos $\y$ and high-frequency components $\z$. To reconstruct high-frame-rate videos $\widehat{\x}$, the inverse process is applied using the compressed LFR videos $\widetilde{\y}$ and the reconstructed high-frequency components $\widehat{\z}$.
To enable end-to-end training, the surrogate network $\mathcal{G}$ is introduced to mimic compression distortions. To further enhance robustness against distribution shifts caused by varying compression levels, we incorporate two identical video quality enhancement modules before (\#1) and after (\#2) upscaling, respectively.
  }
  \label{fig: framework}
\end{figure*}

\section{Related Work}
\label{gen_inst}


\textbf{Invertible Neural Network.} 
Invertible Neural Networks (INNs) are a class of neural networks composed of a series of bijective mappings, enabling exact recovery of the original data from its latent representation. INNs are widely used in generative modeling due to their ability to explicitly compute the probability density of the target distribution. For example, NICE \cite{dinh2014nice} introduced the additive coupling layer as an invertible transformation, splitting input dimensions to facilitate Jacobian computation for likelihood estimation. RealNVP \cite{dinh2016density} extended this with affine coupling layers to increase model expressiveness. Glow \cite{kingma2018glow} further enhanced the architecture with activation normalization and invertible 1$\times$1 convolutions, enabling more complex data distribution modeling. In addition to generative tasks, INNs have been applied to various image processing applications, including denoising, super-resolution, and compressive sensing 
\cite{li2021dehazeflow, huang2022winwave, liu2021idn}.

\textbf{Image Rescaling.} 
Recently, image rescaling has gained increasing attention. Early approaches mainly used encoder-decoder architectures for downscaling and upscaling \cite{kim2018task, li2018learning, sun2020learned, chen2020hrnet}. In contrast, Xiao \textit{et al.} \cite{xiao2020invertible, xiao2023invertible} proposed an invertible framework that models the bidirectional degradation and restoration process from a new perspective.
Subsequent methods have modeled high-frequency components with more complex distributions, such as Gaussian mixtures and conditional distributions \cite{chen2021direct, liang2021hierarchical, zhu2022high}. Other works focus on adapting the reversible framework to practical scenarios \cite{yang2023self, qi2023real,DBLP:journals/spl/HuangSC23,wang2025timestep}. For instance, 
{Huang \textit{et al.} \cite{DBLP:journals/spl/HuangSC23} presented a learned scale-arbitrary image downscaling method to downscale HR images to a target LR ones where it could be well upscaled by traditional simple upscaling methods. 
Wang \textit{et al.} \cite{wang2025timestep} explored extreme image rescaling with a timestep-aware diffusion model  that operates in the latent space of a pretrained autoencoder.
}
Additionally, this framework has been extended to various computer vision tasks, such as image colorization \cite{zhu2022high} and compression \cite{gao2023extremely}.

\textbf{Video Rescaling.}
Several notable studies have focused on jointly optimizing video downscaling and upscaling. MIMO-VRN \cite{huang2021video}  used an invertible coupling architecture with a MIMO strategy to handle multiple frames simultaneously. Tian \textit{et al.} \cite{tian2021self} proposed a self-conditioned probabilistic framework, which was further extended to lossy video compression and efficient action recognition. More recently, researchers have enhanced invertible video rescaling by leveraging temporal correlations within groups of frames \cite{ding2024towards, xiang2022learning,ho2022video}.

Efforts have also been made to extend video scaling from the spatial to the temporal domain \cite{xiang2022learning, zhang2025continuous}. For instance, STAA \cite{xiang2022learning} applies 3D low-pass filtering for temporal downscaling and recovers high-frequency details via space-time pixel shuffle. Concurrently, CSTVR \cite{zhang2025continuous} introduced a partially invertible architecture for continuous frame-rate rescaling. {However, these learned video frame-rate rescaling methods fall short of fully implementing an invertible architecture to mitigate information loss and simplify lossy compression to rounding-based quantization, restricting their practical deployment. 
Compared with existing invertible rescaling and motion-steganography pipelines, our method introduces several task-specific designs.}

\textbf{Video Frame Interpolation.}
Existing video frame interpolation (VFI) methods  based on deep learning can be generally categorized as flow-based or kernel-based. 
 Flow-based methods rely on optical flow estimation to generate interpolated frames.  Some methods~\cite{ niklaus2018context, niklaus2020softmax, sim2021xvfi} infer intermediate optical flows using the flows between two input frames by assuming certain motion types. Others~\cite{park2020bmbc,lu2022video,kong2022ifrnet,jin2025upr,lew2025disentangled, liu2024sparse} directly estimate the intermediate flows. Kernel-based methods argue that optical flows are unreliable in dynamic texture scenes, predicting locally adaptive convolution kernels to synthesize output pixels, thus allowing more flexible many-to-one mapping~\cite{ding2021cdfi, cheng2020video, cheng2021multiple, danier2022enhancing,zhang2023extracting, guo2024generalizable}.

\section{Methods}

\subsection{Preliminaries}
The invertible rescaling process can be viewed  as a distribution transformation. Let \( \mu \) denote the joint distribution of the original image \( \mathbf{x} \), and let \( \mathcal{F}_\sharp \)  be an invertible transformation. The \textit{pushforward measure} of \( \mu \) under \( \mathcal{F}_\sharp \) is defined as:
\begin{equation}
(\mathcal{F}_\sharp)_\# \mu = \nu = \nu_L \times \nu_H,
\end{equation}
where \( \nu \) is the joint distribution of the transformed variables $(\mathbf{x}_L,\mathbf{x}_H)$, and \( \nu_L \) and \( \nu_H \)  are the marginal distributions of the downsampled image \(\mathbf{x}_L\) and the high-frequency components \(\mathbf{x}_H\), respectively.
For any \textit{measurable set} \( {B} \) in the target space of $\mathcal{F}_\sharp$, where each element corresponds to a possible realization of $(\mathbf{x}_L, \mathbf{x}_H)$, the pushforward measure satisfies:
\begin{equation}
(\mathcal{F}_\sharp)_\# \mu(B) = \mu(\mathcal{F}_\sharp^{-1}(B)).
\end{equation}
Following~\cite{xiao2020invertible}, if we approximate \( \nu \) by \( \nu_L \times \mathcal{N}(\boldsymbol{\mu}_H, \boldsymbol{\sigma}_H^2) \), we can transmit only the downsampled image \( \mathbf{x}_L \), generate the high-frequency part \( \mathbf{x}_H \) by sampling from the Gaussian, and then reconstruct the original image by applying \( \mathcal{F}_\sharp^{-1} \).
This assumption is theoretically sound because, for any continuous random variable with density (\textit{e.g.}, \( \mathbf{x}_H \sim p(\mathbf{x}_H \mid \mathbf{x}_L) \)), there exists a bijection \( f_{{\mathbf{x}_H}} \) such that \( f_{{\mathbf{x}_H}}(\mathbf{x}_H) \sim \mathcal{N}(0, \mathbf{I}) \)~\cite{hyvarinen1999nonlinear}. This means \( \nu_H \) can be approximated by a simple Gaussian independent of input content.
In practice, however, we often model 
\( \nu_H \)  as a conditional distribution dependent on \( \nu_L \)   to improve reconstruction quality~\cite{chen2021direct, liang2021hierarchical, zhu2022high}.

\subsection{Framework Overview}
An overview of the proposed temporal video rescaling framework is shown in Fig.~\ref{fig: framework}. The system learns an invertible neural network \(\mathcal{F}_\sharp\) to perform frame-rate downscaling and upscaling through its forward and inverse mappings:
\begin{equation}
    (\mathbf{y}, \mathbf{z}) = \mathcal{F}_\sharp(\mathbf{x}), \quad
    \widehat{\mathbf{x}} = \mathcal{F}_\sharp^{-1}(\widetilde{\mathbf{y}}, \hat{\mathbf{z}}),
\end{equation}
where \(\mathbf{x}\) and \(\mathbf{y}\) are the HFR and LFR videos, \(\mathbf{z}\) is the high-frequency component; \(\widetilde{\mathbf{y}}\) is the compressed LFR video; and \(\hat{\mathbf{z}}\) is the reconstructed high-frequency component.

During downscaling, the proposed Multi-Input Multi-Output Temporal Wavelet Transform (MIMO-TWT, $\S$ \ref{sec: MIMO_TWT}) first decomposes \(\mathbf{x}\) into temporal low-frequency \(\mathbf{x}_l\) and high-frequency \(\mathbf{x}_h\) frames. The Multi-Input Multi-Output Video Rescaling Network (MIMO-VRN, $\S$~\ref{sec: MIMO_VRN}) further converts \(\mathbf{x}_l\) into the visually pleasing LFR video \(\mathbf{y}\), while \(\mathbf{x}_h\) is mapped to the high-frequency component \(\mathbf{z}\). 
The LFR video \(\mathbf{y}\) is then compressed using a lossy codec (e.g., H.265/HEVC \cite{sullivan2012overview}) to produce the transmitted video \(\widetilde{\mathbf{y}}\).
During upscaling, the inverse MIMO-VRN and MIMO-TWT networks reconstruct the original HFR video \(\widehat{\mathbf{x}}\) from the compressed LFR video \(\widetilde{\mathbf{y}}\) and the reconstructed high-frequency component \(\hat{\mathbf{z}}\), which is synthesized with guidance from bi-directional optical flow and contextual features ($\S$~\ref{sec: hf_recon}).

The system is optimized by minimizing a loss that balances reconstruction fidelity and visual similarity between the LFR video $\y$ and the even-indexed frames \(\mathbf{x}_e\) of the original video:
\begin{equation}
    \mathcal{L}_{\text{basic}} = \ell_1(\mathbf{x}, \widehat{\mathbf{x}}) + \lambda \cdot \ell_1(\mathbf{y}, \mathbf{x}_e),
    \label{equ: total}
\end{equation}
where \(\lambda\) stabilizes training and bitrate overhead caused by high-frequency embedding in LFR frames.

To jointly optimize temporal downscaling and upscaling in an end-to-end manner, we design a surrogate neural network \(\mathcal{G}\), built upon compression-aware invertible blocks, as a gradient estimator for non-differentiable lossy codecs ($\S$ \ref{sec: surrogate_network}). The loss function for training the surrogate network is defined as:
\begin{equation}
\mathcal{L}_{\text{surrogate}} = \ell_1(\mathcal{G}(\mathbf{y}), \mathcal{H}(\mathbf{y})),
\label{equ: surrogate_network}
\end{equation}
where \(\mathcal{H}(\cdot)\) denotes the actual video codec. The temporal video rescaling network and surrogate network are jointly trained with an alternate optimization strategy (§ \ref{sec: training_strategy}).

Given that symmetric invertible architectures are sensitive to distribution shifts from compression artifacts \cite{yang2023self}, we extend our framework to an asymmetric architecture by integrating video quality enhancement (VQE) methods ($\S$ \ref{sec: self_asymmetric}). However, directly applying VQE models to remove compression artifacts from LFR frames proves ineffective, particularly at high bitrate compression, as the VQE process distorts the high-frequency components embedded in the downscaled videos. To overcome this, we treat upscaling as a decoupling process that decouples the useful high-frequency components from compressed LFR videos. We further enhance the VQE model by incorporating compression-aware features learned through a pairwise learning-to-rank strategy, improving restoration performance across various compression levels.

\subsection{MIMO-TWT}
\label{sec: MIMO_TWT}

Most previous efforts in image and video rescaling have relied on the 2D-Haar Wavelet Transform to separate low- and high-frequency information \cite{xiao2020invertible, huang2021video, tian2021self, yang2023self, tian2023clsa}. However, traditional temporal wavelet transforms cannot be directly applied to frame-rate rescaling due to their Multi-Input Single-Output (MISO) structure, which requires transferring additional motion cues between forward and inverse transforms \cite{dong2023temporal}. To address this, we introduce the Multi-Input Multi-Output Temporal Wavelet Transform (MIMO-TWT) and Inverse Transform (MIMO-TWIT), which use video frame interpolation instead of unidirectional motion estimation, eliminating the need for transferring motion cues.

MIMO-TWT employs a prediction-first lifting scheme for wavelet transform~\cite{claypoole2003nonlinear}.
As shown in Fig.~\ref{fig: framework}, this process involves three steps: split/merge, prediction, and update. Given a sequence of \(N\) video frames \(\mathbf{x} \in \mathbb{R}^{N \times 3 \times H \times W}\), we first split it into even- and odd-indexed frames: \(\mathbf{x}_e \in \mathbb{R}^{\lceil N / 2 \rceil \times 3 \times H \times W}\) and \(\mathbf{x}_o \in \mathbb{R}^{\lfloor N / 2 \rfloor \times 3 \times H \times W}\), respectively.
Next, the odd frames \(\mathbf{x}_o\) are predicted from the even frames \(\mathbf{x}_e\), and the residuals between original and predicted odd frames, representing the temporal high-frequency components \(\mathbf{x}_h\), are computed. These residuals are then used to update the even frames, yielding the temporal low-frequency components \(\mathbf{x}_l\), as follows:
\begin{equation}
\mathbf{x}_h = \mathbf{x}_o - \mathbf{P}(\mathbf{x}_e), \quad \mathbf{x}_l = \mathbf{x}_e + \mathbf{U}(\mathbf{x}_h).
\label{equ: mimo_twt}
\end{equation}
Here, \(\mathbf{P}\) and \(\mathbf{U}\) represent the prediction and update modules. Specifically, we use a lightweight video frame interpolation network~\cite{jin2025upr} for \(\mathbf{P}\), and stacked Dense3D-T blocks~\cite{iandola2014densenet} for \(\mathbf{U}\).
This design allows MIMO-TWT and MIMO-TWIT to effectively decompose and reconstruct video frames from their temporal low- and high-frequency components.

\subsection{MIMO-VRN}
\label{sec: MIMO_VRN}

Following MIMO-VRN \cite{huang2021video}, we adopt a series of multi-input multi-output invertible coupling blocks, as shown in Fig.~\ref{fig: framework}, to learn a mapping that encodes the case-specific high-frequency motion information within the low-frequency components. Each invertible block comprises four coupling layers. 
Specifically, in the $i$-th invertible block, the forward and inverse transformations are defined as:
\begin{align}
&\text{Forward:}&&\left\{
\begin{aligned}
&\mathbf{x}_{l}^{i+1} = \mathbf{x}_l^i \odot \exp\left(\varphi(\mathbf{x}_h^i)\right) + \phi(\mathbf{x}_h^i) , \\
&\mathbf{x}_h^{i+1} = \mathbf{x}_h^i \odot \exp \left(\rho(\mathbf{x}_l^{i+1})\right) + \eta(\mathbf{x}_l^{i+1}),
\end{aligned}
\right.
\\
&\text{Inverse:}&&\left\{
\begin{aligned}
&\hat{\mathbf{x}}_h^i = \left(\hat{\mathbf{x}}_h^{i+1} - \eta(\hat{\mathbf{x}}_l^{i+1})\right) \odot \exp \left(-\rho(\hat{\mathbf{x}}_l^{i+1})\right), \\
&\hat{\mathbf{x}}_l^i = \hat{\mathbf{x}}_l^{i+1} - \phi(\hat{\mathbf{x}}_h^i) \odot \exp\left(-\varphi(\hat{\mathbf{x}}_h^i)\right).
\end{aligned}
\right.
\label{equ: mimo_vrn}
\end{align}
Here, \(\mathbf{x}_l^i\) and \(\mathbf{x}_h^i\) denote the low- and high-frequency components at the \(i\)-th stage during the downscaling process, and \(\hat{\mathbf{x}}_l^i\), \(\hat{\mathbf{x}}_h^i\) denote their corresponding reconstructions during upscaling. The operator \(\odot\) indicates element-wise Hadamard multiplication. All transformation functions, \(\varphi(\cdot)\), \(\phi(\cdot)\), \(\rho(\cdot)\), and \(\eta(\cdot)\), are implemented using the Dense3D-T architecture \cite{iandola2014densenet}. To maintain numerical stability after exponentiation, we apply a sigmoid activation to normalize the exponential terms $\exp\left(\pm\rho(\cdot)\right)$ and $\exp\left(\pm\varphi(\cdot)\right)$ .

\begin{figure}[t]
  \centering
  \includegraphics[width=0.8\columnwidth]{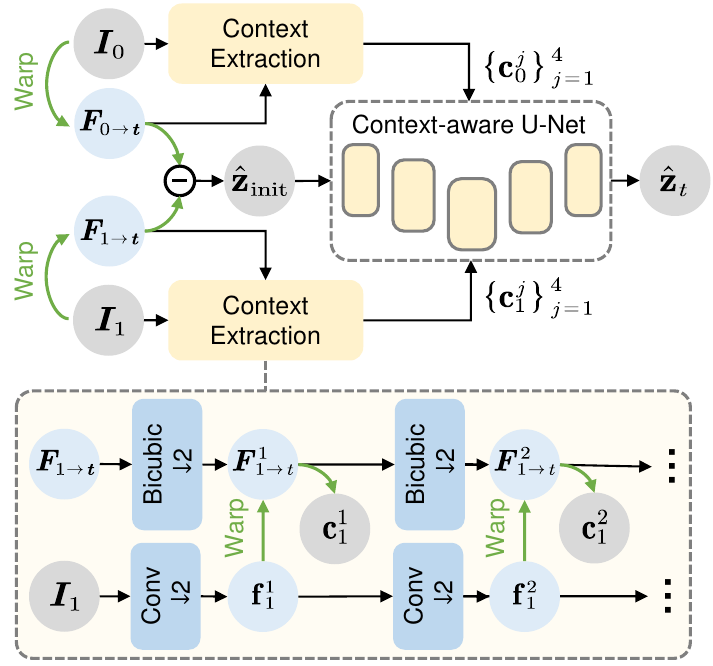}
  \caption{\textbf{Bidirectional Optical Flow-Guided High-Frequency Component Reconstruction Module.} This module reconstructs the high-frequency components $\mathbf{z}_t$ at time $t$ using the neighboring frames $I_0$ and $I_1$. We first compute the difference between the bidirectionally warped frames as the initial estimate $\hat{\mathbf{z}}_{\text{init}}$. Then, multi-scale contextual features are aligned to time $t$ using the bidirectional optical flows $F_{0 \to t}$ and $F_{1 \to t}$. Finally, a context-aware U-Net is designed to refine $\hat{\mathbf{z}}_{\text{init}}$ into $\hat{\mathbf{z}}_t$.
  }
  \label{fig: hf_recon}
\end{figure}

\begin{figure*}
  \centering
  \includegraphics[width=\textwidth]{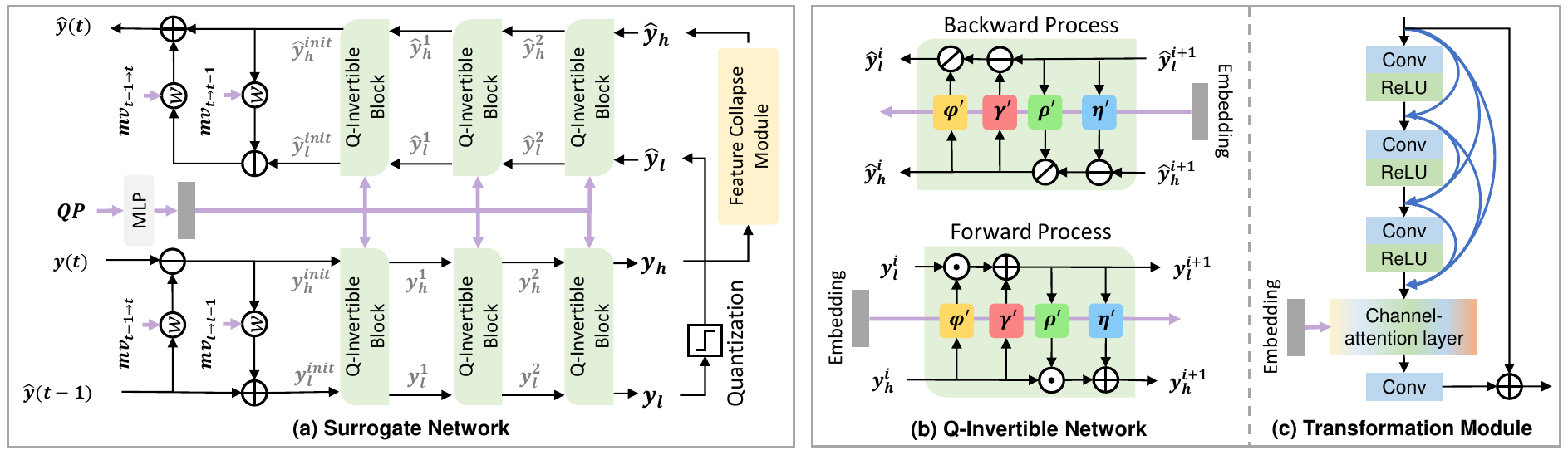}
  \caption{
\textbf{Structure of the proposed surrogate network.} The network emulates non-differentiable lossy video codecs by degrading the original frame $\mathbf{y}_{(t)}$ into a low-quality frame $\hat{\mathbf{y}}_{(t)}$, using the previously degraded frame $\hat{\mathbf{y}}_{(t-1)}$ as a reference. We first apply the MISO temporal wavelet transform with the motion vector derived from bitsreams, followed by stacked Q-Invertible Blocks, to decompose the temporal frequency components. Then, a feature collapse module is used to degrade the high-frequency components $\mathbf{y}_h$, mimicking the lossy compression of inter-frame residuals in traditional video codecs.
  }
  \label{fig: surrogate_network}
\end{figure*}

\subsection{High-Frequency Components Reconstruction}
\label{sec: hf_recon}

Prior works in spatial video rescaling commonly adopt residual block-based predictive modules to estimate the missing high-frequency components from downscaled low-resolution frames~\cite{huang2021video, tian2021self, li2024video}. However, these designs are not ideal for temporal video rescaling, where high-frequency frames may not be spatially aligned with neighboring LFR frames due to motion.
To address this, we propose a bi-directional optical flow-guided high-frequency reconstruction module, as shown in Fig.~\ref{fig: hf_recon}. 
The goal is to reconstruct the discarded high-frequency component \( \hat{\mathbf{z}}_t \) at a given intermediate time \( t \in (0,1) \), using two reference frames \( I_0 \) and \( I_1 \) of the downscaled video \(\hat{\mathbf{y}}\). 
The process consists of four stages:

\textit{1) Flow Estimation:}  
We adopt the lightweight bi-directional motion estimator~\cite{jin2022enhanced} to compute the optical flow \( F_{0 \to 1} \) and \( F_{1 \to 0} \) between the two reference frame \( I_0 \) and \( I_1 \). These are then linearly scaled to generate intermediate flows toward the target time \( t \in (0, 1) \):
\begin{equation}
F_{0 \to t} = t \cdot F_{0 \to 1}, \quad F_{1 \to t} = (1 - t) \cdot F_{1 \to 0}.
\end{equation}

\textit{2) Initial Estimate:}  
The intermediate flows are used to forward-warp the reference frames to time $t$. The difference between the warped frames yields an initial estimate of the high-frequency component:
\begin{equation}
\hat{\mathbf{z}}_\text{init} = \overrightarrow{\mathcal{W}}(I_0, F_{0 \to t}) - \overrightarrow{\mathcal{W}}(I_1, F_{1 \to t}),
\label{eq:initial_estimate}
\end{equation}
where \(\overrightarrow{\mathcal{W}}(\cdot, \cdot)\) denotes the forward warping operation.  

\textit{3) Contextual Feature Extraction:}  
While temporally aligned, $\hat{\mathbf{z}}_\text{init}$ often lacks fine details and suffers from warping holes, as shown in Fig. \textcolor{red}{4}(c) in the supplementary material.  
To refine it, we extract multi-level contextual features from \( I_0 \) and \( I_1 \) using four convolutional layers with stride 2, producing \(\{\mathbf{f}_0^j, \mathbf{f}_1^j\}_{j=0}^3\). Similarly, multi-level intermediate flows \(\{F_{0 \to t}^j, F_{1 \to t}^j\}_{j=0}^3\) are obtained by bicubic downscaling the original flows \(F_{0 \to t}\) and \(F_{1 \to t}\), and used to warp these contextual features toward the target time $t$:
\begin{equation}
\mathbf{c}_i^j = \overrightarrow{\mathcal{W}}(\mathbf{f}_i^j, F_{i \to t}^j), \quad i \in \{0, 1\}, \ j \in \{0,1,2,3\}.
\end{equation}

\textit{4) High-Frequency Synthesis:}  
The aligned contextual features \( \{\mathbf{c}_0^j, \mathbf{c}_1^j\}_{j=0}^3 \) along with the initial estimate \( \hat{\mathbf{z}}_\text{init} \) are fed into a context-aware U-Net to reconstruct the final high-frequency component \( \hat{\mathbf{z}}_t \).  
The U-Net architecture details are provided in the supplementary material.

\subsection{Surrogate Network for Lossy Codecs}
\label{sec: surrogate_network}
In practical applications, video rescaling is often accompanied by lossy video compression. To enable end-to-end training, prior works on spatial video rescaling~\cite{tian2021self, lu2024preprocessing, tian2023clsa} introduced a differentiable surrogate network that mimics compression artifacts. However, existing surrogate networks based on Dense3D-T blocks struggle to accurately simulate the distortions caused by inter-frame prediction in conventional video codecs~\cite{sullivan2012overview}.
To overcome this limitation, we propose a recurrent surrogate network \(\mathcal{G}\), which mimics traditional codec behavior under low-delay configurations. 
As shown in Fig.~\ref{fig: surrogate_network}, \(\mathcal{G}\)  is composed of compression-aware invertible blocks and a feature collapse module, and it progressively degrades the original LFR video $\y$ into a compressed counterpart $\widehat{\y}$. 

At each timestep $t$, the network receives the previously degraded frame \(\widehat{\mathbf{y}}_{(t-1)}\) and the current original frame \(\mathbf{y}_{(t)}\). Both frames are first transformed into temporal frequency components using the Haar temporal wavelet transform~\cite{dong2023temporal}:
\begin{align}
    \mathbf{y}_h^{\text{init}} &= \mathbf{y}_{(t)} - \overleftarrow{\mathcal{W}}(\widehat{\mathbf{y}}_{(t-1)}, \mathbf{mv}_{t -1\to t}), \\
    \mathbf{y}_l^{\text{init}} &= \widehat{\mathbf{y}}_{(t-1)} + \overleftarrow{\mathcal{W}}(\mathbf{y}_{(t)}, \mathbf{mv}_{t \to t-1}),
\end{align}
where \(\mathbf{mv}_{t \to t+1}\) and \(\mathbf{mv}_{t+1 \to t}\) denote motion vectors derived from the HEVC bitstreams for backward warping, respectively.
The initial frequency components \((\mathbf{y}_{l}^\text{init}, \mathbf{y}_{h}^\text{init})\) are processed by three compression-aware invertible coupling layers, termed Q-Invertible Blocks, to produce low- and high-frequency components $(\mathbf{y}_{l}, \mathbf{y}_{h})$. These are then degraded to mimic codec artifacts: $\mathbf{y}_{l}$ is quantized into $\widehat{\mathbf{y}}_l$, and $\mathbf{y}_{h}$ is passed through a feature collapse module to obtain a distorted residual $\widehat{\mathbf{y}}_h$, simulating lossy residual encoding.  The degraded components \(\widehat{\mathbf{y}}_l\) and \(\widehat{\mathbf{y}}_h\) are then passed through the inverse Q-Invertible Blocks and the Haar Temporal Wavelet Transform to generate the compressed frame \(\widehat{\mathbf{y}}_{(t)}\). 

As depicted in Fig. \ref{fig: surrogate_network} (b), the Q-Invertible Block follows the invertible coupling design of MIMO-VRN, with a key difference: the incorporation of the quantization parameter (QP) into the transformation functions \(\varphi'(\cdot)\), \(\phi'(\cdot)\), \(\rho'(\cdot)\), and \(\eta'(\cdot)\).
Specifically, for each input frame, a two-layer MLP processes the QP to generate a scale vector \( \mathbf{s} \in \mathbb{R}^{C \times 1 \times 1} \), which modulates intermediate features \( \mathbf{f}_d \in \mathbb{R}^{C \times H \times W} \) via channel-wise multiplication, \textit{i.e.}, \( \mathbf{f}_d' = \mathbf{f}_d \odot \mathbf{s} \). This channel-wise attention mechanism enables each Q-Invertible Block to dynamically adapt to different compression levels.

\subsection{Asymmetric Rescaling with Compression-Aware Low-Frame-Rate Video Enhancement}
\label{sec: self_asymmetric}

Since symmetric rescaling with invertible architectures is sensitive to distribution shifts caused by varying compression levels \cite{yang2023self}, we extend our framework to an asymmetric design by incorporating video quality enhancement (VQE) models. While applying a VQE model before upscaling may seem like a simple solution to remove compression artifacts, it proves inefficient. As shown in Fig.~\ref{fig: restoration_locations}, training a single model across multiple quantization parameters (QPs) significantly underperforms compared to QP-specific models. Moreover, the performance improvement of QP-specific models diminishes as the QP decreases. This highlights two key challenges for integrating VQE models into temporal rescaling: 1) improving performance at higher-bitrate compression, and 2) enhancing model generalization across different QPs.

\begin{figure}[t]
  \centering
    \includegraphics[width=1.01\columnwidth]{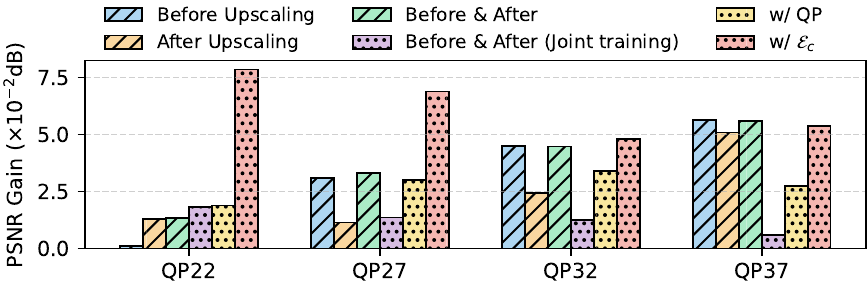}
  \caption{
  \textbf{Reconstruction performance gain using different strategies for integrating enhancement modules into our framework}, compared to the symmetric rescaling network.
The striped box denotes separate models trained for each QP,
while the spotted box indicates a single model trained across all four QPs. For a clear comparison of PSNR gains at different compression levels, the pre-trained MIMO-TWT and MIMO-VRN are frozen to keep the bitrate unchanged.
  }
  \label{fig: restoration_locations}
\end{figure}

For the first challenge, applying the VQE model before upscaling improves performance at lower bitrates (\textit{e.g.}, QP 27, 32, 37), as the enhanced LFR frames provide better input for upscaling. However, at higher bitrates (\textit{e.g.}, QP 22), VQE models struggle to distinguish useful high-frequency details from compression artifacts. As shown in Fig.~\ref{fig: freq_LFR}, downscaled LFR frames retain more high-frequency content, which is better preserved at higher bitrates.
We also observe that upscaling helps decouple meaningful details from artifacts, simplifying the enhancement task.
\textit{Based on this insight, we treat upscaling as a decoupling step and apply VQE afterward.}
This improves reconstruction at QP 22, as demonstrated by the comparison between “Before Upscaling” and “After Upscaling” in Fig.~\ref{fig: restoration_locations}.

Regarding the second challenge, we find that training a single model across compression levels proves less effective than training QP-specific models, as shown by the comparison between ``Before\&After'' and ``Before\&After (Joint Training)'' in Fig.\ref{fig: restoration_locations}. Furthermore, the optimal placement of enhancement modules—before or after upscaling—varies with QP. These findings motivate the use of compression-aware guidance. \textit{We address this by leveraging compression-aware features learned through a pairwise learning-to-rank strategy to adaptively control the enhancement process.} This approach outperforms directly incorporating QPs, as demonstrated by the performance gap between ``w/ $\mathcal{E}_c$'' and ``w/ QP'' in Fig.~\ref{fig: restoration_locations}.

\begin{figure}[t]
  \centering
  \includegraphics[width=1.0\columnwidth]{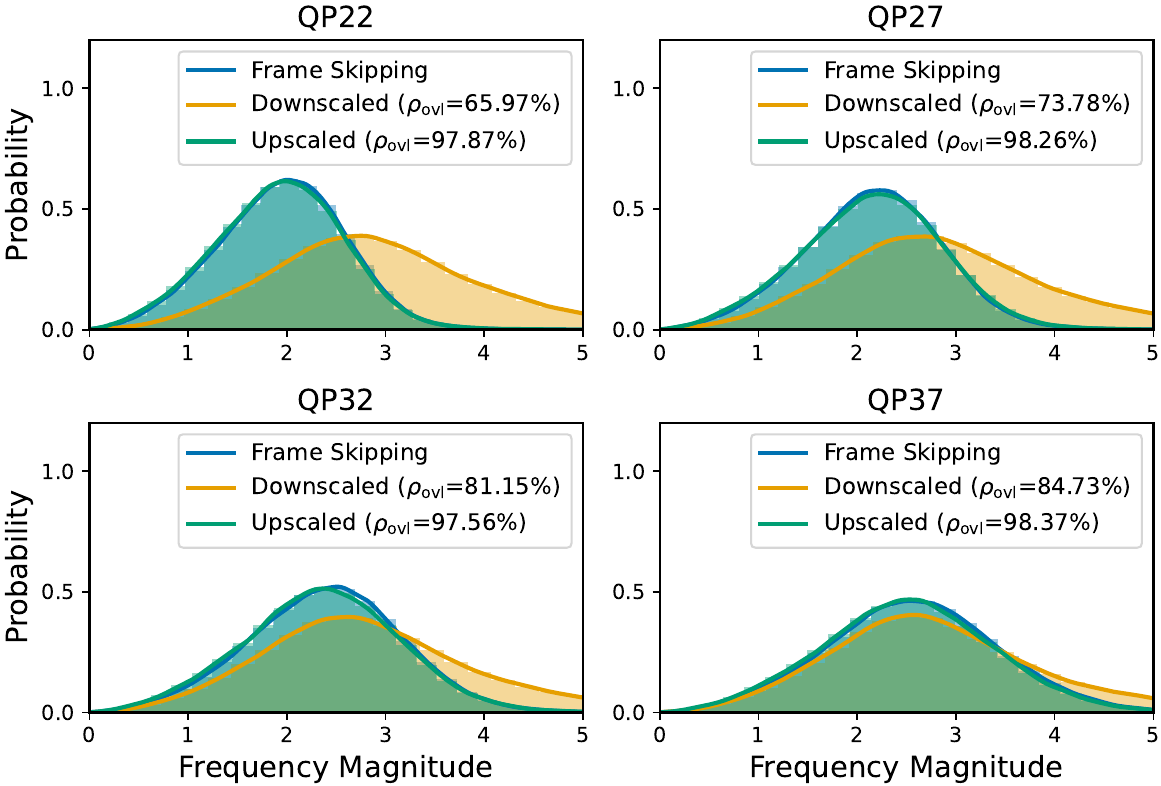}
\caption{\textbf{Frequency Analysis of Compressed LFR Frames}.
We visualize the log-magnitude spectrum of LFR frames produced by frame skipping, downscaling, and downscaling-then-upscaling on the SNU-FILM dataset. The overlap ratio  \(\rho_\mathrm{ovl}\)  quantifies spectral similarity to frame skipping. Results indicate that as QP decreases, downscaled frames exhibit increasingly rich high-frequency content, complicating subsequent enhancement. Computation details are provided in the supplementary material.}
  \label{fig: freq_LFR}
\end{figure}

 \begin{figure}[t]
  \centering
  \includegraphics[width=0.95\columnwidth]{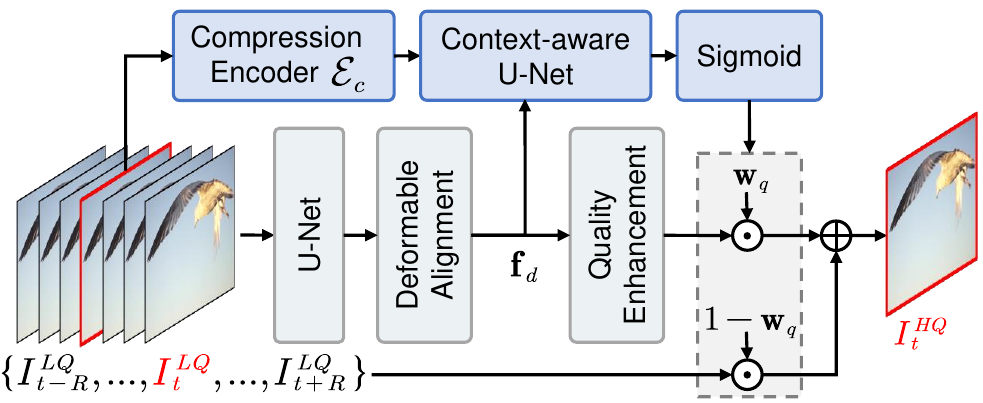}
  \caption{\textbf{Structure of the LFR restoration module.} 
  Built on the classical video enhancement method STDR~\cite{luo2022spatio}, we integrate compression-aware features learned via a learning-to-rank strategy to adaptively control enhancement strength through a residual shortcut.
  }
  \label{fig:structure_STDR}
\end{figure}

We adopt the classical video quality enhancement method STDR~\cite{luo2022spatio} both before and after upscaling to suppress compression artifacts. As shown in Fig.~\ref{fig:structure_STDR}, STDR aligns $2R$ neighboring frames with the target frame and performs enhancement in three stages: multi-level feature extraction with a U-Net backbone, deformable alignment, and quality refinement.  
To further guide enhancement, a pretrained compression encoder $\mathcal{E}_c$ extracts multi-level compression-aware features from the low-quality input frame $I_t^{LQ}$. These are fused with the aligned features $\mathbf{f_d}$ via a context-aware U-Net to produce a spatially adaptive weight map $\mathbf{w}_q$, as follows:
\begin{equation}
    \mathbf{w}_q = \text{Sigmoid}\left(\text{U-Net}\left(\mathcal{E}_c(I_t^{LQ}), \mathbf{f_d}\right)\right).
\end{equation}
The normalized map $\mathbf{w}_q$ adaptively controls the enhancement strength through a residual shortcut:
\begin{equation}
    I_t^{HQ} = \mathbf{w}_q \cdot \mathbf{f}_q  + (1 - \mathbf{w}_q) \cdot I_t^{LQ}.
\end{equation}

To enable the compression encoder $\mathcal{E}_c$ to capture subtle differences across varying compression levels, we adopt a learning-to-rank paradigm~\cite{wang2023compression}. Specifically, two frames encoded with two different QPs, $Q(i)$ and $Q(j)$, are passed through a shared compression encoder to extract a pair of compression representations. These are then fed into a shared ranker to produce two corresponding scores,  $\mathbf{s}_i$ and $\mathbf{s}_j$. The model is trained using a pairwise margin-ranking loss:
\begin{align}
    \mathcal{L}_{rank} = \max\left(0, \left(\mathbf{s}_i - \mathbf{s}_j\right) \cdot \kappa + \xi\right), \notag \\
\text{where} \quad \kappa = \begin{cases} 
1 & \text{if } Q(i) < Q(j) \\
-1 & \text{if } Q(i) > Q(j).
\end{cases} 
\label{equ:L_ranker}
\end{align}
Here, \( \kappa \) indicates the ground-truth preference between the pair based on their QP values, and \( \xi \) denotes the ranking margin, which is empirically set to 0.5. Details of the compression encoder and ranker are provided in the supplementary material.

\subsection{{Alternate} Training Strategy}
\label{sec: training_strategy}
To enable stable optimization of the temporal video rescaling network and the surrogate network, we adopt an alternate training strategy inspired by~\cite{ye2025rate}, as outlined in Algorithm~\ref{algorithm-training_strategy}. 
During each training iteration, the forward propagation process begins by compressing the downscaled video \(\mathbf{y}\) using the HEVC codec \(\mathcal{H}\), yielding both the compressed video \(\widetilde{\mathbf{y}}\) and associated metadata. These are fed into the surrogate network \(\mathcal{G}\), which simulates compression distortions and generates a disturbed LFR video \(\widehat{\mathbf{y}}\). As indicated in Line 5 of Algorithm~\ref{algorithm-training_strategy}, the values of \(\widehat{\mathbf{y}}\) are \textit{reassigned} to match those of \(\widetilde{\mathbf{y}}\).
Next, the reassigned \(\widehat{\mathbf{y}}\) is passed through the high-frequency reconstruction module \(\mathcal{F}_{\mathbf{z}}\) to recover the high-frequency component \(\widehat{\mathbf{z}}\). Both \(\widehat{\mathbf{y}}\) and \(\widehat{\mathbf{z}}\) are then input into the inversed invertible neural network \(\mathcal{F}_{\sharp}^{-1}\) to reconstruct the high-frame-rate video. 
For backward propagation, we adopt an alternate update strategy: the surrogate network is frozen while updating the temporal rescaling network using the loss defined in Eq. (\ref{equ: total}); conversely, the rescaling network is frozen while updating the surrogate network with the loss in Eq. (\ref{equ: surrogate_network}).

\begin{algorithm}[t]
\caption{Alternate Training Strategy}
\label{algorithm-training_strategy}
\textbf{Input: } Training dataset \( D = \{\mathbf{x}_k\}_{k=1}^m \); learning rates \( \eta_1, \eta_2 \) \\
\textbf{Output: }Trained MIMO-TWT and MIMO-VRN with parameters \( \mathbf{\Theta}_{\sharp} \), HF reconstruction module with parameters \( \mathbf{\Theta}_{z} \), and LF enhancement module with parameters \( \mathbf{\Theta}_{r} \) \\
\textbf{Training Variables:} \( \mathbf{\Theta}_{\sharp}, \mathbf{\Theta}_{z}, 
\mathbf{\Theta}_{r},
\mathbf{\Theta}_{c} \)
\begin{algorithmic}[1]
\For{each \( \mathbf{x} \in D \)}
    \State \( \mathbf{y}, \mathbf{z} \leftarrow \mathcal{F}_{\sharp}(\mathbf{x} ; \mathbf{\Theta}_{\sharp}) \) 
    \Comment{Downscaling}
    \State \( \widetilde{\mathbf{y}}, \text{Metadata} \leftarrow \mathcal{H}(\mathbf{y}) \) 
    \Comment{Lossy compression}
    \State \( \widehat{\mathbf{y}} \leftarrow \mathcal{G}(\mathbf{y}, \text{Metadata} ; \mathbf{\Theta}_{c}) \) 
    \Comment{Compression simulation}
    \State \( \widehat{\mathbf{y}}.\text{data} \leftarrow \widetilde{\mathbf{y}}.\text{data} \) 
    \State $\widehat{\mathbf{z}} \leftarrow \mathcal{F}_\mathbf{z}(\widehat{\y}; \mathbf{\Theta}_{z})$
    \State \( \widehat{\mathbf{x}} \leftarrow \mathcal{F}_{\sharp}^{-1}(\widehat{\mathbf{y}}, \widehat{\mathbf{z}} ; \mathbf{\Theta}_{\sharp},  \mathbf{\Theta}_{r}) \) 
    \Comment{Upscaling}
    \State \( \mathcal{L}_{\text{basic}} \leftarrow \ell_1(\mathbf{x}, \widehat{\mathbf{x}}) + \lambda \cdot \ell_2(\mathbf{y}, \mathbf{x}_e) \) \Comment{Eq.~(\ref{equ: total})}
    \State \( [\mathbf{\Theta}_{\sharp}, \mathbf{\Theta}_{z},
    \mathbf{\Theta}_{r}] \leftarrow [\mathbf{\Theta}_{\sharp}, \mathbf{\Theta}_{z},
    \mathbf{\Theta}_{r}] - \eta_1 \cdot \nabla_{\mathbf{\Theta}_{\sharp}, \mathbf{\Theta}_{z},
    \mathbf{\Theta}_{r}} \mathcal{L}_{\text{basic}} \) 
    \State /* Clear computational graph */        
        \State \( \mathcal{L}_{\text{s}} \leftarrow \ell_1\left(\mathcal{G}(\mathbf{y}, \text{Metadata} ; \mathbf{\Theta}_{c}), \widetilde{\mathbf{y}}\right) \) \Comment{Eq.~(\ref{equ: surrogate_network})}
    \State \( \mathbf{\Theta}_{c} \leftarrow  \mathbf{\Theta}_{c} - \eta_2 \cdot  \nabla_{\mathbf{\Theta}_c} \mathcal{L}_{\text{s}}  \) 
\EndFor
\State \Return \( \mathbf{\Theta}_{\sharp}, \mathbf{\Theta}_{z} ,
    \mathbf{\Theta}_{r}\)
\end{algorithmic}
\label{alg:interactive_training}
\end{algorithm}

\section{Experiments}

\label{exp}
\subsection{Dataset and Metrics}
\label{dataset}
\textbf{(1) Training Dataset.}
We utilize the widely adopted Vimeo90K-Septuplet dataset~\cite{xue2019video}, consisting of 65K 7-frame video clips with various motion types. The training set is compressed using H.265 with QP values of 17, 22, 27, 32, and 37 to generate the compressed sequences, which are used for pretraining the ranker and fine-tuning the baseline methods.

\textbf{(2) Testing Dataset.}
For the evaluation of the temporal video rescaling task, we use three standard benchmarks: UCF-101~\cite{soomro2012ucf101}, Vimeo90K test~\cite{xue2019video}, and SNU-FILM Medium~\cite{choi2020scene}, covering various levels of motion complexity. We extend these datasets from 3-frame sequences to full-length sequences to evaluate performance on longer sequences.

\textbf{(3) Metrics.}
We evaluate the quality of reconstructed HFR video using Peak Signal-to-Noise Ratio (PSNR) and Structural Similarity Index (SSIM)~\cite{wang2004image} in the RGB space, and use bits-per-pixel (bpp) to measure the average bits required for storing or transmitting each pixel of HFR videos. 
{
To assess perceptual and temporal quality, we further include the following metrics: VMAF~\cite{li2016toward} and LPIPS~\cite{zhang2018unreasonable} for perceptual fidelity, and tOF~\cite{chu2020learning} and warping error~\cite{lai2018learning} for temporal consistency.
Specifically, tOF quantifies motion inconsistency by comparing the optical flow fields between two consecutive frames in the reconstructed sequence $\{b_{t-1}, b_t\}$ and the ground truth $\{g_{t-1}, g_t\}$:
\begin{equation}
\text{tOF} = \left\| \text{OF}(b_{t-1}, b_t) - \text{OF}(g_{t-1}, g_t) \right\|_1,
\end{equation}
where $\text{OF}(\cdot,\cdot)$ denotes optical flow estimated using the RAFT network~\cite{teed2020raft}.
The warping error is computed by first estimating optical flow with RAFT, then warping the previous frame to align with the current frame, and finally measuring the mean squared error (MSE). The result is reported in PSNR and denoted as $\text{PSNR}_{\text{warp}}$. Together, tOF and $\text{PSNR}_{\text{warp}}$ provide complementary assessments of temporal coherence from motion and reconstruction perspectives.
}
Running time and FLOPs are measured by upscaling a $256\times448$ LFR video with four frames to an HFR video with seven frames on an NVIDIA GeForce RTX 3090 GPU using PyTorch~1.13.1 for a fair comparison.

 \begin{table*}[t]
\caption{
{Quantitative evaluation of temporal upscaling efficiency and reconstruction fidelity compared to other methods on the UCF101~\cite{soomro2012ucf101}, Vimeo90K~\cite{xue2019video}, and SNU-FILM Medium~\cite{choi2020channel} benchmarks.} 
Here, $\sigma_\text{PSNR}$ denotes the standard deviation of per-frame PSNR values within each upscaled video. “OOM” stands for out-of-memory.
The measurements are made on an NVIDIA GeForce RTX 3090 GPU with PyTorch-1.13.1. {\textbf{Bold}}: best performance, {\underline{Underline}}: second best performance.
}
\centering
 \resizebox{\textwidth}{!}{
 \large
\begin{tabular}{{l l  c  c c c c c c c c c c }}
\toprule
\multirow{2}{*}{\textbf{Downscaler}} & \multirow{2}{*}{\textbf{Upscaler}}& \multicolumn{2}{c}{\textbf{Bitrate-Distortion}} & \multicolumn{2}{c}{\textbf{Upscaling Cost} $\downarrow$ }  & \multicolumn{3}{c}{\textbf{Reconstructed HFR PSNR/SSIM} $\uparrow$} & \multirow{2}{*}{\textbf{$\sigma_\text{PSNR}\downarrow$}}\\
\cmidrule(l){3-4} \cmidrule(l){5-6} \cmidrule(l){7-9} 
 & & {bpp} & {PSNR/SSIM} & {Time(s)} & {TFLOPs} & {UCF101} & {Vimeo90k} & {SNU-FILM} & &\\
 \midrule
 \multirow{9}{*}{\begin{tabular}[c]{@{}c@{}}\textbf{Frame}\\      \textbf{Skipping}\end{tabular}}
 & VFIformer {\small \textcolor{gray}{(CVPR'22)}} \cite{lu2022video} & 0.1652 &  28.40/0.8822   & 6.940 & 1.77 & 28.30/0.8834 & 28.19/0.8719 & 28.73/0.8912 & 5.2513 \\
   & XVFI$_v$ {\small \textcolor{gray}{(ICCV'21)}} \cite{sim2021xvfi} & 0.1652 & 33.96/0.9220 & 1.270 & 0.48 &  31.84/0.9108 & 34.70/0.9194  & 35.35/0.9359 & 1.2722 \\ 
& MoMo {\small \textcolor{gray}{(AAAI'25)}} \cite{lew2025disentangled} & 0.1652    & 34.19/0.9224 & 1.782 & 2.34  & 31.70/0.9104 & 34.93/0.9205 & 35.63/0.9362 & 1.0457  \\
 & EBME-H {\small \textcolor{gray}{(WACV'23)}} \cite{jin2022enhanced} & 0.1652  & 34.22/0.9236 & {\textbf{0.089}} & \textbf{0.06} & 31.96/0.9121 & 34.98/0.9216 & 35.71/0.9375 & 1.0144  \\
& SGM-VFI {\small \textcolor{gray}{(CVPR'24)}} \cite{liu2024sparse} & 0.1652   &  34.26/0.9240 &  1.848 & {1.07} & 32.13/0.9129  & 34.99/0.9215 & 35.67/0.9376 &0.9989  \\
& UPR-Net L {\small \textcolor{gray}{(IJCV'25)}} \cite{jin2025upr} & 0.1652 &  34.30/0.9241   & 0.739 & 0.58 &  32.05/0.9127 & 35.08/0.9220 & 35.76/0.9376  &0.9629  \\
 & IFRNet {\small \textcolor{gray}{(CVPR'22)}} 
 \cite{kong2022ifrnet} &  0.1652 &  34.31/{0.9243}  &  0.423 & 0.30 & {32.09}/{0.9130} & 35.11/0.9222 & 35.74/0.9376 &0.9623  \\
& EMA-VFI {\small \textcolor{gray}{(CVPR'23)}} 
 \cite{zhang2023extracting} & 0.1652    & 34.33/0.9243 & 0.982 & 1.21 & 32.07/0.9129 & 35.12/0.9224 &  35.79/0.9378 &0.9407  \\
& GIMM-VFI {\small \textcolor{gray}{(NIPS'24)}} \cite{guo2024generalizable} & 0.1652 & {34.40}/{0.9251} & {0.867} & 3.21 & {32.16}/{0.9141} & {35.23}/{0.9231} & {35.82}/{0.9381} &0.9229  \\
\midrule
\multirow{5}{*}{\begin{tabular}[c]{@{}c@{}}\textbf{Learnable}\\\textbf{Network}\end{tabular}}
&{STAA {\small \textcolor{gray}{(ECCV'22)}} \cite{xiang2022learning}}& 0.1675 & 34.27/0.9210 & 1.541 & 1.63 & 32.04/0.9092& 35.05/0.9188 & 35.71/0.9349 & 0.9887\\
&CSTVR {\small \textcolor{gray}{(CVPR'25)}} \cite{zhang2025continuous} & - & - & 1.226  & 1.37 &29.45/0.8783 & 31.00/0.8766 & OOM & 3.2115\\ 
&CSTVR{*} {\small \textcolor{gray}{(CVPR'25)}} \cite{zhang2025continuous} & - & - & 1.226 & 1.37 & 31.96/0.9121 & 34.40/0.9233 & OOM & 1.6802\\
&{TVRN-S} & {0.1654} & {\underline{34.65}}/{\underline{0.9280}}  & {\underline{0.310}} & \underline{0.21}  &  {\underline{32.32}/\underline{0.9173}} & {\underline{35.55}/\underline{0.9259}} & {\underline{36.09}/\underline{0.9409}} & \underline{0.8798}  \\ 

&\cellcolor[HTML]{FFEEED}{TVRN} & \cellcolor[HTML]{FFEEED}{\textbf{0.1650}} & \cellcolor[HTML]{FFEEED}\textbf{34.73}/\textbf{0.9307}  & \cellcolor[HTML]{FFEEED}{{0.911}} & \cellcolor[HTML]{FFEEED}{0.97}  & \cellcolor[HTML]{FFEEED}\textbf{32.37}/\textbf{0.9212}  & \cellcolor[HTML]{FFEEED}\textbf{35.65}/\textbf{0.9288}  & \cellcolor[HTML]{FFEEED}{\textbf{36.16}/\textbf{0.9422}} & \cellcolor[HTML]{FFEEED}\textbf{0.8488}  \\ 
\bottomrule
\end{tabular}
}
\label{tab:qualitative_comparison}
   \begin{tablenotes}
      \footnotesize
      \item[1] {*} indicates that the model is retrained using our proposed surrogate network as a gradient estimator.
    \end{tablenotes}
\end{table*}
\begin{figure*}[!t]
\setlength{\tabcolsep}{0.5pt}
\centering

\resizebox{\textwidth}{!}{
\scriptsize
\begin{tabular}{cccccccc>{\columncolor[HTML]{FFEEED}}c}
Ground Truth & UPR-Net L \cite{jin2025upr} & IFRNet \cite{kong2022ifrnet} & EMA-VFI \cite{zhang2023extracting} & GIMM-VFI \cite{guo2024generalizable} & STAA \cite{xiang2022learning} & CSTVR \cite{zhang2025continuous} & CSTVR* \cite{zhang2025continuous} & \textbf{TVRN (Ours)} \\
\includegraphics[width=0.09\textwidth]{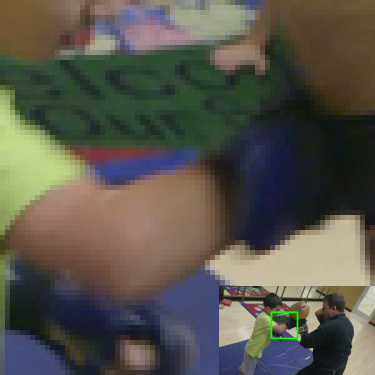}
&\includegraphics[width=0.09\textwidth]{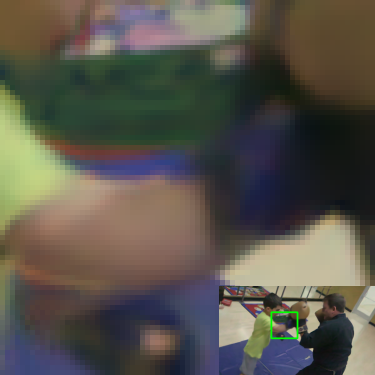}
&\includegraphics[width=0.09\textwidth]{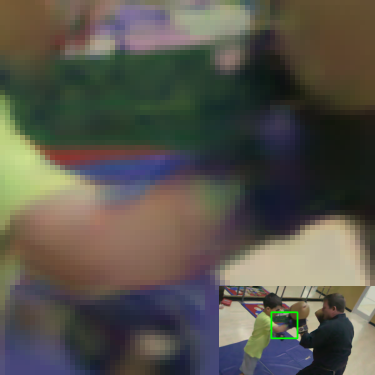}
&{\includegraphics[width=0.09\textwidth]{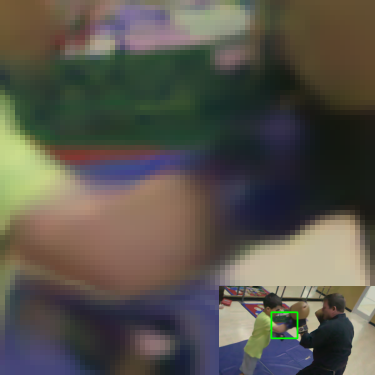}}
&\includegraphics[width=0.09\textwidth]{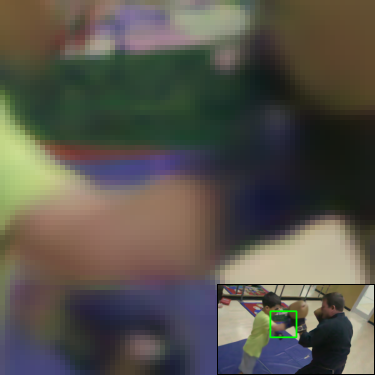}
&\includegraphics[width=0.09\textwidth]{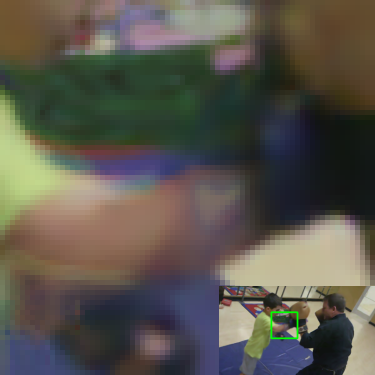}
&\includegraphics[width=0.09\textwidth]{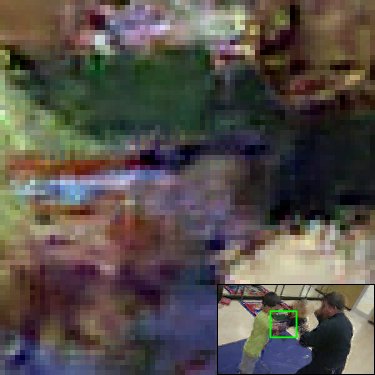}
&\includegraphics[width=0.09\textwidth]{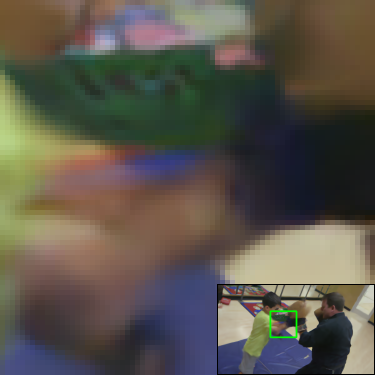}
&\includegraphics[width=0.09\textwidth]{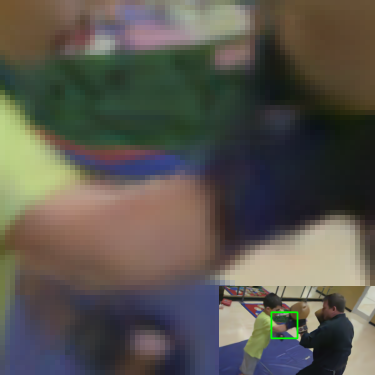}\\
BPP/PSNR &  0.1866/34.99 & 0.1866/34.79 & \secondbest{0.1866}/\secondbest{35.05} &  0.1866/34.98 & 0.2126/34.76 & 0.1847/27.72 & 0.2081/33.86 & \textbf{0.1741}/\textbf{35.28}
\end{tabular}}

\resizebox{\textwidth}{!}{
\scriptsize
\begin{tabular}{cccccccc>{\columncolor[HTML]{FFEEED}}c}
\includegraphics[width=0.09\textwidth]{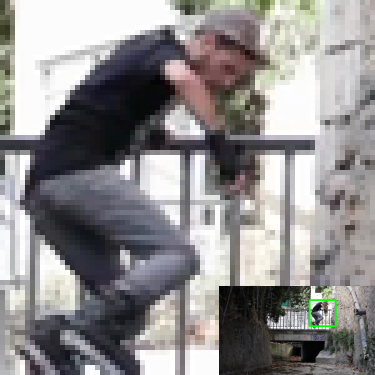}
&\includegraphics[width=0.09\textwidth]{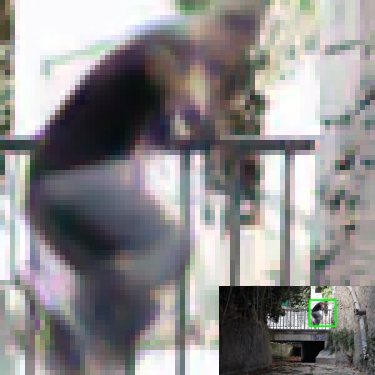}
&\includegraphics[width=0.09\textwidth]{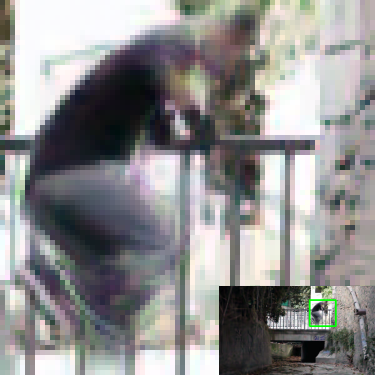}
&\includegraphics[width=0.09\textwidth]{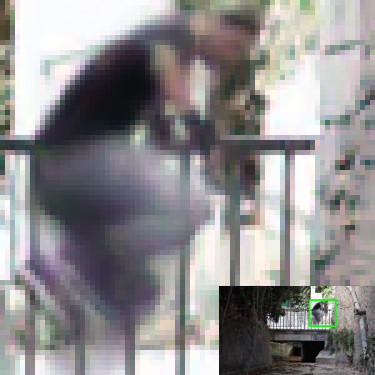}
&\includegraphics[width=0.09\textwidth]{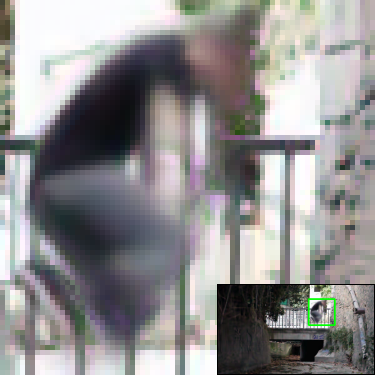}
&\includegraphics[width=0.09\textwidth]{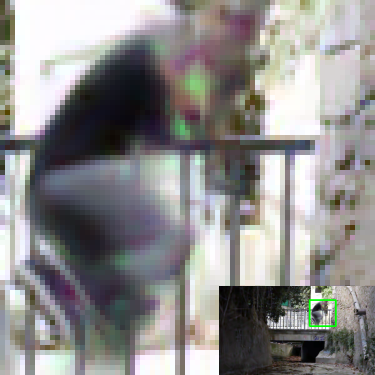}
&\includegraphics[width=0.09\textwidth]{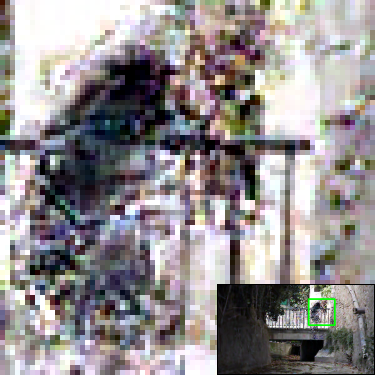}
&\includegraphics[width=0.09\textwidth]{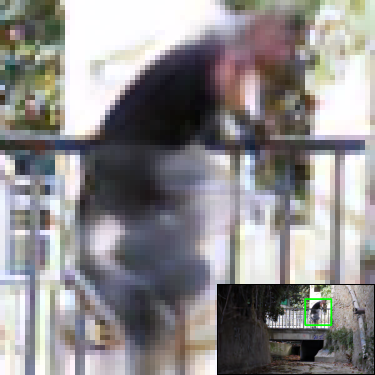}
&\includegraphics[width=0.09\textwidth]{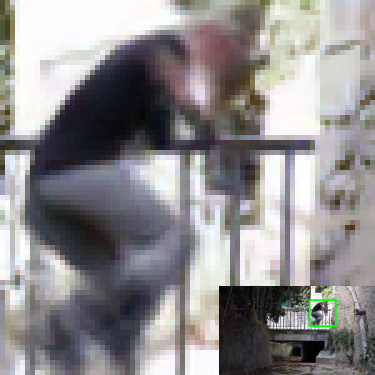}\\
BPP/PSNR  & \secondbest{0.3185}/\secondbest{32.39} & 0.3185/31.05 &  0.3185/30.49 & 0.3185/31.43 & {0.3208}/31.90 & 0.3602/27.65 & 0.3496/30.84 & \textbf{0.3185}/\textbf{32.89} \\
\end{tabular}}

\resizebox{\textwidth}{!}{
\scriptsize
\begin{tabular}{cccccccc>{\columncolor[HTML]{FFEEED}}c}
\includegraphics[width=0.09\textwidth]{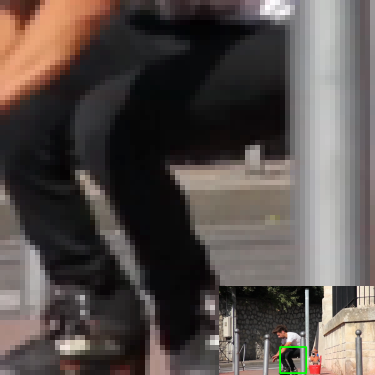}
&\includegraphics[width=0.09\textwidth]{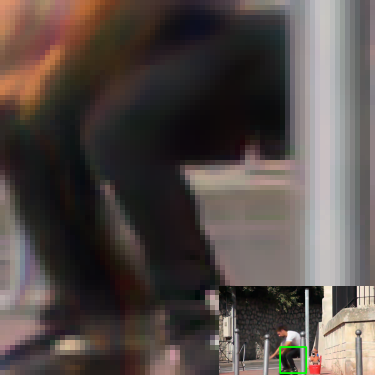}
&\includegraphics[width=0.09\textwidth]{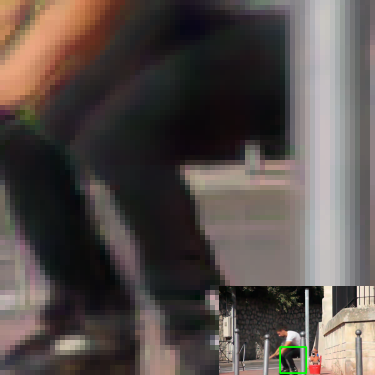}
&\includegraphics[width=0.09\textwidth]{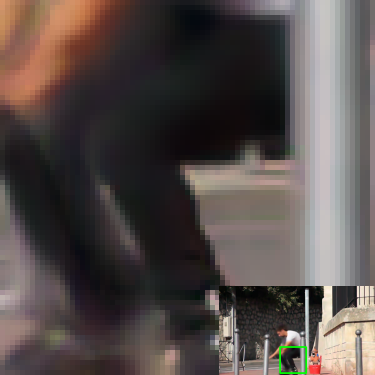}
&\includegraphics[width=0.09\textwidth]{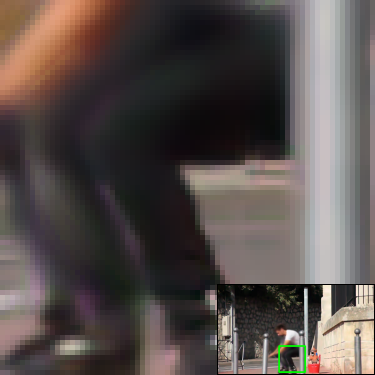}
&\includegraphics[width=0.09\textwidth]{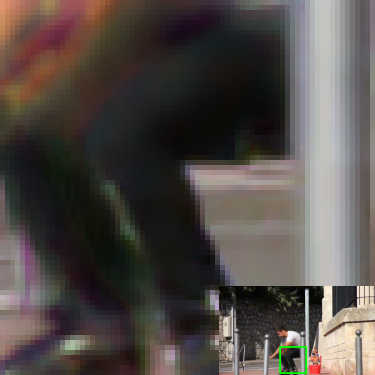}
&\includegraphics[width=0.09\textwidth]{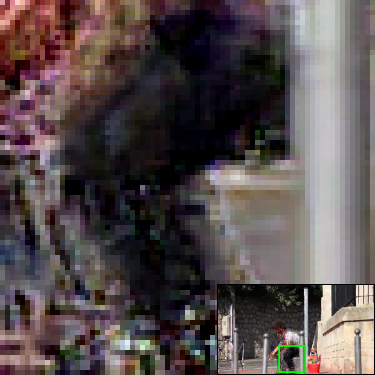}
&\includegraphics[width=0.09\textwidth]{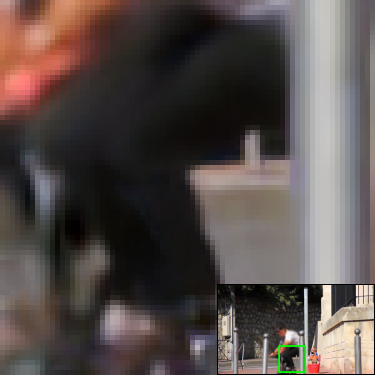}
&\includegraphics[width=0.09\textwidth]{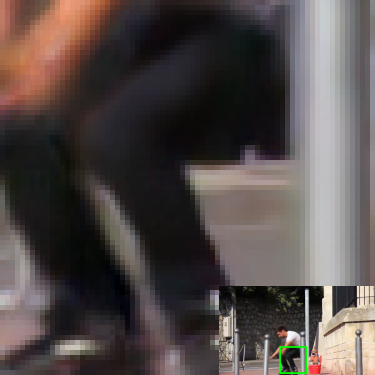}\\
BPP/PSNR & 0.3225/30.13 & 0.3225/30.47 & 0.3225/30.02 & 0.3225/\secondbest{30.51} & 0.3340/30.01 & 0.3282/26.26& \secondbest{0.3167}/30.40 & \textbf{0.3040}/\textbf{30.61} \\
\end{tabular}}

\resizebox{\textwidth}{!}{
\scriptsize
\begin{tabular}{cccccccc>{\columncolor[HTML]{FFEEED}}c}
\includegraphics[width=0.09\textwidth]{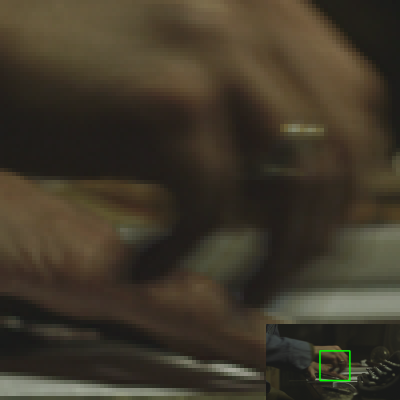}
&\includegraphics[width=0.09\textwidth]{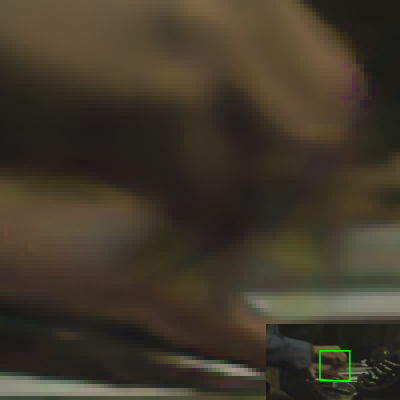}
&\includegraphics[width=0.09\textwidth]{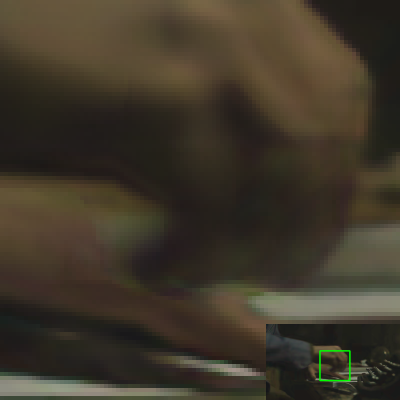}
&\includegraphics[width=0.09\textwidth]{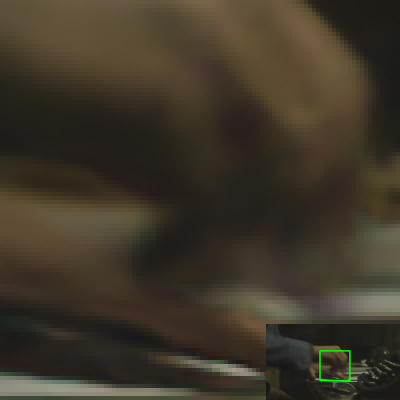}
&\includegraphics[width=0.09\textwidth]{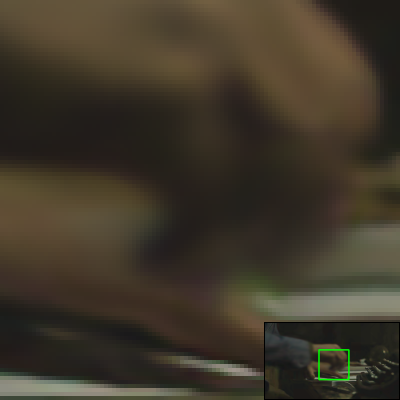}

&\includegraphics[width=0.09\textwidth]{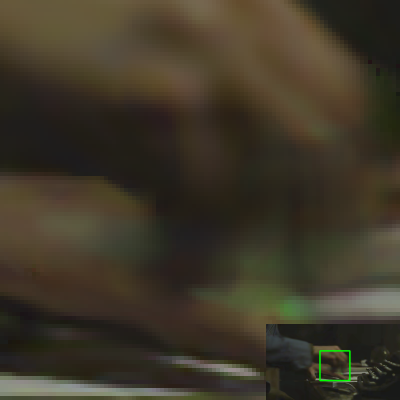}
&\includegraphics[width=0.09\textwidth]{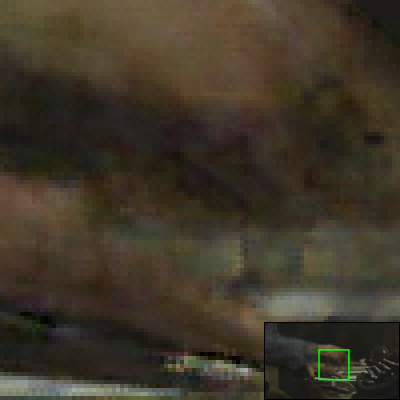}
&\includegraphics[width=0.09\textwidth]{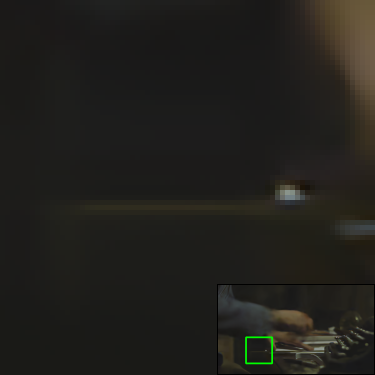}
&\includegraphics[width=0.09\textwidth]{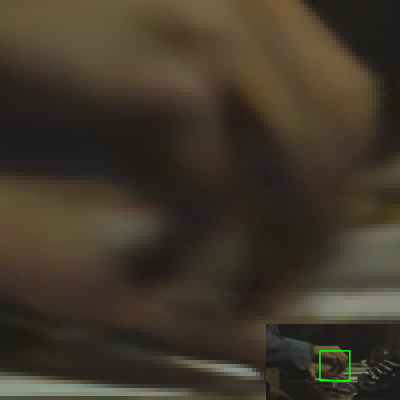}\\
BPP/PSNR &   {0.0771}/{38.61}  & \secondbest{0.0771}/\secondbest{39.03} & 0.0771/38.46 & 0.0771/38.61 & 0.0839/38.40 & 0.0835/34.15 & 0.0891/37.02 & \textbf{0.0764}/\textbf{39.79}
\label{fig: LFR_recon}
 \end{tabular}}


\caption{\textbf{Qualitative comparison between our method and competitive baselines on sequences \textit{00023\_0059}, \textit{00026\_0036}, \textit{00026\_0003}, and \textit{00025\_0034} from the Vimeo90K septuplet dataset}~\cite{xue2019video}.
To facilitate visual comparison, we crop and display regions with noticeable interpolation artifacts. The full frames are shown at the bottom right for context. 
{\textbf{Bold}}: best performance, {{\underline{Underline}}}: second best performance.
\textit{Best viewed by zooming in.}
}
\label{fig:comp_diff_urban100_and_div2k}
\end{figure*}

\subsection{Experimental Setup}
\label{setup}
\textbf{Network setup.}
We adopt the lightweight VFI network UPR-Net \cite{jin2025upr} as the prediction module in MIMO-TWT and utilize Dense3D blocks \cite{iandola2014densenet} as the update module. The surrogate network is employed only during training to estimate gradients of the x265 encoder. During inference, the actual H.265 encoder is used to compress LFR videos. 
For the first frame, a DenseBlock-based surrogate network \cite{tian2021self} directly estimates its compressed representation without temporal references.
We choose a group size of 7 instead of 3 for three reasons: (1) it aligns with the maximum sequence length in the Vimeo90k training set; (2) it represents a more challenging frame-rate rescaling ratio; and (3) it enables a fair comparison with the concurrent method CSTVR \cite{zhang2025continuous}, which adopts the same group size and frame-rate rescaling ratio.

\textbf{Training setup.}
The overall framework is trained in two stages. 
First, the compression encoder and ranker are pre-trained with a pairwise learning-to-rank loss defined in Eq. \ref{equ:L_ranker} for 100,000 iterations at a fixed learning rate of \(1 \times 10^{-4}\). Then, with the compression encoder frozen, we jointly train the invertible network and the surrogate network using an alternate strategy, optimizing Eq.~\ref{equ: total} and Eq.~\ref{equ: surrogate_network} for 50,000 iterations. 
The initial learning rate is set to $1\times10^{-4}$ and is halved every 10,000 iterations. The weight $\lambda$ for guidance loss is set to 10.
We adopt the Adam optimizer \cite{kingma2014adam} with $\beta_1 = 0.9$ and $\beta_2 = 0.999$. 
The entire framework is trained on 8 NVIDIA 1080Ti GPUs with a batch size of 8. In each iteration, the H.265 codec compresses the downscaled LFR videos and extracts metadata such as motion vectors and quantization parameters. The compressed videos and associated metadata are then used to train the surrogate network.
During training, QPs are randomly sampled from integers between 17 and 27 to enhance adaptation to various compression levels. Training the entire model from scratch takes approximately 90 hours.

\begin{table*}[tb]
\caption{BDBR (PSNR/SSIM) results of various methods against frame skipping with GIMM-VFI \cite{guo2024generalizable} on different codecs. 
The BDBR results of CSTVR$^{*}$ \cite{zhang2025continuous} are omitted on the SNU-FILM test set due to out-of-memory issues, and are also unavailable on VVC and AV1 due to insufficient PSNR/SSIM overlap with the anchor for reliable BDrate computation.}
\vspace{3pt}
\centering
\Large
\resizebox{\linewidth}{!}
{%
\begin{tabular}{|l|ccc|ccc|ccc|}
\hline
\multirow{2}{*}{\textbf{Upscaler}} & \multicolumn{3}{c|}{\textbf{HEVC}} & \multicolumn{3}{c|}{\textbf{VVC}} & \multicolumn{3}{c|}{\textbf{AV1}} \\
\cline{2-10}
& \textbf{UCF101} & \textbf{Vimeo90k} & \textbf{SNU-FILM} & \textbf{UCF101} & \textbf{Vimeo90k} & \textbf{SNU-FILM} & \textbf{UCF101} & \textbf{Vimeo90k} & \textbf{SNU-FILM} \\
\hline\hline
IFRNet \cite{kong2022ifrnet}& 8.80/1.78& 4.52/1.64& 2.60/1.28 
& 4.55/1.21 & 2.69/0.75 & -1.67/0.01 
& 10.32/1.32 & 6.04/1.27 & -1.00/0.10 \\

EMA-VFI \cite{zhang2023extracting} & 10.00/1.94& 4.13/1.32 & 1.22/1.00 
& 6.57/1.36 & 2.82/0.32 & 1.40/0.09 
& 15.96/3.70 & 7.27/1.70 & 8.18/1.09 \\

GIMM-VFI \cite{guo2024generalizable}& \textcolor{gray}{0.00/0.00} & \textcolor{gray}{0.00/0.00} & \textcolor{gray}{0.00/0.00} 
& \textcolor{gray}{0.00/0.00}& \textcolor{gray}{0.00/0.00}&\textcolor{gray}{0.00/0.00} 
& \textcolor{gray}{0.00/0.00}&\textcolor{gray}{0.00/0.00} &\textcolor{gray}{0.00/0.00} \\

GIMM-VFI+VQE \cite{guo2024generalizable}  
& -0.08/-0.05 & -0.13/-0.27 & -0.31/-0.46 
& -1.10/-0.67 & -0.45/-0.42 & -0.75/-0.80 
& -1.24/-0.85 & -0.76/-0.85 & -1.24/-1.30 \\

\hline

STAA \cite{xiang2022learning} 
& 12.26/11.41& 11.45/11.14& 4.10/9.48 
& 7.15/26.37 & 4.64/9.82 & 1.08/2.99 
& 19.66/93.81 & 11.30/24.78 & 11.38/7.68 \\

CSTVR{*} \cite{zhang2025continuous}
& 2.78/-8.56 & 21.94/-2.52&  -- 
& -- & -- & -- 
& - & -- & -- \\

\cellcolor[HTML]{FFEEED}TVRN 
& \cellcolor[HTML]{FFEEED}\textbf{-7.22/-16.40}& \cellcolor[HTML]{FFEEED}\textbf{-6.68/-3.88} & \cellcolor[HTML]{FFEEED}\textbf{-12.82/-10.62}  
& \cellcolor[HTML]{FFEEED}\textbf{-17.53/7.83} & \cellcolor[HTML]{FFEEED}\textbf{-15.56/-16.92} & \cellcolor[HTML]{FFEEED}\textbf{-24.02/-20.15} 
& \cellcolor[HTML]{FFEEED}\textbf{-35.76/-4.85} & \cellcolor[HTML]{FFEEED}\textbf{-34.94/-29.21} & \cellcolor[HTML]{FFEEED}\textbf{-42.61/-29.22} \\

\hline
\end{tabular}%
}
\label{tab:bdbr_performance}
\end{table*}

\begin{figure*}[!t]
\centering
\setlength{\tabcolsep}{0pt}
\renewcommand{\arraystretch}{1.0}
\scriptsize
\begin{tabular}{c}
\includegraphics[width=\textwidth]{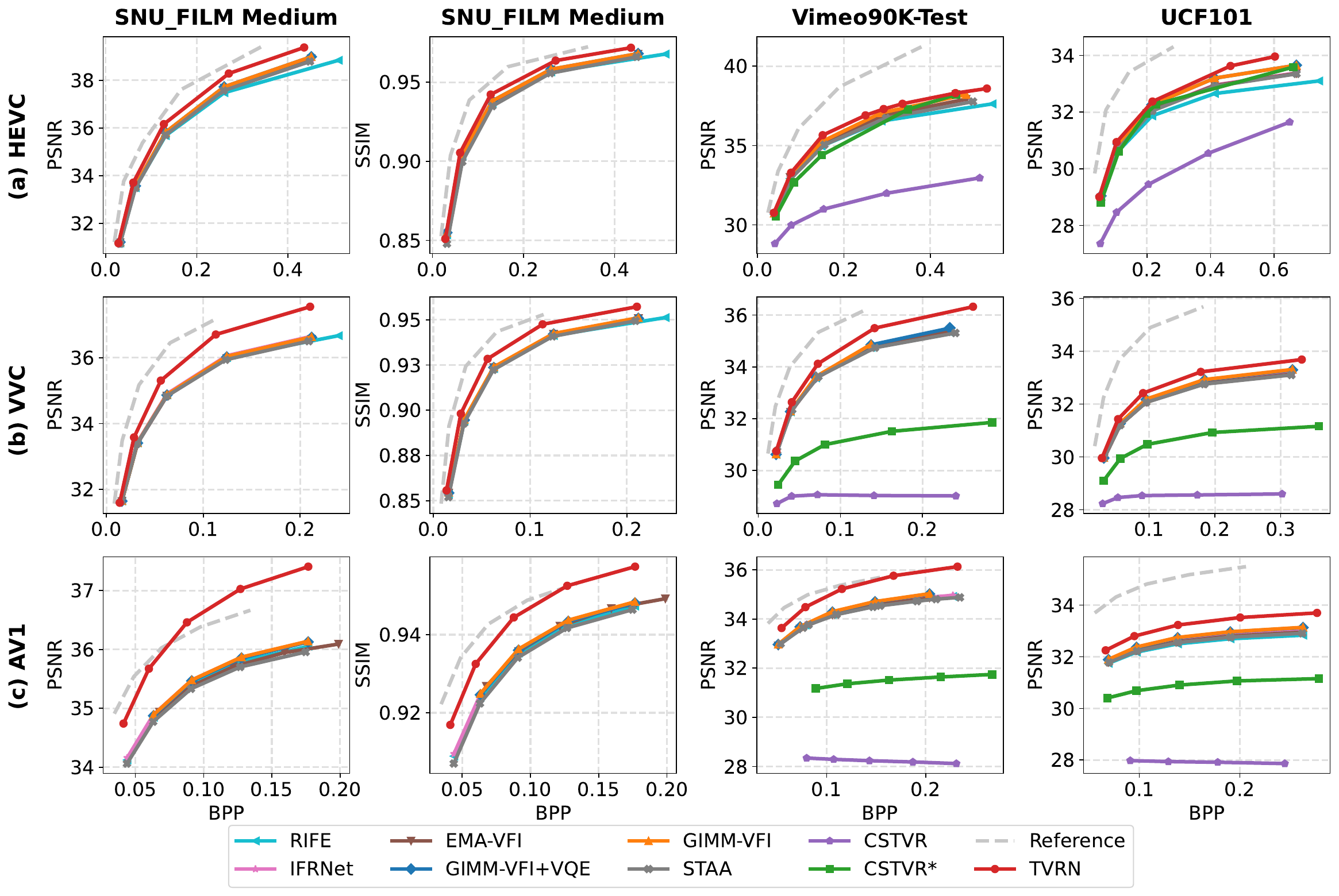} \\
\end{tabular}
\vspace{-10pt}
\caption{\textbf{Rate–Distortion Curves on test datasets}. 
For frame-skipping methods \cite{kong2022ifrnet,zhang2023extracting,guo2024generalizable}, fixed QPs of 18, 22, 27, 32, and 37 are used for HEVC and VVC, and 20, 26, 32, 38, and 44 for AV1. For learned frame-rate downscaling methods \cite{xiang2022learning,zhang2025continuous}, the QP is slightly increased to align bitrates.
}

\label{fig:rd_curve}
\end{figure*}

\begin{figure*}[!h]
  \centering
    \includegraphics[width=\linewidth]{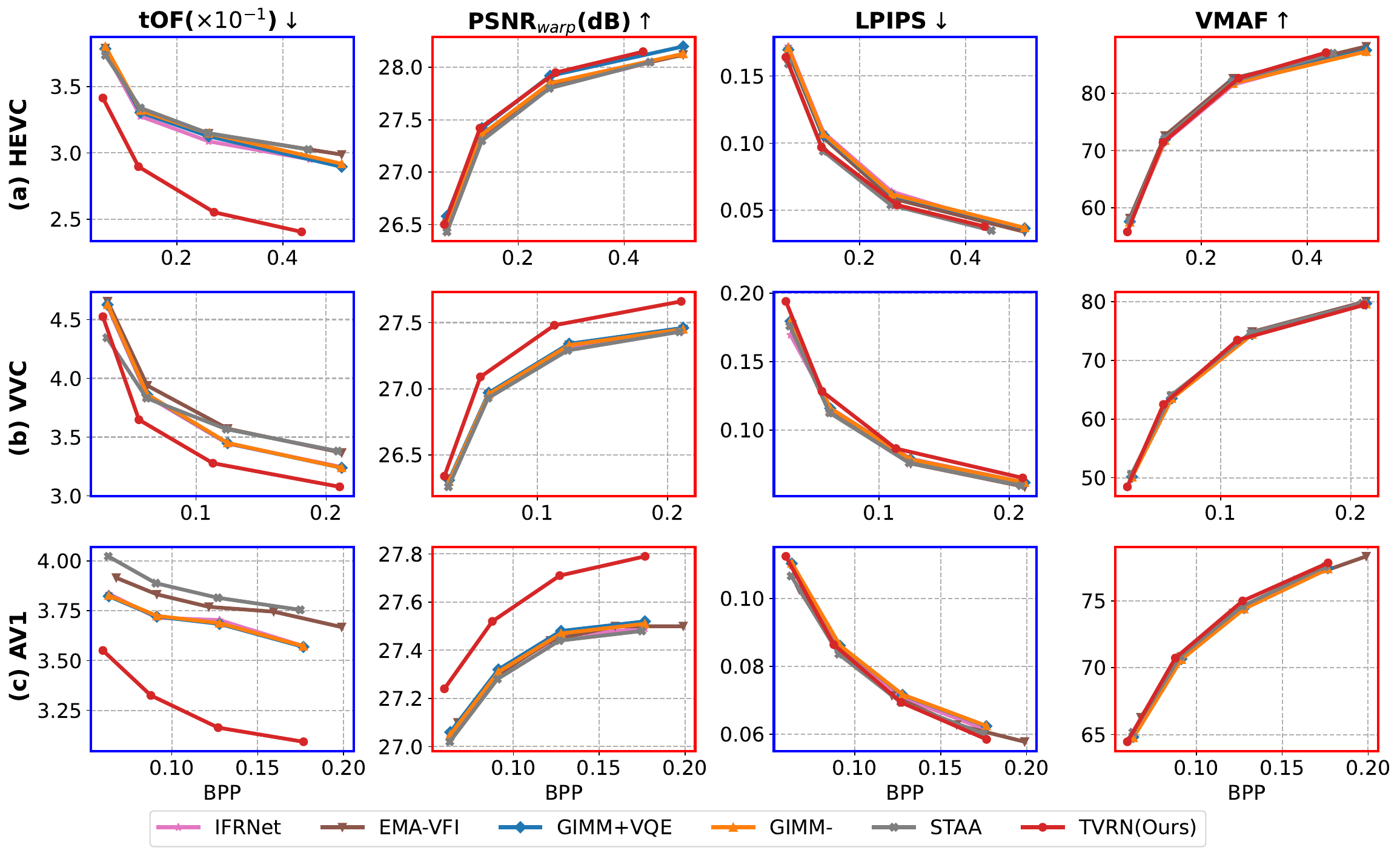}
    \vspace{-1em}
\caption{\textbf{Rate–Perception Curves on the SNU-FILM Medium testset \cite{choi2020channel}}. 
Temporal consistency is evaluated using tOF \cite{chu2020learning} and $\text{PSNR}_{\text{warp}}$ \cite{lai2018learning}, 
while perceptual quality is assessed by LPIPS \cite{zhang2018unreasonable} and VMAF \cite{li2016toward}. 
Red boxes denote metrics where higher values indicate better performance, 
whereas blue boxes denote metrics where lower values are preferred.}
    \label{fig:bd_vary_psnr_perceptual}
\end{figure*}

\textbf{Inference setup.}
Inference experiments are conducted on a server equipped with a single NVIDIA RTX 3090 GPU. The input video is divided into non-overlapping clips with a temporal length of 7, each of which is rescaled independently. For videos that cannot be evenly divided, the last frame is duplicated to pad the final clip.
Following the prior work CSTVR \cite{zhang2025continuous}, we set the frame-rate downscaling ratio to $4/7$ and further compress the downscaled video using the H.265 codec in \textit{zero-latency} mode \footnote{\texttt{ffmpeg -pix\_fmt yuv444p -s WxH -r 50 -i video.yuv -c:v libx265 -preset veryfast -tune zerolatency -x265-params "qp=QP"}}.
{
 We further extend the evaluation to include both AV1 \cite{han2021av1overview} and VVC \cite{bross2021vvcoverview}. 
Specifically, for AV1, we adopt the SVT-AV1 encoder with preset 4 in CQP mode
\footnote{\texttt{ffmpeg -pix\_fmt yuv444p -s WxH -f rawvideo -i video.yuv -c:v libsvtav1 -svtav1-params crf=0:qp=QP:chromafmt=444 -preset 4}}.
For VVC, we use the VVenC encoder ~\cite{VVenC} in fast preset under CQP mode \footnote{\texttt{ffmpeg -pix\_fmt yuv444p -s WxH -f rawvideo -i video.yuv -c:v libvvenc -qp QP -vvenc-params qpa=0 -preset fast -b:v 0 -maxrate 0 -bufsize 0}}. For frame-skipping-based baselines, the tested QPs are 20, 26, 32, 38, and 44 for AV1, and 18, 22, 27, 32, and 37 for HEVC and VVC. For learning-based temporal rescaling methods, we increase the QPs by 1 to better align the resulting bitrate with the baseline methods.
}

\subsection{Comparison to Baseline Methods}
We consider two categories of temporal video rescaling methods as our baselines:
(1) direct frame skipping followed by upscaling using various VFI models \cite{lu2022video, sim2021xvfi, lew2025disentangled, jin2022enhanced, liu2024sparse, jin2025upr, kong2022ifrnet, zhang2023extracting, guo2024generalizable}; and
(2) learnable frame-rate downscaling methods \cite{xiang2022learning, zhang2025continuous}.
Here, STAA \cite{xiang2022learning} performs downscaling using 3D convolutional layers and reconstructs LFR videos using space-time pixel-shuffle and deformable convolutional layers. The concurrent method CSTVR \cite{zhang2025continuous} also employs a \textit{partially} invertible architecture to embed motion information into LFR videos.
We fine-tune VFI models on LFR videos compressed by the H.265 codec to restore HFR videos for a fair comparison. To compare our RD performance against baselines, we keep the bpp on UCF101, Vimeo90K, and SNU-FILM Medium to be around 0.2158, 0.1480, and 0.1317, respectively, by adjusting the quantization parameters of the H.265 codec.

\begin{figure*}[t]
  \centering
  \includegraphics[width=\textwidth]{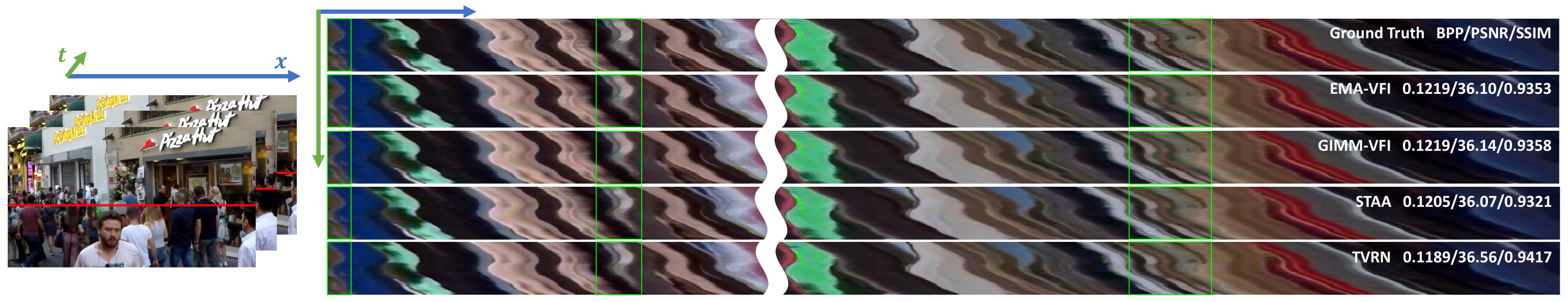}
  \caption{
\textbf{Comparison of temporal profiles in reconstructed high-frame-rate videos.} 
We examine a specific row across consecutive frames to assess temporal consistency. The red line indicates row 400 of the \textit{GOPR0384\_11\_05} sequence from the SNU-FILM test dataset. The visualization spans 50 frames starting from the 200th frame. Error-prone regions are marked with green boxes and arrows. 
\textit{Best viewed by zooming in.}
  }
  \label{fig:temporal_diagram}
\end{figure*}

\textbf{Upscaling efficiency and HFR fidelity.}
Table \ref{tab:qualitative_comparison} compares the upscaling efficiency and reconstruction fidelity of all methods at an average bitrate of 0.1650 bpp. The first group summarizes the methods based on frame skipping followed by VFI models, while the second group  is dedicated to temporal video rescaling methods, which include jointly learned downscaling networks.
Our proposed TVRN framework outperforms all other methods, achieving state-of-the-art performance: 32.37 dB on UCF101, 35.65 dB on Vimeo90K-Test, and 36.16 dB on SNU-FILM. Compared to the most recent VFI model, our method achieves improvements of +0.21 dB, +0.42 dB, and +0.34 dB on these three datasets, respectively. Furthermore, the lightweight variant TVRN-S, which omits multi-frame enhancement modules, reduces FLOPs to 6.54\% and upscaling time to 35.75\% of GIMM-VFI \cite{guo2024generalizable}, thanks to its multi-input multi-output parallel inference strategy.

We observe that state-of-the-art temporal video rescaling methods, such as STAA \cite{xiang2022learning} and CSTVR \cite{zhang2025continuous}, fall short of VFI-based approaches in terms of HFR reconstruction quality. This is mainly because high-frequency components embedded in downscaled LFR videos are often distorted by lossy compression, making them difficult for the upscaler to recognize and utilize during frame interpolation. The large performance gains achieved by our framework over prior learnable rescaling methods can be attributed to two key factors: (1) the invertible architecture that effectively regularizes information loss during downscaling, and (2) the surrogate network that accurately mimics the distortion introduced by non-differentiable video codecs. 
To surgically study the contribution of the surrogate network, we adopt the same architecture as CSTVR, replace its gradient estimator with our proposed surrogate network, and retrain the model from scratch under identical configurations
\footnote{The original pretrained weights of CSTVR are available at \url{https://gitee.com/zhanghahaxixi/cvrs}. The retrained CSTVR* weights are available on our project homepage: \url{https://github.com/fengxinmin/TVRN_public/releases}.}.
As shown in Table \ref{tab:qualitative_comparison}, the resulting model CSTVR* outperforms the original CSTVR by a large margin, achieving 3.40 dB higher PSNR (\textbf{34.40} vs. 31.00) and a 0.0493 improvement in SSIM (\textbf{0.9259} vs. 0.8766) on the Vimeo90k test set, strongly validating the effectiveness of our surrogate network. 
Note that while CSTVR* surpasses VFI-based methods in SSIM, it still lags in PSNR. This is primarily due to the downscaling and upscaling processes are not fully reversible and thus cannot effectively regularize information loss.

We present a visual comparison of the reconstructed HFR videos in Fig. \ref{fig:comp_diff_urban100_and_div2k}. From left to right, we show the ground truth and the results of VFI-based methods \cite{jin2022enhanced, jin2025upr, kong2022ifrnet, zhang2023extracting, guo2024generalizable}, learnable temporal video rescaling methods \cite{xiang2022learning, zhang2025continuous}, and our framework. Compared to other approaches, our framework  reconstructs object positions and recovers finer texture details more accurately.
For instance, in the second row, only our method successfully recovers the position of the human foot, while other models exhibit severe artifacts or color shifts. In the last row, competing methods fail to reconstruct the ring finger, whereas ours restores both its shape and position.
These challenging cases highlight how rapid motion complicates reconstructing HFR videos, particularly when upscaling is performed independently. Our framework addresses this limitation by learning a frame-rate downscaling strategy that embeds high-frequency motion information into the downsampled LFR frames, facilitating more effective HFR reconstruction during upscaling.
Full sequences are provided in the supplementary material. 
{Moreover, we compare perceptual metrics in Fig. \ref{fig:bd_vary_psnr_perceptual}, including LPIPS \cite{zhang2018unreasonable} and VMAF \cite{li2016toward}. Our method achieves competitive perceptual quality compared to baseline approaches and shows slight improvements at higher bitrates, which may be attributed to the use of an L1-based training objective.  }

Moreover, we compare the temporal consistency of different approaches. As shown in Table~\ref{tab:qualitative_comparison}, our method effectively reduces frame quality fluctuation on the test video sequences, as indicated by the smaller standard deviation ($\sigma_\text{PSNR}$) of per-frame PSNR values in the upscaled videos. This is mainly due to our model’s improved interpolation quality, as evidenced by complete video results in the supplementary material.
We also visualize the temporal variations in Fig.~\ref{fig:temporal_diagram}.
{Furthermore, we provide quantitative metrics for temporal consistency in Fig. \ref{fig:bd_vary_psnr_perceptual}, including tOF \cite{chu2020learning} and wapring error \cite{lai2018learning}, which clearly indicate that the proposed method achieves significant improvements in temporal consistency. }
Compared to other methods, our framework not only reconstructs finer details but also effectively suppresses temporal inconsistencies across frames subjected to varying compression levels.

{
Finally, we conduct a Mean Opinion Score (MOS) study, where 25 participants (21 males, 4 females) aged 22$\sim$31 rate 61 video pairs on a scale of 1 to 5 in terms of visual quality and temporal consistency. The interface included standard media controls to facilitate precise visual comparison of the sequences. The presentation order was randomized to avoid bias. According to  ITU-T P.913 \cite{itu2016p913}, the Pearson correlation between the data provided by each subject and the average of all resulted in no subject being removed because of being considered
an outlier.  The MOS results are shown in Figure \ref{fig:user_study_results}, where 95\% confidence intervals are included to properly measure the agreement between subjects, according to the ITU-R BT.500-13 \cite{itu2012bt500}. 
Our method achieves the highest average MOS score of 4.02, outperforming the second-best method GIMM \cite{guo2024generalizable} with a score of 3.69, indicating a relatively clear preference for our method. Qualitative comparisons are provided in the supplementary material, and the user study interface is available on our project page.
}

\begin{table}[tb]
\caption{
Quality comparison (PSNR in dB and SSIM) of {compressed LFR videos} produced by different methods.  The original input video is used as the reference.
}
\centering
\begin{tabular}{|l|c|c|c|} 
\hline
\textbf{Downscaler} & \textbf{UCF101} & \textbf{Vimeo90k} & \textbf{SNU-FILM} \\
\hline 
\hline
Frame Skipping   &   35.94/0.9292 & 36.84/0.9254 &  37.13/0.9392  \\ \hline
STAA  \cite{xiang2022learning}  &    35.43/0.9221    &    36.26/0.9142    &    36.71/0.9317       \\
CSTVR{*} \cite{zhang2025continuous}   &   35.25/0.9150 &  36.29/0.9132  & --  \\ 
\rowcolor[HTML]{FFEEED}
TVRN   &   \textbf{35.56/0.9230}   &      \textbf{{36.32}/0.9150}  &  \textbf{36.75/0.9320}         \\
\hline
\end{tabular}
\label{tab:LFR_quality}
\end{table}

\begin{figure}[t]
  \centering

\setlength{\tabcolsep}{0.5pt}
\centering
\resizebox{\linewidth}{!}{
\scriptsize
\begin{tabular}{cccc>{\columncolor[HTML]{FFEEED}}c}
 {Ground Truth} & Frame Skipping & STAA~\cite{xiang2022learning} & CSTVR* \cite{zhang2025continuous}  & \textbf{TVRN}\\
\includegraphics[width=0.09\textwidth]{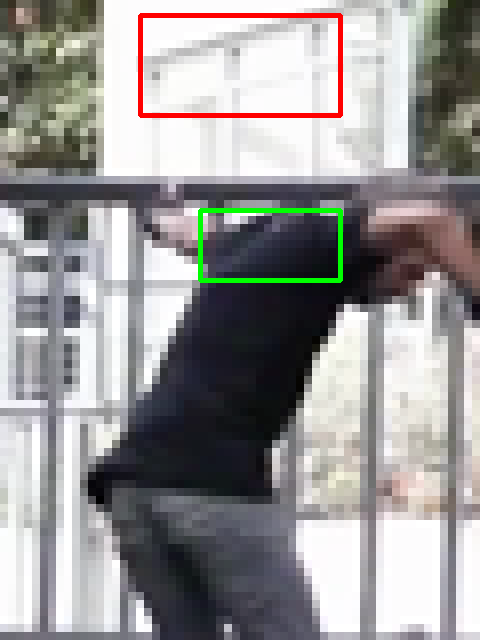}
&\includegraphics[width=0.09\textwidth]{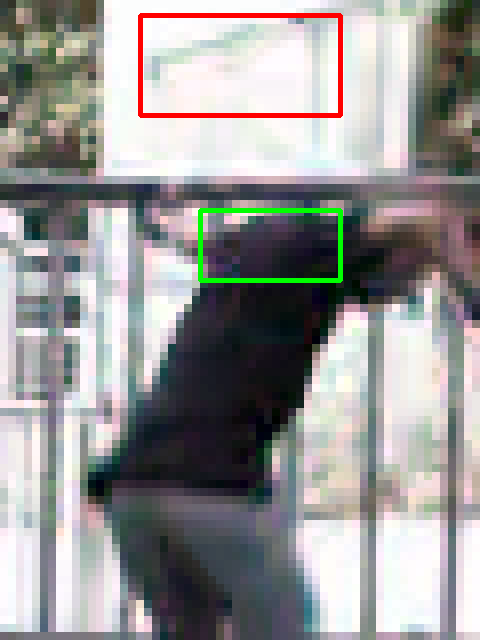}
&\includegraphics[width=0.09\textwidth]{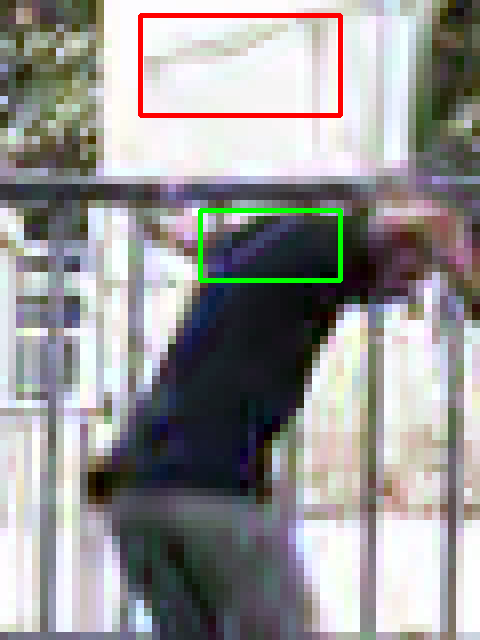}
&\includegraphics[width=0.09\textwidth]{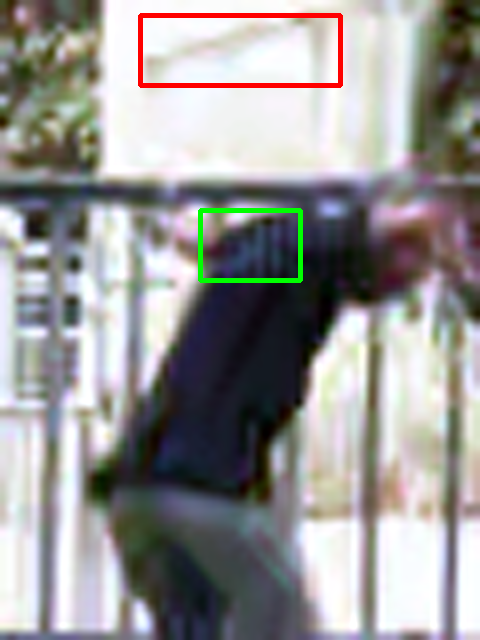}
&\includegraphics[width=0.09\textwidth]{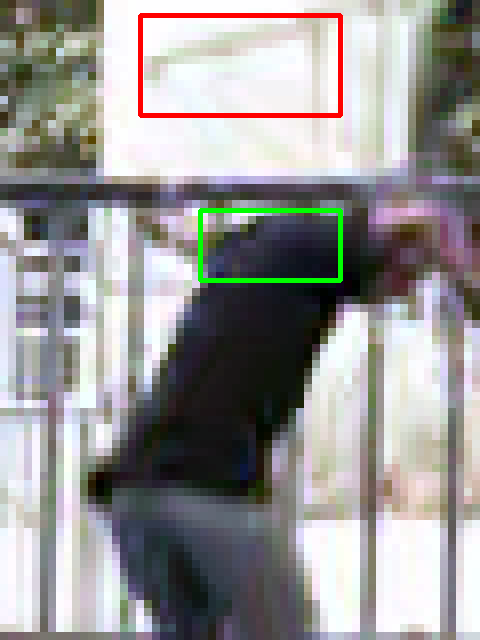}\\
\end{tabular}}

\resizebox{\linewidth}{!}{
\scriptsize
\begin{tabular}{ccccc}
\includegraphics[width=0.09\textwidth]{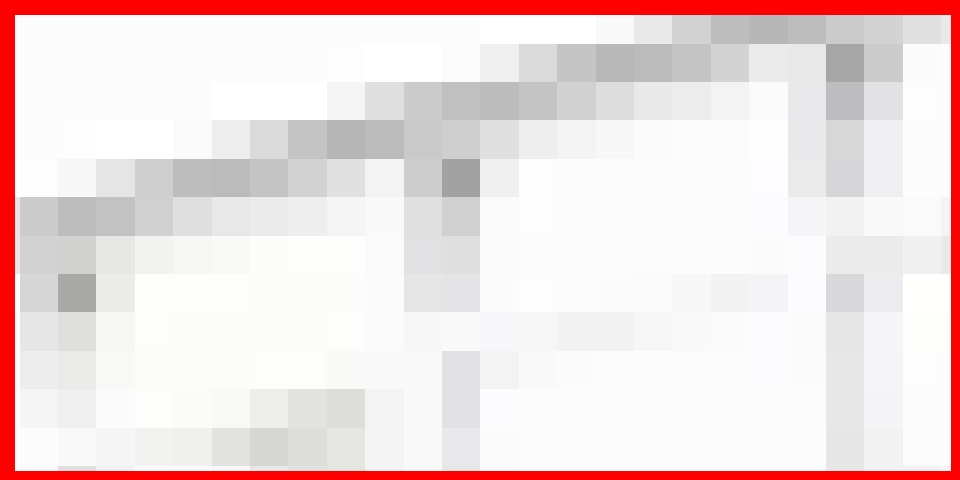}
&\includegraphics[width=0.09\textwidth]{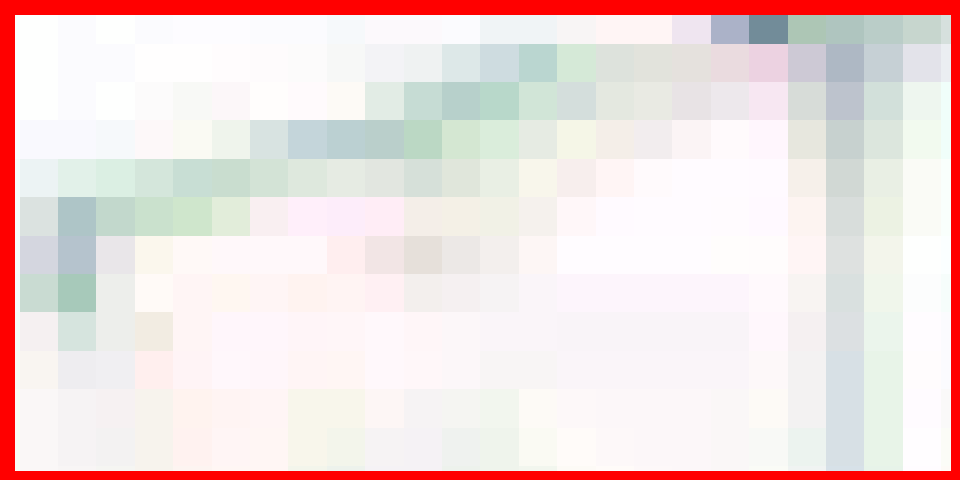}
&\includegraphics[width=0.09\textwidth]{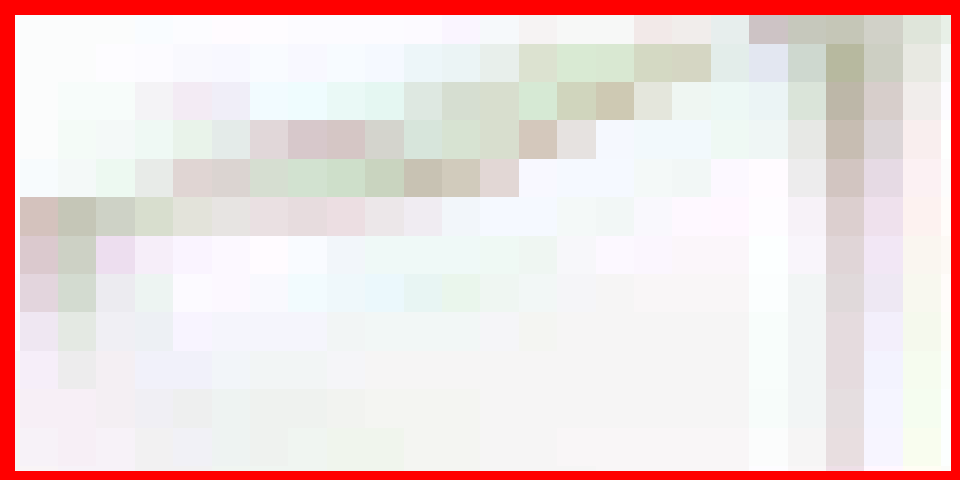}
&\includegraphics[width=0.09\textwidth]{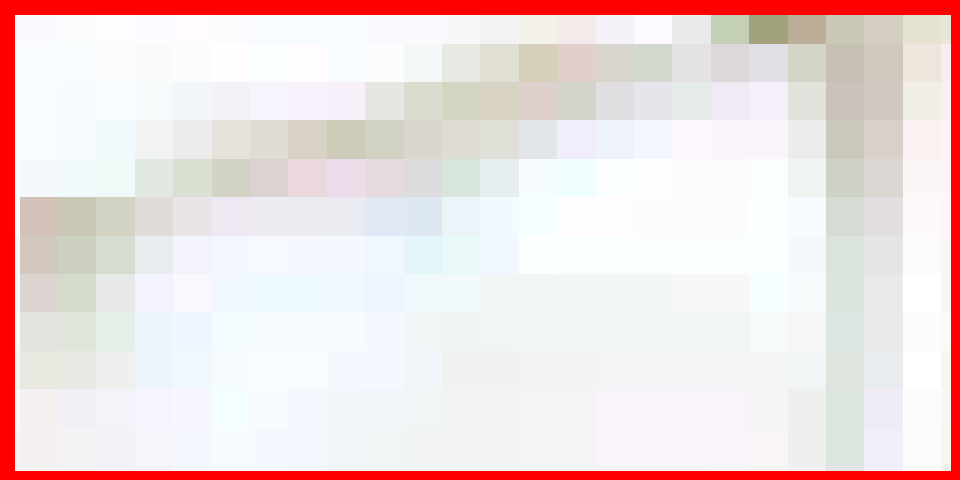}
&\includegraphics[width=0.09\textwidth]{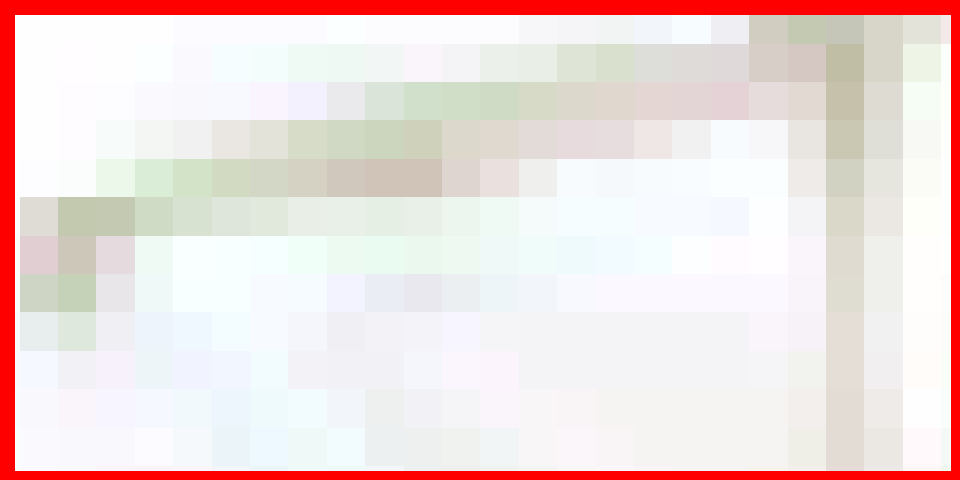}\\
\end{tabular}}

\resizebox{\linewidth}{!}{
\scriptsize
\begin{tabular}{ccccc}
\includegraphics[width=0.09\textwidth]{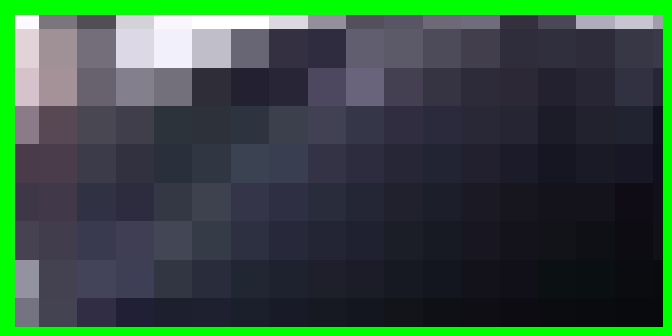}
&\includegraphics[width=0.09\textwidth]{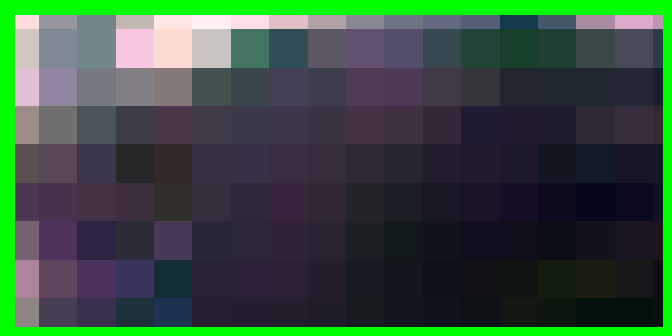}
&\includegraphics[width=0.09\textwidth]{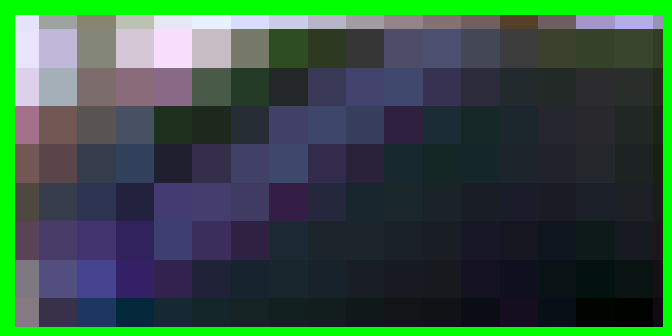}
&\includegraphics[width=0.09\textwidth]{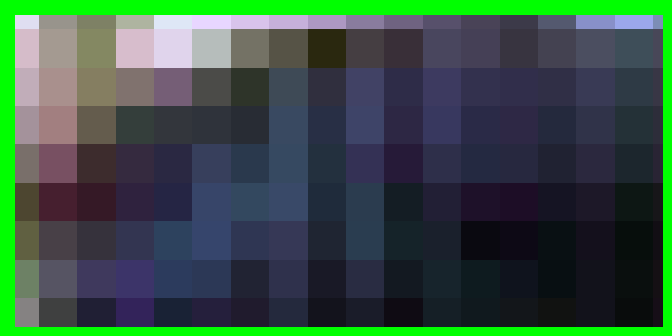}
&\includegraphics[width=0.09\textwidth]{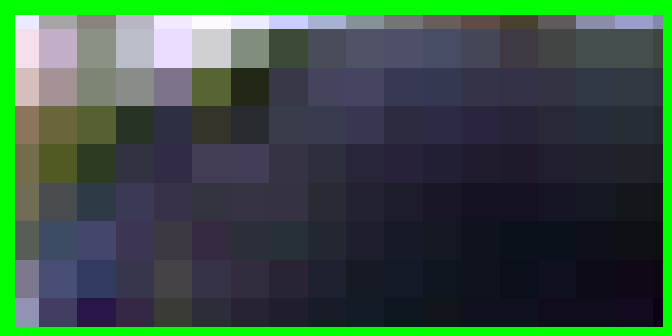}\\

BPP/PSNR &  0.3185/34.27 & 0.3208/33.39 & 0.3496/33.16 & \cellcolor[HTML]{FFEEED}\textbf{0.3185/33.77} 
\end{tabular}}

\caption{
\textbf{Qualitative comparison on sequence \textit{00026\_0036} from the Vimeo90K septuplet dataset, showing low-frame-rate frames compressed by HEVC}. Compared to frame-skipping, learned frame-rate downscaling methods better suppress color shifts ({red} box) and preserve fine details ({green} box). However, since high-frequency motion information is implicitly embedded in the downscaled frames, these methods typically yield lower PSNR than directly compressing the original frames.
}

  \label{fig: LFR_viz}
\end{figure}

\textbf{LFR qualitative evaluation.}
For temporal video rescaling, the visual quality of the LFR videos is also critical, as users with limited resources may preview them directly. As shown in Table~\ref{tab:LFR_quality}, our method achieves better LFR video quality compared to other temporal video rescaling approaches in terms of both PSNR and SSIM. 
For instance, our framework outperforms CSTVR* by 0.31 dB on the UCF101 test dataset.
Nevertheless, since downscaled videos are embedded with high-frequency temporal information from neighboring frames, their quality is inevitably lower than that of the compressed original frames. This degradation remains within an acceptable range.
Fig.~\ref{fig: LFR_viz} presents a qualitative comparison of LFR frames produced by different downscalers. Our approach generates visually pleasing results with fewer compression artifacts, such as blocking effects and color shifts, compared to directly compressing the original high-frame-rate frames.

\textbf{Rate-distortion performance.}
While the previous analysis focuses on reconstruction quality at a fixed bitrate, we further evaluate the RD performance of reconstructed HFR and downscaled LFR videos across varying compression levels. The RD curves on the test datasets are shown in Fig.~\ref{fig:rd_curve}. {Here, we denote directly encoding HFR videos as “Reference,” which serves as an upper-bound reference for temporal rescaling methods, as it compresses and transmits the intermediate frames directly.} To further quantify bitrate savings, we report Bjøntegaard Delta Bit Rate (BD-rate) \cite{bjontegaard2001calculation} results on three benchmark datasets, comparing our method with the state-of-the-art VFI model, GIMM-VFI \cite{guo2024generalizable}. For fairness, we also apply the same VQE model \cite{luo2022spatio} before GIMM-VFI for artifact removal, denoted as “GIMM-VFI{+VQE}.”  
Experimental results show that our approach is the \textbf{first} learnable temporal video rescaling method to surpass state-of-the-art VFI models in RD performance. Compared to GIMM-VFI, our method achieves BD-rate improvements (PSNR/SSIM) of -7.22\% / -16.40\%, -6.88\% / -3.88\%, and -12.82\% / -10.62\% on the UCF101, Vimeo90K, and SNU-FILM test sets, respectively.
Our approach is particularly effective at higher bitrates, where the HEVC codec better preserves the high-frequency information embedded in the downscaled LFR videos.

\begin{figure}[t]
\centering
    \begin{minipage}{0.5\textwidth}
        \centering
        \includegraphics[trim=0cm 0cm 0cm 0cm, width=0.75\textwidth]{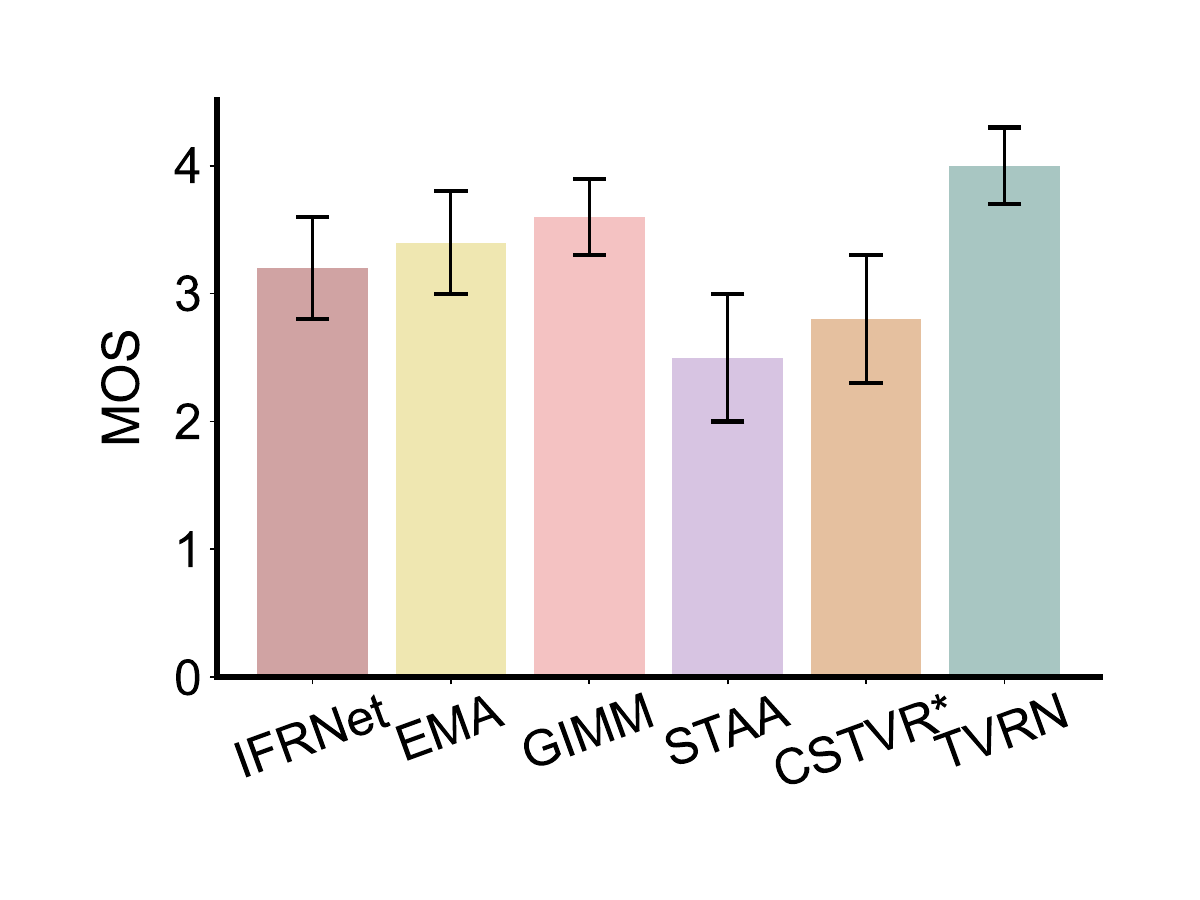}
        \caption{\textbf{Subjective quality comparison in terms of Mean Opinion Score (MOS).} 
Scores are averaged over all users, and error bars indicate standard deviation. Higher MOS values indicate better perceived visual quality.} 
        \label{fig:user_study_results}
    \end{minipage}
\end{figure}

\subsection{Ablation Studies}
\label{ablation}

\begin{table}[!t]
\centering
\caption{Ablation study of different gradient simulation strategies. Experiments are conducted on the Vimeo90K-Septuplet test set. 
BDBR is calculated using GIMM-VFI \cite{guo2024generalizable} as the baseline.}
\resizebox{0.5\textwidth}{!}{\Large
\begin{tabular}{|l|cccc|ccc|}
\hline
\multirow{2}{*}{\textbf{Method}}                   & \multicolumn{4}{c|}{\textbf{Simulation PSNR} }  &\multirow{2}{*}{\textbf{BDBR}}& \multirow{2}{*}{\textbf{Params.}} & \multirow{2}{*}{\textbf{FLOPs}}\\ 
\cline{2-5} 
 &  QP22 & QP27 & QP32 & QP37  & & & \\
\hline\hline
STE \cite{bengio2013estimating}              & 43.80 & 43.06 & 41.14 & 38.19  &          566.77\% &  N/A   & N/A \\ 
DenseBlock~\cite{tian2021self}               & 44.98 & 43.38 & 41.85 & 39.04     &     192.99\%    &  77.38M  & 2.39T \\ 
TINN (w/o Q-Inv)        & 45.88 & 43.40 & 42.25 & 39.40  &           -2.85\%&    \textbf{0.28M}  & \textbf{37.7G} \\
\rowcolor[HTML]{FFEEED}
TINN       & \textbf{46.58} & \textbf{43.75} & \textbf{43.44} &\textbf{41.86} &         \textbf{-12.82\%}&    0.28M  & 37.7G \\  
\hline
\end{tabular}
}
\label{table:error}
\end{table}

\begin{figure}[!t]
\setlength{\tabcolsep}{0.5pt}
\centering
\resizebox{\linewidth}{!}{
\scriptsize
\begin{tabular}{ccccc}
 {HEVC} & STE \cite{bengio2013estimating} & DenseBlock~\cite{tian2021self} & TINN {\tiny (w/o Q-Inv)} & TINN\\
\includegraphics[width=0.09\textwidth]{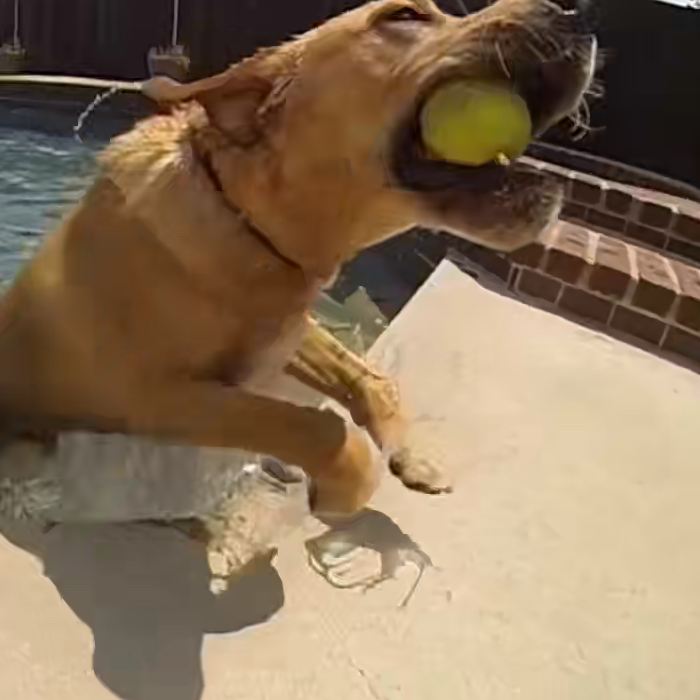}
&\includegraphics[width=0.09\textwidth]{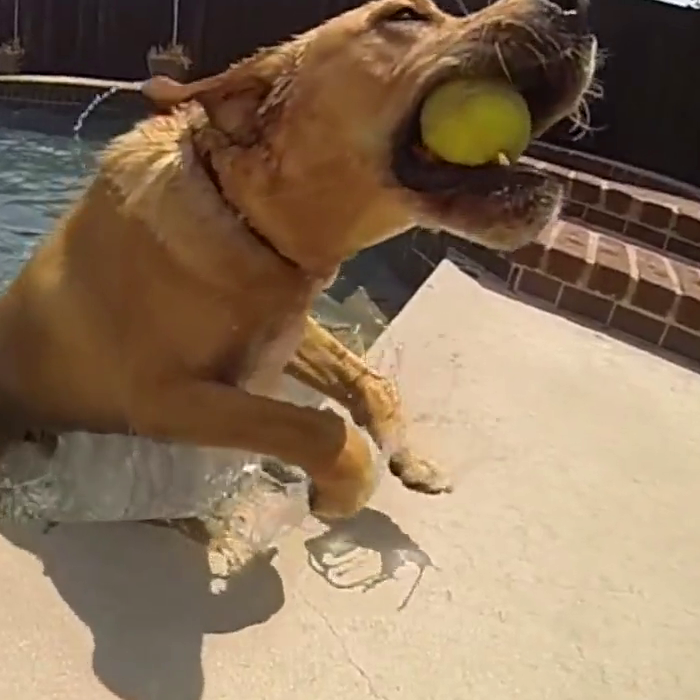}
&\includegraphics[width=0.09\textwidth]{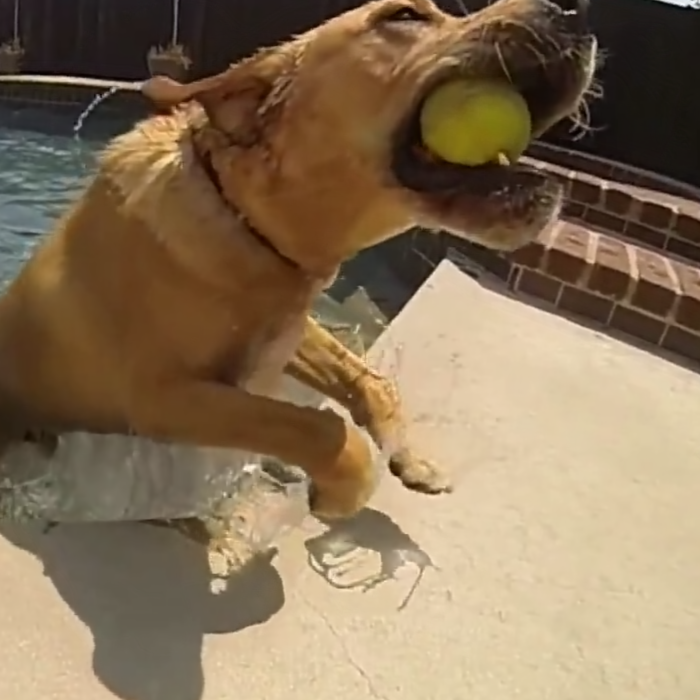}
&\includegraphics[width=0.09\textwidth]{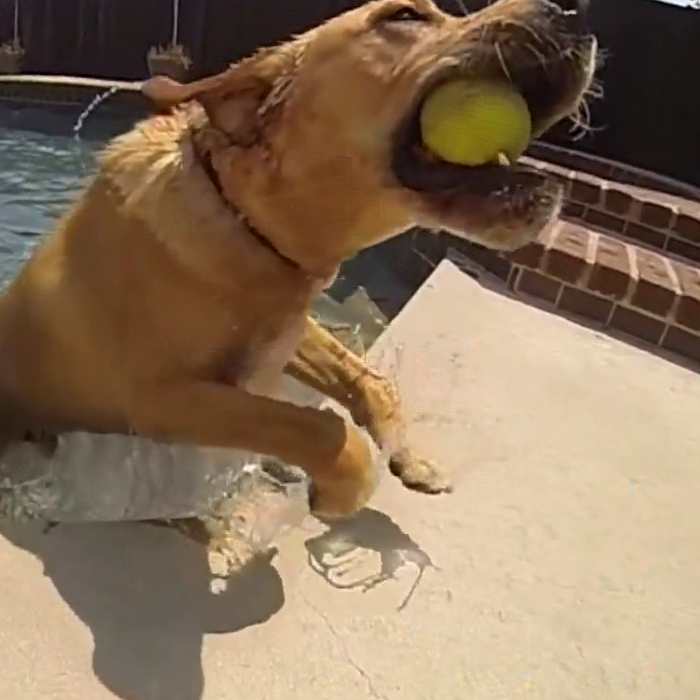}
&\includegraphics[width=0.09\textwidth]{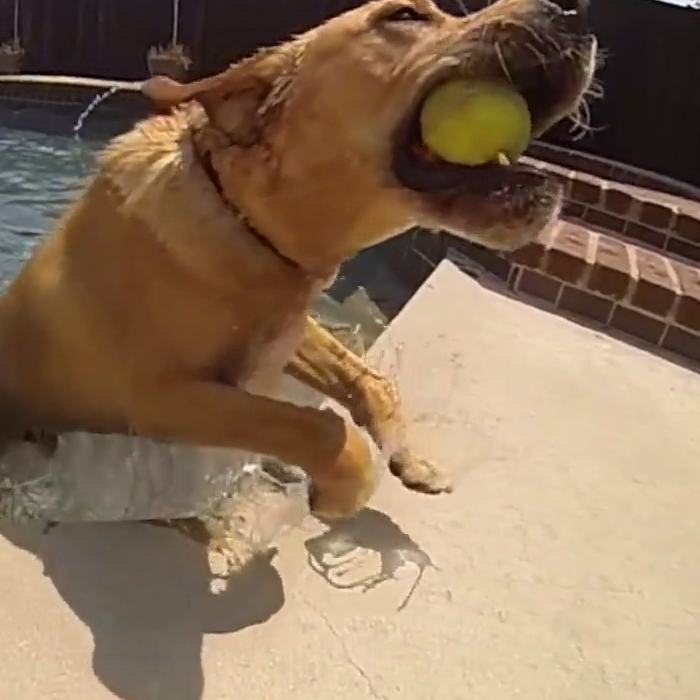}\\
PSNR &  38.80 & 39.88 & 39.98 & 41.41
\end{tabular}}
\resizebox{\linewidth}{!}{
\scriptsize
\begin{tabular}{ccccc}
\includegraphics[width=0.09\textwidth]{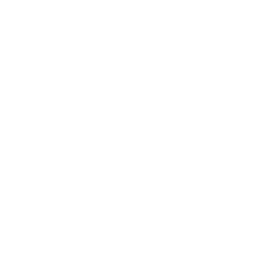}
&\includegraphics[width=0.09\textwidth]{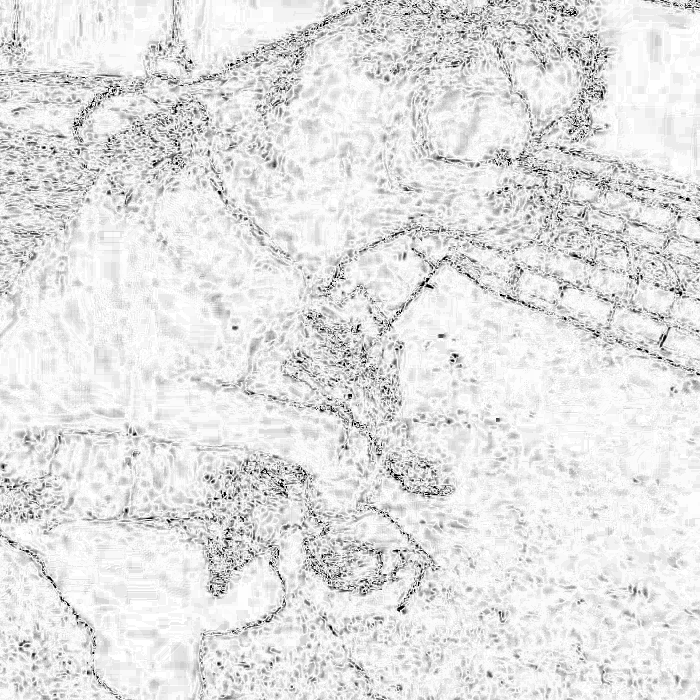}
&\includegraphics[width=0.09\textwidth]{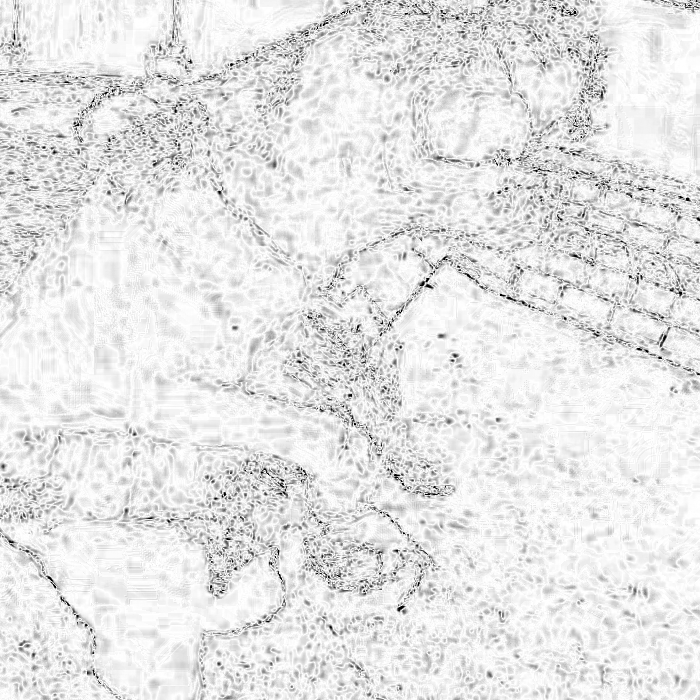}
&\includegraphics[width=0.09\textwidth]{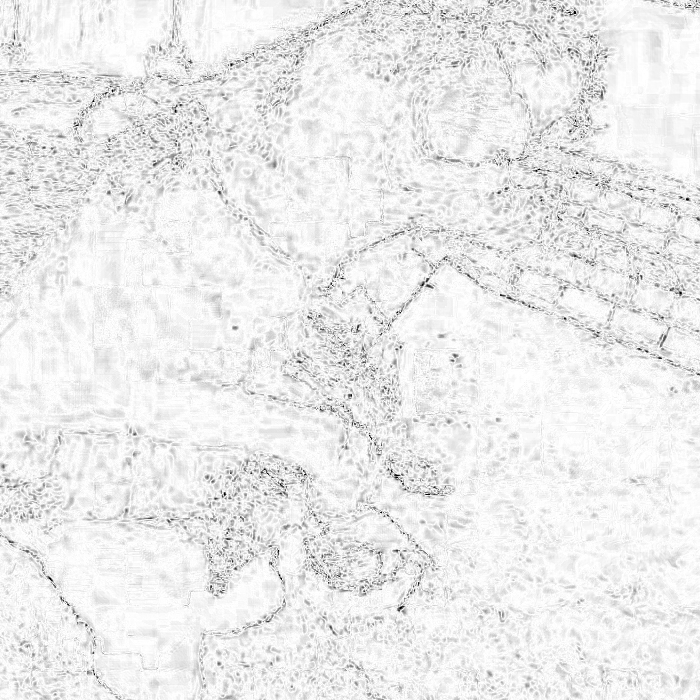}
&\includegraphics[width=0.09\textwidth]{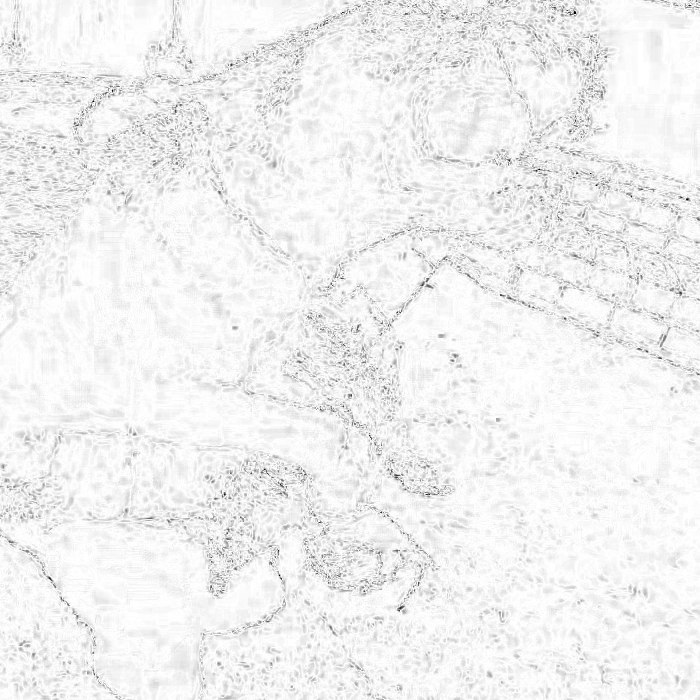}\\
\end{tabular}}

\caption{\textbf{Comparison of predicted compressed frames among compression simulation methods} on the \textit{YouTube\_0000} sequence from the SNU-FILM test dataset with the QP of 37. 
The second row shows the difference between the predicted and actual compressed frames in the luma channel.}
\label{fig:compression_simulation_viz}
\end{figure}
\textbf{Surrogate Network.}
The surrogate network plays a critical role in enabling the joint optimization of downscaling and upscaling in an end-to-end fashion. We explore three strategies for simulating the gradients of H.265 compression: (1) the biased straight-through estimator (STE) \cite{bengio2013estimating}; (2) stacked Dense3D-T blocks, used in spatial video rescaling methods \cite{tian2021self, iandola2014densenet}; and (3) our proposed Temporal Invertible Neural Network (TINN), which incorporates compression metadata.
Table~\ref{table:error} presents the simulation PSNR between predicted and actual H.265 compression residuals on the SNU-FILM test set, along with the BDBR-PSNR results for reconstructed HFR videos. Compared to the surrogate network built upon DenseBlock \cite{tian2021self}, the proposed TINN reduces the prediction PSNR from 39.04 to 41.86dB, thanks to the introduction of temporal frequency decomposition and Q-Invertible Blocks. In addition, our surrogate network is highly lightweight, requiring only 37.70 GFLOPs in total, of which the invertible coupling layers consume 11.32 GFLOPs, and 26.38 GFLOPs by the feature collapse module. To validate the effectiveness of the Q-Invertible  Blocks, we ablate the channel attention layer in each transformation module (w/o Q-Inv). While this modification has a negligible impact on model size and computational cost, it significantly degrades the surrogate network’s ability to generalize across different QPs. Furthermore, we visualize both the actual and predicted compressed frames in Fig.~\ref{fig:compression_simulation_viz}. The proposed TINN mimics compression distortions more accurately, particularly around the boundaries of moving objects.

Another insight from Table~\ref{table:error} is that a better gradient simulator also contributes to higher-quality HFR reconstruction.  Our method achieves a significant reduction in BDBR compared to the straight-through gradient estimation baseline \cite{bengio2013estimating} by minimizing the difference between the tensors used for gradient backpropagation and those used in the forward computation (line 5 of Algorithm \ref{algorithm-training_strategy}).

\begin{figure}[!t]
\centering
\resizebox{\linewidth}{!}{
\begin{tabular}{c}
\begin{tikzpicture}
\node[anchor=south west,inner sep=0] (image) at (0,0) {\includegraphics[width=1.1\linewidth]{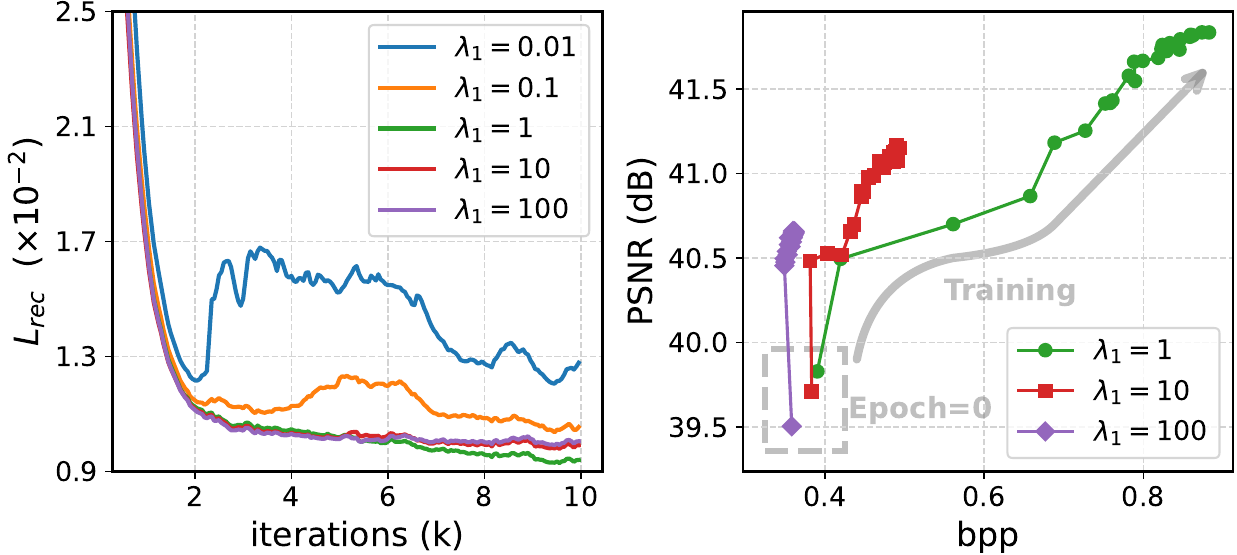}};
\node at ($(image.south)+(0,-0.2)$) {\small {(a) Training loss and validation performance}};

\end{tikzpicture}
\end{tabular}
}
\resizebox{\linewidth}{!}{
\begin{tabular}{c}
\begin{tikzpicture}
\node[anchor=south west,inner sep=0] (image) at (0,0) {\includegraphics[width=1.0\linewidth]{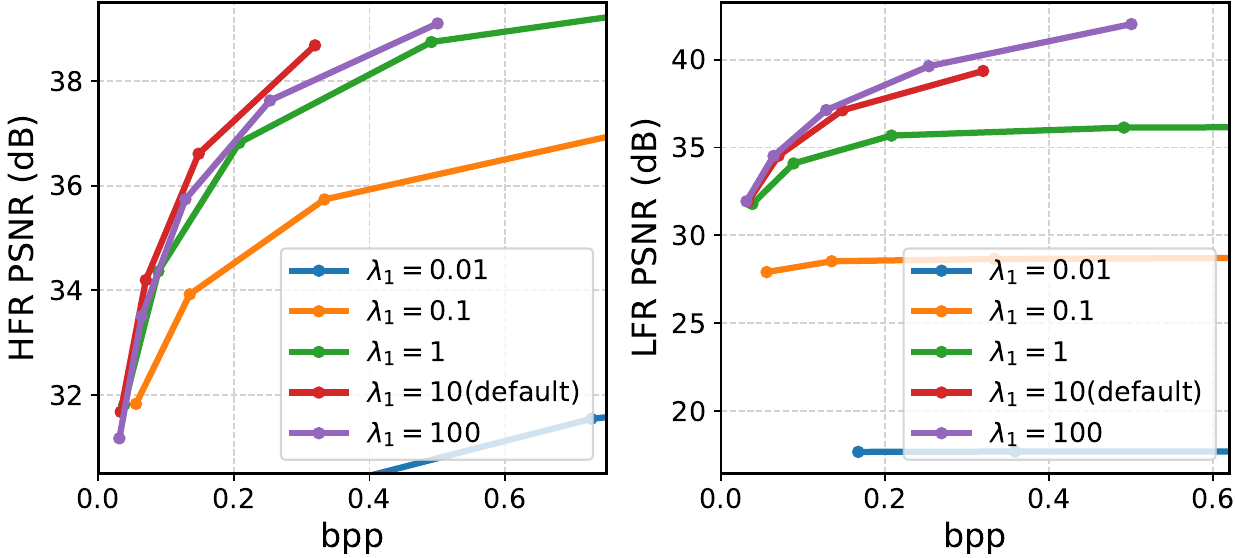}};
\node at ($(image.south)+(0,-0.2)$) {\small {(b) Rate-Dsitortion Curve on SNU-FILM test dataset}};
\end{tikzpicture}
\end{tabular}
}
\caption{\textbf{Ablation study on the guidance loss weight $\lambda$.} (a) Training loss and validation performance across different $\lambda$ values. The left subfigure shows HFR reconstruction loss over the first 10,000 iterations, while the right subfigure shows RD performance on the validation set over 50,000 iterations, evaluated every 2,000 iterations. (b) RD performance comparison for reconstructed HFR (left) and downscaled LFR (right) videos under varying $\lambda$ settings. The default value is $\lambda = 10$.}
\label{fig:ablation_combined}
\end{figure}

\textbf{Guidance Loss.}
Considering that the proposed surrogate network  cannot simulate the bitrate of the HEVC codec, we introduce a guidance loss to indirectly constrain the bitrate of the downscaled LFR videos, thereby adjusting the RD trade-off. To evaluate the effectiveness of the guidance loss, we vary the weight $\lambda$ in Eq.~\ref{equ: total} among ${10^{-2}, 10^{-1}, 1, 10, 100}$.
Figure~\ref{fig:ablation_combined}(a) shows the HFR reconstruction loss $\mathcal{L}_{rec}$ over the first 10K training iterations and the RD performance over 50K iterations on the validation set. The results highlight the crucial role of the guidance loss in stabilizing training and regulating the bitrate of the downscaled LFR videos. Specifically, a very small $\lambda$ hinders convergence, requiring $\lambda \geq 1$. However, an excessively large $\lambda$ overly restricts the bitrate of high-frequency information passed from the downscaler to the upscaler, thereby limiting joint optimization. Thus, we set $\lambda = 10$ to strike a balance between convergence and RD efficiency.
Figure~\ref{fig:ablation_combined}(b) further presents RD curves of reconstructed HFR and downscaled LFR videos for different $\lambda$ values, with QP values set to 17, 22, 27, 32, and 37. The results validate that $\lambda=10$ yields the better RD performance while maintaining high visual quality of the LFR videos.

\textbf{Network Design.}
To demonstrate the effectiveness of our proposed network design, we start with the baseline method, MIMO-VRN \cite{huang2021video}, and progressively replace or augment its components. The experimental results on the SNU-FILM test set are summarized in Table~\ref{tab:ablation_architecture}. Here, BDBR results are calculated on the SUN-FILM test dataset using GIMM-VFI  \cite{zhang2023extracting} as the anchor. Qualitative results under different ablation settings are provided in the Fig.~\ref{fig:compression_ablation_viz}. For the baseline MIMO-VRN, we directly divide the HFR video into even- and odd-indexed frames, feeding them into the two branches of the network, as shown in the first row of Table~\ref{tab:ablation_architecture}.

\subsubsection{MIMO-TWT}
We first evaluate the impact of integrating MIMO-TWT either before or after MIMO-VRN, denoted as \textit{+MIMO-TWT (before)} and \textit{+MIMO-TWT (after)} in Table~\ref{tab:ablation_architecture}. Results show that placing MIMO-TWT before MIMO-VRN yields better BDBR performance (\textbf{-0.65\%} vs. 15.64\%), as it effectively decomposes temporal frequency components, making them easier for MIMO-VRN to further process. As shown in Fig.~\ref{fig:compression_ablation_viz}, compared to the original MIMO-VRN (C1), incorporating MIMO-TWT (C3) leads to noticeably sharper texture reconstruction and significantly improves overall reconstruction quality ({38.36}\,dB vs. 30.34\,dB).

\subsubsection{High-Frequency Reconstruction Module}
Building on \textit{+MIMO-TWT (before)}, we replace the Dense2D-based HF reconstruction module in MIMO-VRN with our U-Net featuring coarse initialization, referred to as \textit{+HF Recon.} in Table~\ref{tab:ablation_architecture}. This is further enhanced by incorporating  contextual features from neighboring frames, denoted as \textit{+HF Recon. (w/ Ctx.)}. Experimental results show that our HF reconstruction module improves performance by \textbf{-10.34\%} BDBR-PSNR and \textbf{-9.91\%} BDBR-SSIM over the baseline, with a little additional inference cost. This is achieved using a lightweight bi-directional optical flow estimation module~\cite{jin2022enhanced}, which is 15$\times$ smaller than PWC-Net~\cite{sun2018pwc} while effectively handling complex motion. Visual comparisons are provided in the supplement material.

\begin{table}[tb]
\caption{
Investigation of different variants of TVRN. All models are trained with the same configuration for a fair comparison. BDBR results are calculated on the SUN-FILM test dataset using the VFI method GIMM-VFI  \cite{zhang2023extracting} as the anchor.
}
\centering
\resizebox{0.5\textwidth}{!}{
\begin{tabular}{|c|l| c|  c | c| c|  c | } 
\hline
\multirow{2}{*}{\textbf{Cfg}}&
\multirow{2}{*}{\textbf{Variants}} & \multirow{2}{*}{\textbf{Params.}} & \multirow{2}{*}{\textbf{FLOPs}}  & \multicolumn{2}{c|}{\textbf{BDBR}(\%)}\\
\cline{5-6}
& & & &  \textbf{PSNR} & \textbf{SSIM}\\
\hline \hline
1 & MIMO-VRN & 166.52K & 18.95G & 317.37 & 500.26 \\ \hline
2 & + MIMO-TWT{ (after)} & 1.96M & 185.05G & 15.64 & 27.31 \\
\rowcolor[HTML]{FFEEED}
3 & + MIMO-TWT{ (before)} & 1.96M & 185.05G & -0.65 & 2.12 \\
\hline
4 & \ + HF Recon. & 2.76M & 208.92G & -5.65 & -6.15 \\
\rowcolor[HTML]{FFEEED}
5 & \ + HF Recon. (w/ Ctx.) & 2.78M & 209.07G & -10.99 & -7.79 \\ \hline
6 & \ \ \ + Restore. (sep.) & 8.08M & 970.23G & -12.19 & -10.33 \\
7 & \ \ \ + Restore. (base) & 4.10M & 970.23G & -11.41 & -8.02 \\
8 & \ \ \ + Restore. (w/ QP) & 4.11M & 970.25G & -11.89 & -9.85 \\
\rowcolor[HTML]{FFEEED}
9 & \ \ \ + Restore. (w/ $\mathbf{f}_c$) & 4.53M & 972.81G & -12.82 & -10.61 \\ \hline
10 & w/o Downscaler & 4.53M & 972.81G & -0.27 & 2.80 \\

\hline
\end{tabular}
}
\label{tab:ablation_architecture}
\end{table}

\begin{figure}[!t]
\centering
\setlength{\tabcolsep}{0pt}
\renewcommand{\arraystretch}{1.0}
\scriptsize
\begin{tabular}{c}
\includegraphics[width=0.475\textwidth]{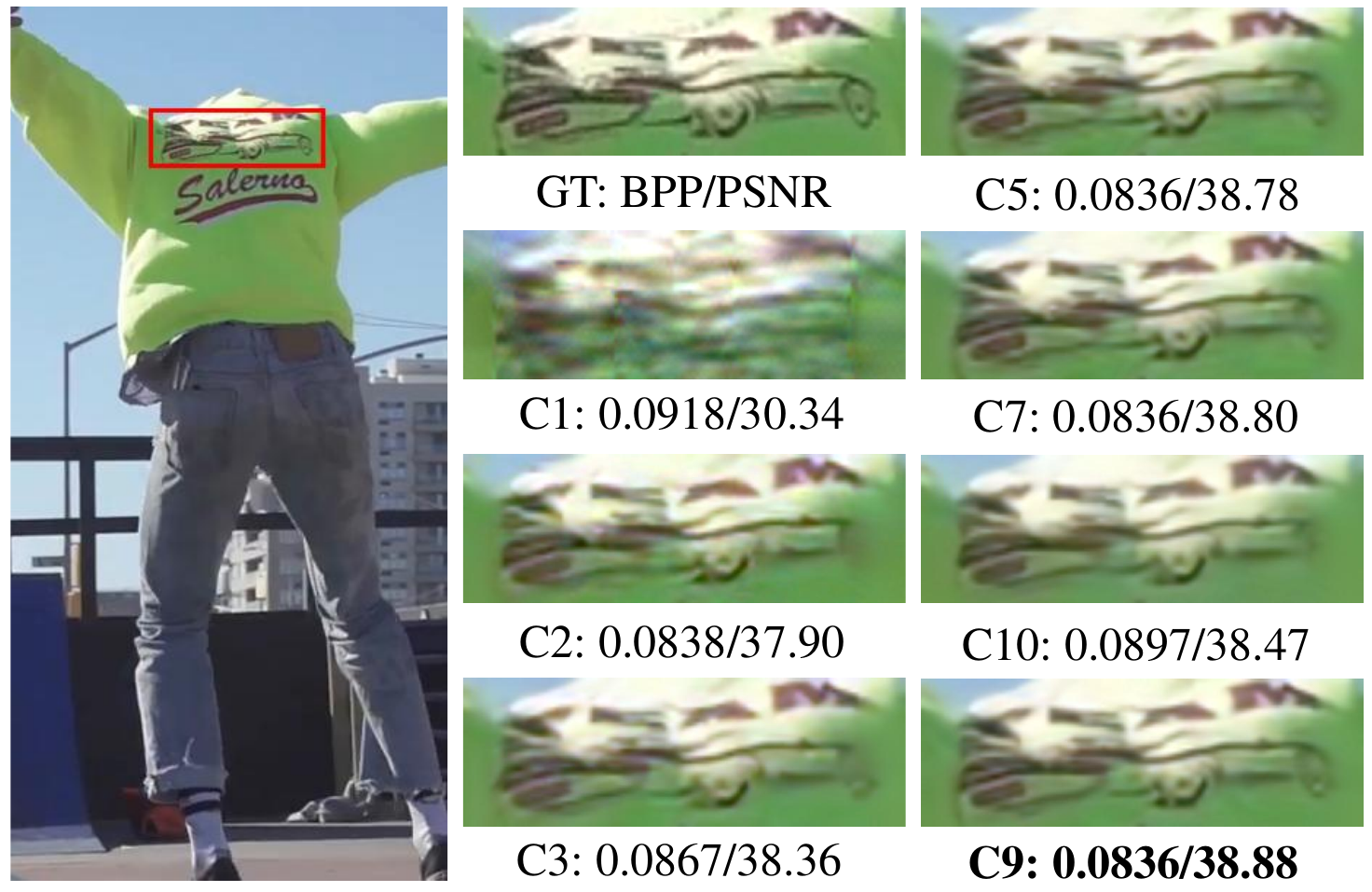} \\
\end{tabular}
\caption{\textbf{Visual comparison of ablation studies} on the \textit{YouTube\_0017} sequence from the SNU-FILM test dataset.}
\label{fig:compression_ablation_viz}
\end{figure}

\begin{figure}[!h]
  \centering
    \includegraphics[width=\linewidth]{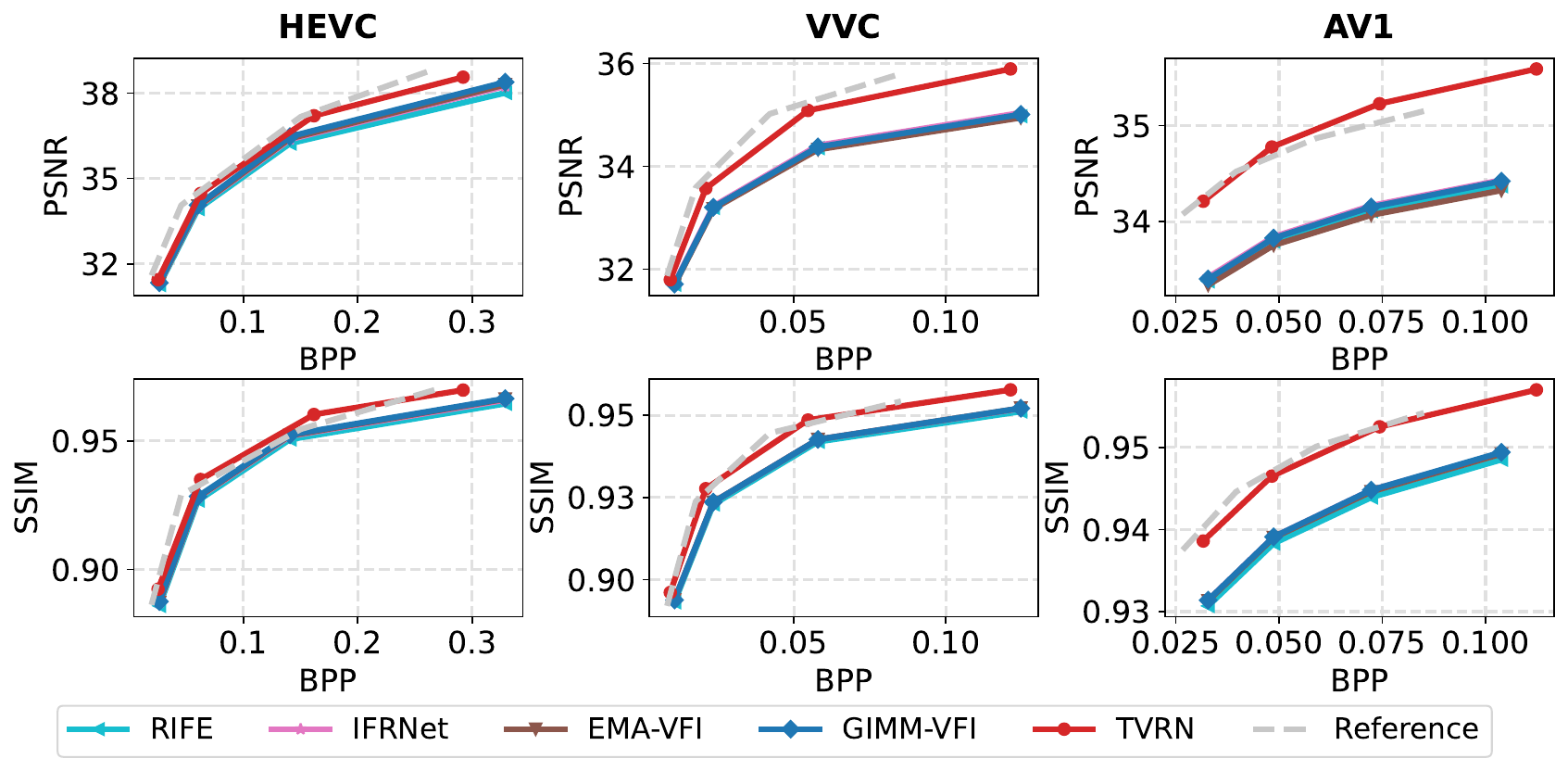}
    \vspace{-1em}
    \caption{{\textbf{Rate–distortion curves on 65-frame clips from the SNU-FILM test set}. The comparison includes three categories of methods: (1) frame-skipping-based approaches~\cite{kong2022ifrnet,zhang2023extracting,guo2024generalizable}; (2) learned frame-rate downscaling methods~\cite{xiang2022learning,zhang2025continuous}; and (3) direct compression of high-frame-rate (HFR) videos using various lossy codecs, which serves as a reference for offline coding efficiency.}}
    \label{fig:bd_vary_psnr_65frames}
\end{figure}

\subsubsection{Low-Frequency Restoration Module}
Building on \textit{+HF Recon. (w/ Ctx.)}, we apply VQE models for low-frequency restoration and evaluate four strategies:
\begin{inparaenum}[\itshape i\upshape)]
\item Training separate models for each QP, denoted \textit{+Restore. (sep.)};
\item Using a single model across all QPs, denoted \textit{+Restore. (base)};
\item Incorporating QPs or compression-aware features $\mathbf{f}_c$ into a unified model, denoted \textit{+Restore. (w/ QP)} and \textit{+Restore. (w/ $\mathbf{f}_c$)}, respectively.
\end{inparaenum}
While training QP-specific models is impractical for deployment, the first strategy acts as a performance benchmark. As shown in Table~\ref{tab:ablation_architecture}, applying generic VQE models yields limited gains due to mismatched behavior across compression levels (refer to Fig.~\ref{fig: restoration_locations}). Incorporating QPs improves generalization but still underperforms QP-specific models. In contrast, using compression-aware features allows content-adaptive restoration and enhances robustness across various compression levels. Leveraging these features, a single model even outperforms QP-specific models in BDBR performance (\textbf{-12.82\%} vs. -12.19\%).

\subsubsection{Downscaler}
To assess the impact of hidden HF information embedded in downscaled videos on HFR reconstruction, we replace the TVRN downscaler with frame skipping, denoted \textit{w/o Downscaler} in Table~\ref{tab:ablation_architecture}. This results in a significant drop in HFR reconstruction performance (-0.27\% vs. -12.82\%), highlighting two key insights: (1) the HF information embedded in downscaled videos is crucial for high-quality reconstruction, and (2) the performance gains primarily stem from the joint optimization of downscaling and upscaling, rather than post-processing.

{
\textbf{Sequence Length.}
To validate the effectiveness of the proposed method on long sequences,  we have conducted additional evaluations on clips with 65 frames from the SNU-FILM dataset~\cite{choi2020channel}. Specifically, we select 720p and 1080p sequences with more than 65 frames and split them into clips of 65 frames for evaluation, while still using the 7/5 frame rate conversion ratio. For cases where the sequence length is not divisible by the model’s group-of-pictures size, we pad the input by replicating the last frame and exclude the padded frames when computing distortion and bitrate. 
As shown in Fig.~\ref{fig:bd_vary_psnr_65frames}, the results indicate that our method consistently outperforms prior approaches in longer clips, confirming its generalization on clip length. 
 Here, we present directly compressing HFR videos as “reference” for offline coding efficiency, rather than a baseline expected to be surpassed in our scenario. 
As expected, directly encoding HFR videos often achieves better RD performance than temporal rescaling methods in Fig. \ref{fig:rd_curve}, since it explicitly encodes and transmits all intermediate frames. However, we observe that on longer sequences, the proposed method achieves comparable performance to direct HFR compression. This highlights the effectiveness of our approach under practical constraints, while maintaining competitive performance against coding efficiency focused references. 
}

\section{Model Analysis}
\begin{figure}[!t]
\setlength{\tabcolsep}{0.1pt}
\centering
\resizebox{\linewidth}{!}{
\scriptsize
\begin{tabular}{ccc}
\vspace{-3pt}
\includegraphics[width=0.09\textwidth]{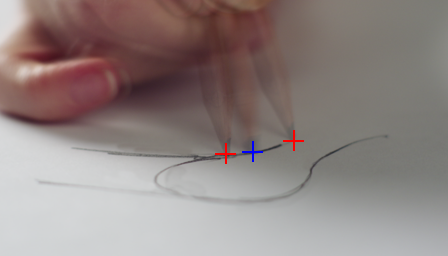}&
\includegraphics[width=0.09\textwidth]{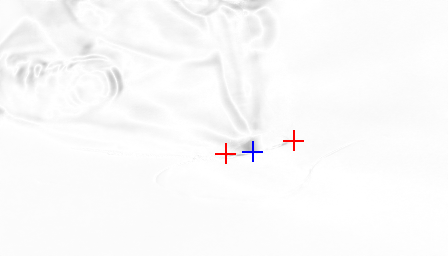}
&\includegraphics[width=0.09\textwidth]{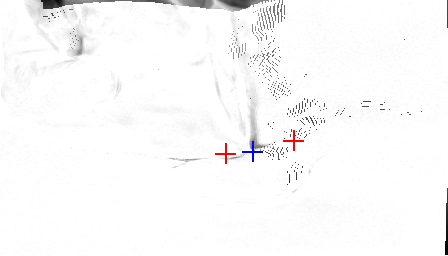}\\
{\tiny (a) Overlapped frames} & {\tiny (b) Ground truth} & {\tiny (c) Initial estimation} \\
\vspace{-3pt}
\includegraphics[width=0.09\textwidth]{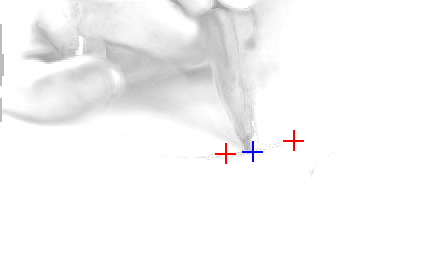}
&\includegraphics[width=0.09\textwidth]{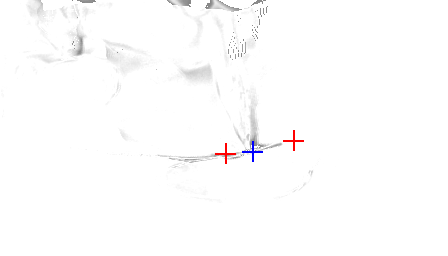}
&\includegraphics[width=0.09\textwidth]{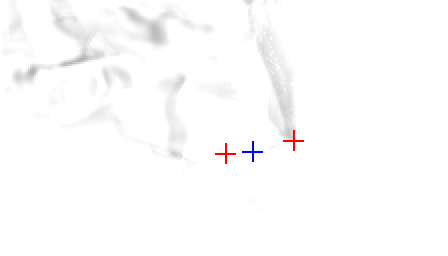}\\
{\tiny (d) Reconstructed HF} & {\tiny (e) w/o Context} & {\tiny (f) DenseBlock~\cite{huang2021video}}\\
\end{tabular}}
\caption{
{
\textbf{Visualization of temporal High-Frequency (HF) Components.}
(a) Overlapped neighbouring and target frames,
(b) Ground-truth HF components produced by the downscaler,
(c) Initial estimation of HF components using bidirectional optical flow at $t = 0.3$, as defined in Eq. (\textcolor{red}{9}),
(d) Reconstructed HF components from our model,
(e) Reconstruction without contextual features,
(f) Reconstruction produced by stacked DenseBlocks~\cite{huang2021video}.
For visual clarity, the pencil's positions in neighbouring frames are marked with red crosses, and its position in the target frame is marked with blue crosses. 
}
}
\label{fig:viz_hf_discarded}
\end{figure}

{
\subsection{Analysis of Temporal High-Frequency Components}
\label{subsec: viz_temporal_hf}
Compared to previous methods, the proposed invertible architecture effectively regularizes the information loss from frame-rate downscaling into temporal high-frequency (HF) components. Thus, accurately reconstructing these components is essential for HFR video reconstruction. To this end, our framework uses bidirectional optical flow to align contextual features from neighboring frames with the HF information of discarded frames.
As shown in Fig.~\ref{fig:viz_hf_discarded}, we visualize the original and reconstructed HF components from different reconstruction modules. The aligned residuals of warped neighboring frames better approximate the original HF components than direct residuals, making them more suitable for initialization. We further show the impact of removing contextual features (w/o Context) and replacing the context-aware U-Net with a DenseBlock-based module~\cite{huang2021video}. The results highlight the importance of both contextual information and the U-Net architecture in improving HF reconstruction.
}

{\subsection{Analysis of Failure cases}
We sample one frame every 50 frames from the \textit{GOPR0881\_11\_01} sequence in the SNU-FILM test set \cite{choi2020channel}, as shown in Fig. 
\textcolor{red}{18}. In this example, the large temporal gap leads to long-range motion that exceeds the effective receptive field of optical flow estimation. As a result, flow-based methods, including EMA \cite{zhang2023extracting}, STAA \cite{xiang2022learning}, and our method, suffer from noticeable misalignment artifacts. In contrast, GIMM \cite{guo2024generalizable}, which relies on implicit modeling rather than explicit motion estimation, is more robust to such long-range motion. This highlights a limitation of flow-based approaches under extreme temporal displacement.}

{\begin{table}[t]
\centering
\caption{End-to-end efficiency comparison in a sender--receiver pipeline. 
The sender performs temporal downscaling followed by video encoding; the receiver handles decoding and temporal upscaling. 
}
\label{tab:end_to_end_efficiency}
\vspace{4pt}
\resizebox{0.5\textwidth}{!}{
\begin{tabular}{l|ccccc|cc}
\toprule
\multirow{2}{*}{\textbf{Methods}} & \multicolumn{5}{c|}{\textbf{Latency (second/clip)}} & \multicolumn{2}{c}{\textbf{Peak Memory (GB)}} \\
\cline{2-8}
& Downscaling & Encoding & Decoding & Upscaling & Total & CPU & GPU \\
\midrule
IFRNet~\cite{kong2022ifrnet}      & 0.00 & 1.01 & 2.06 & 0.46 & 3.53 & 2.05 & 1.04 \\
EBME-H~\cite{jin2022enhanced}     & 0.00 & 1.01 & 2.06 & 1.01 & 4.08 & 2.00 & 0.63 \\
RIFE~\cite{huang2022rife}         & 0.00 & 1.01 & 2.06 & 1.22 & 4.29 & 1.97 & 0.58 \\
GIMM~\cite{guo2024generalizable}$^\dagger$  & 0.00 & 1.01 & 2.06 & 6.04 & 9.11 & 2.08 & 8.40 \\
\midrule
STAA~\cite{xiang2022learning}     & 1.06 & 1.01 & 2.06 & 5.94 & 10.07 & 2.09 & 2.88 \\
TVRN-S (Ours)                     & 1.02 & 1.01 & 2.06 & 1.12 & 5.21 & 2.06 & 1.62 \\
TVRN (Ours)                       & 1.56 & 1.01 & 2.06 & 5.75 & 10.38 & 2.07 & 3.74 \\
\bottomrule
\end{tabular}
}
\begin{tablenotes}
  \footnotesize
  \item[$\dagger$] $\dagger$ Due to the Out-Of-Memory problem, we utilize an Nvidia RTX 1080Ti with 11GB memory for temporal upscaling.
\end{tablenotes}
\end{table}}

{\subsection{End-to-end Complexity Analysis}}

{To reflect practical deployment, we use a server-grade sender (NVIDIA GeForce RTX 3090 GPU and dual Intel Xeon Gold 5118 CPUs @ 2.30\, GHz, 48 cores) and a consumer-grade laptop receiver (NVIDIA GeForce RTX 3050 4\, GB Laptop GPU and Intel Core i7-13700H CPU, 14 cores). All experiments are conducted on the SNU-FILM Medium dataset~\cite{choi2020channel} using HEVC codecs under the same configurations as the main experiments.
For all methods, we report end-to-end latency, peak memory usage, and RD performance in Table~\ref{tab:end_to_end_efficiency}. 
To further assess the behavior under different resource constraints, we evaluate the receiver side in both GPU-enabled and GPU-disabled settings. 
For memory usage, we report the peak GPU memory usage (Maximum Resident Set Size) measured over the entire end-to-end pipeline. 
Notably, all models are evaluated without specialized optimizations.}

{As shown in Table \ref{tab:end_to_end_efficiency}, with GPU enabled, the proposed models achieve end-to-end latency similar to lightweight VFI methods. Specifically, for encoding a 1080p sequence with 7 frames, the full TVRN model incurs an end-to-end latency of 10.38 seconds, while the lightweight variant TVRN-S has a latency of 5.21 seconds. In comparison, lightweight VFI methods show latencies ranging from 3.53 to 8.74 seconds. Thus, the full TVRN model is primarily focused on quality, whereas TVRN-S strikes a better quality-efficiency balance, making it more suitable for deployment under tighter computation and memory constraints. For peak memory usage, the full model reaches a maximum of 3.74 GB, which is compatible with a laptop GPU with 4 GB of memory. }

\begin{figure}[!t]
\label{fig:bad_cases}
\setlength{\tabcolsep}{0.5pt}
\centering
\resizebox{\linewidth}{!}{
\scriptsize
\begin{tabular}{ccccccc}
Overlapped Input& Ground Truth &  EMA \cite{zhang2023extracting} & GIMM \cite{guo2024generalizable}& STAA \cite{xiang2022learning} & TVRN(Ours)\\

\includegraphics[width=0.09\textwidth]{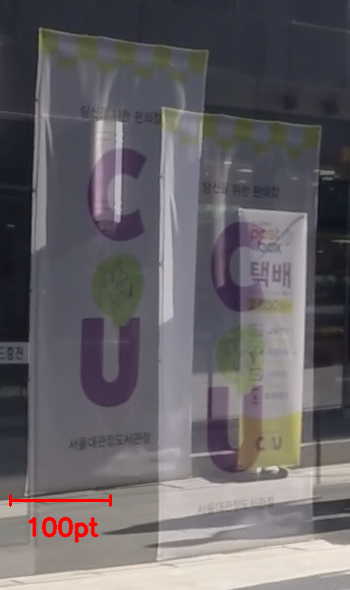}
&\includegraphics[width=0.09\textwidth]{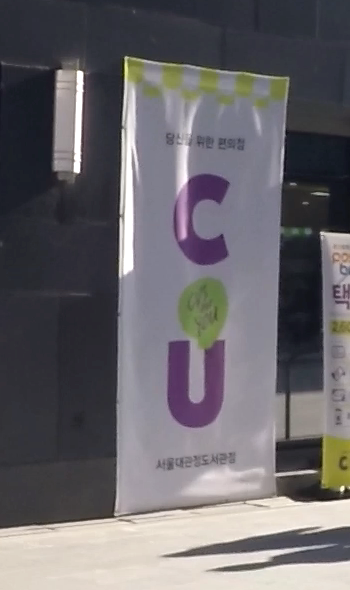}
&\includegraphics[width=0.09\textwidth]{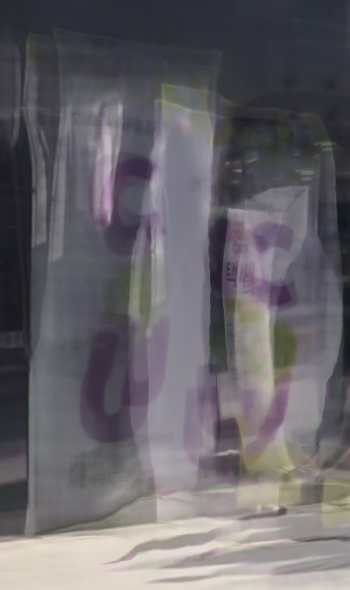}
&\includegraphics[width=0.09\textwidth]{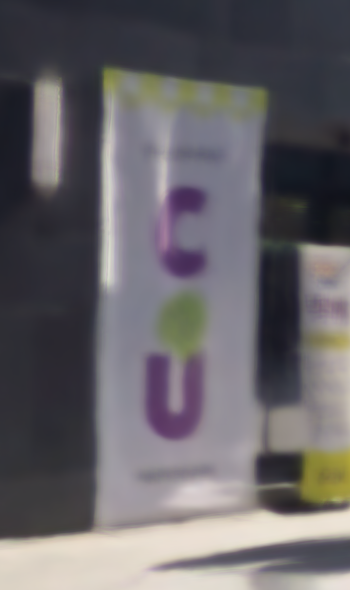}
&\includegraphics[width=0.09\textwidth]{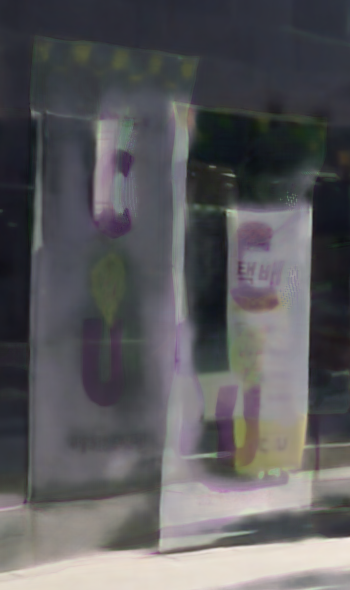}
&\includegraphics[width=0.09\textwidth]{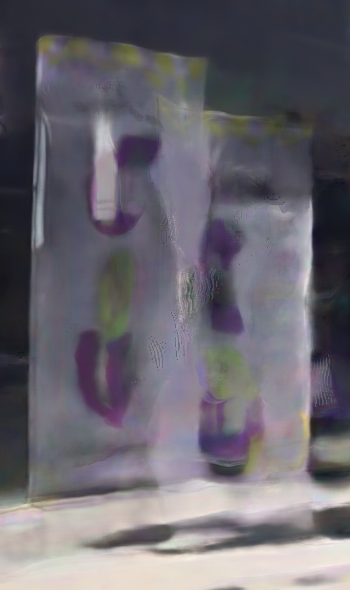}\\

& PSNR/BPP &  21.75/0.4791  & 21.67/0.4791 & 20.90/0.5094 & 21.20/0.5340
\end{tabular}
}
{\caption{Visualization of failure cases on  the \textit{GOPR0881\_11\_01} sequence from the SNU-FILM test dataset. }}
\end{figure}

\section{Conclusion}
This paper proposes TVRN, a Temporal Video Rescaling Network that validates the effectiveness of invertible network architectures for compression-aware temporal video rescaling. To improve robustness against non-differentiable lossy compression and enable end-to-end training, we introduce a novel surrogate network that efficiently mimics the compression distortion introduced by the  lossy video codec. Furthermore, we identify a key challenge in applying video quality enhancement models to HFR reconstruction and address it by incorporating compression-aware features. Experimental results show that TVRN enables efficient temporal video rescaling while maintaining compatibility with modern video codecs.

Although the proposed model achieves high-fidelity HFR reconstruction, several challenges remain for deployment in bandwidth-efficient video streaming:
(1) dynamically adjusting the grouping strategy beyond simple odd/even splitting to adapt to network fluctuation \cite{james2019beta};
(2) reducing the upscaling complexity to meet the constraints of resource-limited edge devices.
In addition, the proposed surrogate network has potential applications in various computer vision tasks, such as video pre-processing \cite{lu2024preprocessing} and steganography \cite{zhang2024editguard}. 
 Another direction is to integrate temporal and spatial video rescaling for more flexible adaptation to diverse transmission scenarios.

\bibliographystyle{IEEEtran}
\bibliography{reference}

\clearpage


\twocolumn[
\begin{center}
    \Huge \emph{Supplementary Material for} TVRN: Invertible Neural Networks for Compression-Aware Temporal Video Rescaling
    \vspace{15mm}
\end{center}
]

\appendix

This supplementary material provides additional analyses, network architecture details, and more results of our proposed compression-aware temporal video rescaling method. The content is organized as follows:

\noindent
\begin{tabularx}{\linewidth}{@{}X >{\raggedleft\arraybackslash}p{.5em}@{}}
\textbf{\ref{sec: model_analysis}. Additional Model Analyses} \dotfill & \textbf{\pageref{sec: model_analysis}} \\
\  {\textcolor{red}{I}. \textit{Impact of the Downscaler on Downscaled Videos}} \dotfill & \textbf{\pageref{subsec: impact_downscaler}} \\
\  {\textcolor{red}{II}. \textit{Analysis of Learned Compression-Aware Features}} \dotfill & \textbf{\pageref{subsec: viz_compression_aware}} \\
\ {\textcolor{red}{III}. \textit{Computation Details of Frequency Analysis}} \dotfill & \textbf{\pageref{subsec: frequency_analysis}} \\
\  {\textcolor{red}{IV}. \textit{Theoretical Justification}} \dotfill & \textbf{\pageref{subsec:theoretical}} \\
\  {\textcolor{red}{V}. \textit{Upper Bound on Gradient Estimation Error}} \dotfill & \textbf{\pageref{subsec:upper_bound}} \\

\textbf{\ref{sec: supp_network}. Detailed Network Architecture} \dotfill & \textbf{\pageref{sec: supp_network}} \\

\  {\textcolor{red}{VI}. \textit{Transformation Modules}} \dotfill & \textbf{\pageref{subsec: transformation_module}} \\

\  {\textcolor{red}{VII}. \textit{Context-Aware U-Net}} \dotfill & \textbf{\pageref{subsec: context_unet}} \\

\  {\textcolor{red}{VIII}. \textit{Lightweight Bi-directional Optical Flow Network}}  & \textbf{\pageref{subsec: lightweight_op_net}} \\

\  {\textcolor{red}{IX}. \textit{Compression Encoder and Ranker}} \dotfill & \textbf{\pageref{subsec: compression_encoder}} \\

\textbf{\ref{sec: lossless_performance}. Experiments with Rounding-Based Quantization}  & \textbf{\pageref{sec: lossless_performance}} \\

\  {\textcolor{red}{X}. \textit{Experimental Setup}} \dotfill & \textbf{\pageref{subsec: experiment_setup}} \\

\  {\textcolor{red}{XII}. \textit{Experimental Results}} \dotfill & \textbf{\pageref{subsec: experimental_results}} \\

\textbf{\ref{sec: more_viz}. More Visualization Results} \dotfill & \textbf{\pageref{sec: more_viz}} \\
\end{tabularx}

\setcounter{section}{5}
\section{Additional Model Analyses}
\label{sec: model_analysis}

\subsection{Impact of the Downscaler on Downscaled Videos}
\label{subsec: impact_downscaler}
It is necessary to investigate downscaled videos before/after downscaling for information regularization, since temporal video rescaling methods aim to embed motion information from high-frame-rate videos into lower-frame-rate downsampled frames in a visually imperceptible manner. We compare the LFR videos generated by CSTVR* \cite{zhang2025continuous} and our method in terms of visual quality and encoding statistics.
As shown in Fig. \ref{fig:viz_downscaled}, both methods produce LFR videos with appealing visual fidelity with the original frames. However, compared to CSTVR*, our approach introduces fewer distortions and better preserves background details.

We further investigate the impact of downscaling on motion information generated by video codecs. Specifically, we present motion vector fields extracted from the HEVC codec with a quantization parameter of 22 in Fig. \ref{fig:viz_downscaled}. Our method generates more motion vectors in relatively static background regions, while producing fewer motion vectors in areas with complex motion compared to CSTVR*.
This selective motion vector distribution explains the superior performance of our method: (1) Allocating more motion vectors to static backgrounds helps maintain fine details and temporal consistency; (2) Reducing motion vectors in complex motion areas avoids noisy or erroneous estimates, effectively preserving embedded high-frequency information. This adaptive allocation improves reconstruction quality by balancing smooth regions and robustness in challenging motion areas.

\setcounter{figure}{12}
\begin{figure}[!t]
\setlength{\tabcolsep}{0.1pt}
\centering
\resizebox{\linewidth}{!}{
\scriptsize
\begin{tabular}{ccc}
\includegraphics[width=0.09\textwidth]{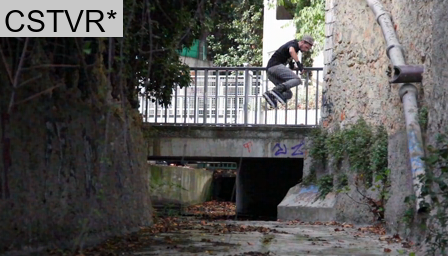}
&\includegraphics[width=0.09\textwidth]{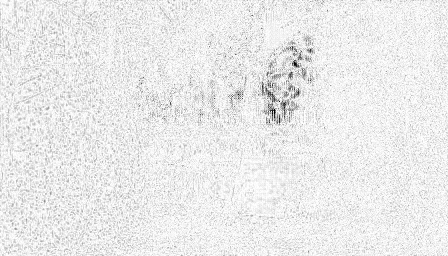}
&\includegraphics[width=0.09\textwidth]{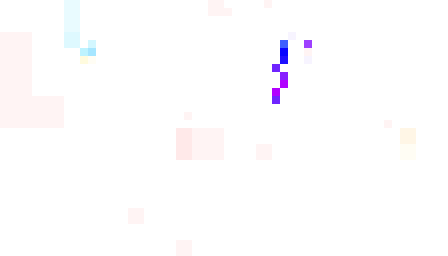}\\
\includegraphics[width=0.09\textwidth]{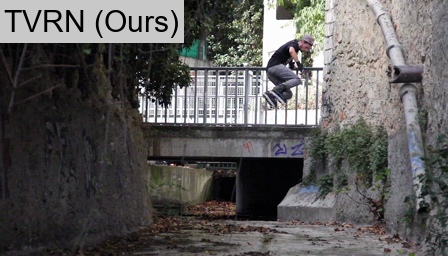}
&\includegraphics[width=0.09\textwidth]{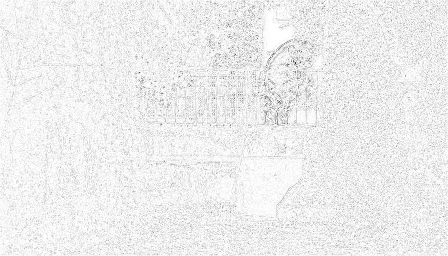}
&\includegraphics[width=0.09\textwidth]{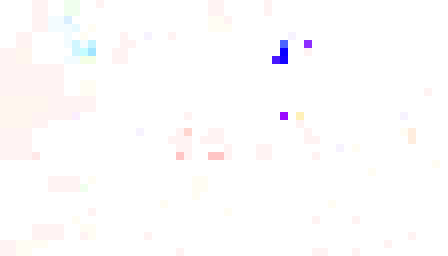}\\
{\tiny Downscaled LFR video} & {\tiny Residue} & {\tiny Motion vector}\\
\end{tabular}}
\caption{\textbf{Visualization of downscaled videos produced by CSTVR{*} and our method on the \textit{00026\_0036} sequence from the Vimeo test dataset}. Both methods produce visually pleasant results. We also show the difference between the downscaled and original frames after compression, amplified 10$\times$ for clarity.
Lastly, we visualize the motion vector field derived from the HEVC codec.}
\label{fig:viz_downscaled}
\end{figure}

\subsection{Analysis of Learned Compression-Aware Features }
\label{subsec: viz_compression_aware}

\begin{figure}[!t]
\setlength{\tabcolsep}{8pt}  
\centering
\begin{tabular}{cc}
\includegraphics[width=0.22\textwidth]{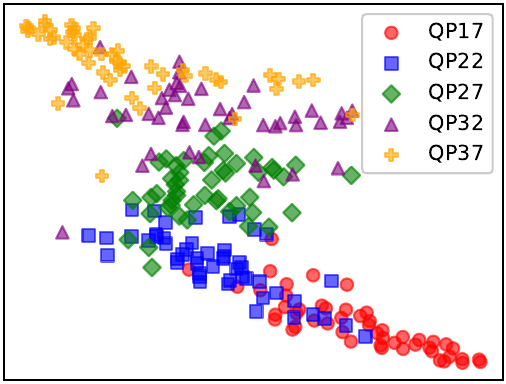}
& \includegraphics[width=0.22\textwidth]{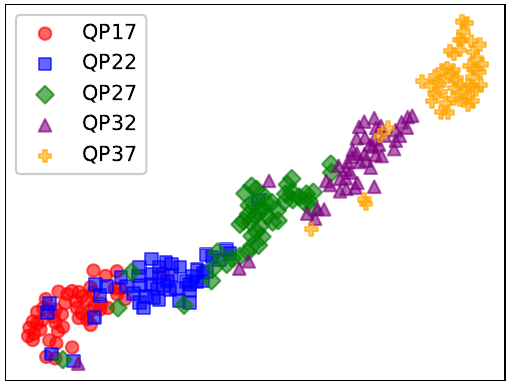}\\
{\small (a) L1 loss} & { \small (b) Ranking Loss $\mathcal{L}_{rank}$} \\
\end{tabular}
  \caption{\textbf{Visualization of compression-aware features learned with different loss functions using t-SNE}  \cite{van2008visualizing}. 
  Forty videos from the Vimeo90k test dataset are compressed with a set of different QPs. 
  Rank loss effectively strengthens the clustering of latent variables for medium QPs.
  }
\label{fig:viz_feature_tsne}
\end{figure}

\begin{figure*}[!t]
\setlength{\tabcolsep}{1pt} 
\renewcommand{\arraystretch}{0.5} 
\centering
\begin{tabular}{>{\centering\arraybackslash}m{0.02\textwidth} >{\centering\arraybackslash}m{0.2\textwidth} >
{\centering\arraybackslash}m{0.2\textwidth} >
{\centering\arraybackslash}m{0.2\textwidth} >
{\centering\arraybackslash}m{0.2\textwidth}}
\multirow{2}{*}{\rotatebox{90}{\hspace{-0.5em}\small QP22}} & 
\includegraphics[width=\linewidth]{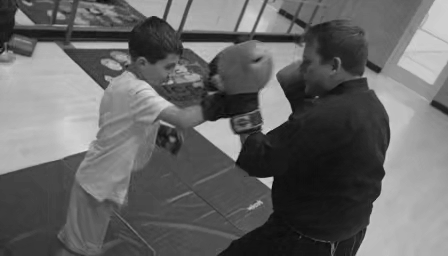} &
\includegraphics[width=\linewidth]{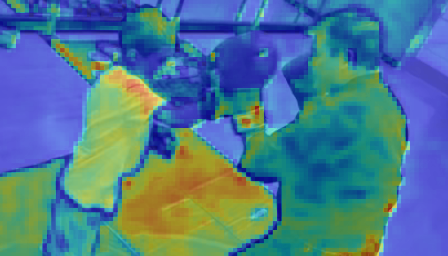} &
\includegraphics[width=\linewidth]{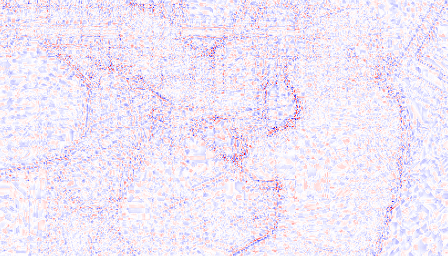} &
\includegraphics[width=\linewidth]{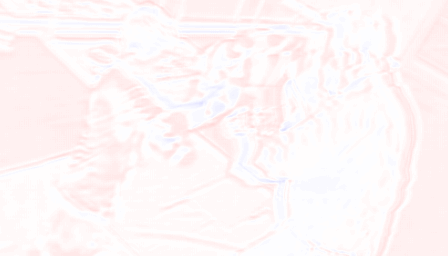} 
\\
\multirow{2}{*}{\rotatebox{90}{\hspace{-0.5em}\small QP32}} & 
\includegraphics[width=\linewidth]{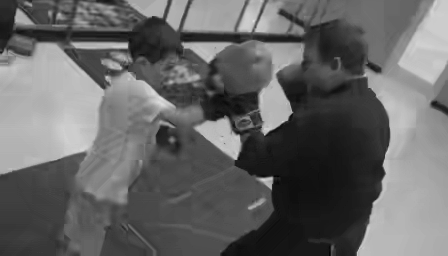} &
\includegraphics[width=\linewidth]{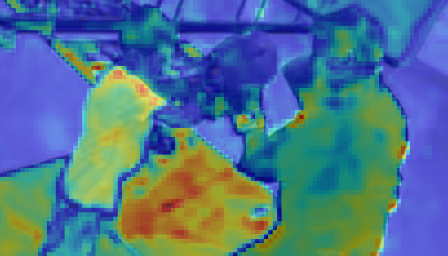} &
\includegraphics[width=\linewidth]{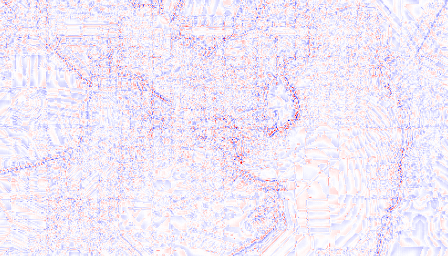} &
\includegraphics[width=\linewidth]{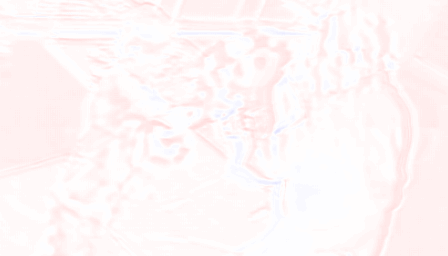} 
\\
& {\small Compressed LFR video} & {\small Compression-aware feature} & {\small Compensation (\#1)} & {\small Compensation (\#2)}   \\
\end{tabular}
\caption{\textbf{Visualization of compression-aware results}. 
From left to right: the compressed low-frame-rate video, the compression-aware feature $\mathbf{f}_{c}$, and the compensation residues produced by the restoration modules before (\#1) and after (\#2) upscaling. The visualization of $\mathbf{f}_{c}$ is obtained by averaging the absolute values across all channels. For the compensation residues, negative and positive values are represented by blue and red, respectively, with color intensity indicating their magnitude.
Our model tends to compensate for compression artifacts more actively after upscaling at QP 22 compared to QP 32.
}

\label{fig:viz_compression_aware_results}
\end{figure*}

To address the conflicting requirements of VQE models under different compression levels, we incorporate compression-aware features to guide the restoration of downscaled LFR videos. We first visualize the features generated by the compression encoder using two types of loss functions through t-SNE \cite{van2008visualizing}, as shown in Fig.~\ref{fig:viz_feature_tsne}. Specifically, the ranking loss improves the separability of latent features across different QPs compared to the L1 loss, demonstrating the effectiveness of the learning-to-rank strategy in capturing subtle compression distortions.
Furthermore, we show the compressed LFR videos, the compression-aware features $\mathbf{f}_{c} := \mathcal{E}_c(I_c^{LQ})$, and the compensation residues produced by the restoration modules before and after upscaling at QP of 22 and 32, as depicted in Fig.~\ref{fig:viz_compression_aware_results}. 
Specifically,  we average the absolute values of  $\mathbf{f}_{c}$ across all channels. The visualization results indicate that the compression encoder effectively captures distorted textures, such as carpet regions. Also, we present the compensation residues produced by the restoration modules before and after upscaling in Fig. \ref{fig:viz_compression_aware_results}. From the compensation results, we find that the model tends to compensate for compression artifacts more actively after upscaling at QP 22 compared to QP 32, consistent with the observations of Fig. 4 in the main text.

\subsection{Computation Details of Frequency Analysis}
\label{subsec: frequency_analysis}

\begin{figure}[t]
\centering
\begin{lstlisting}[language=Python]
def get_log_magnitude_spectrum(img):
    fft = np.fft.fft2(img)
    fft_shift = np.fft.fftshift(fft)
    magnitude = np.abs(fft_shift)
    # log(1 + magnitude)
    log_magnitude = np.log1p(magnitude)
    return log_magnitude
\end{lstlisting}
\caption{\textbf{Python-style code of calculation log-magnitude spectrum.}}
\label{fig:code_batch}
\end{figure}

In Fig. 5 of the main text, we visualize the log-magnitude spectrum of LFR frames generated by frame skipping, downscaling, and downscaling followed by upscaling on the SNU-FILM dataset. Here, we detail the frequency analysis procedure. To highlight differences in spectral content, we first subtract the original frame from each compressed LFR frame to suppress shared low-frequency components. The log-magnitude spectrum is then computed following the procedure shown in Fig.~\ref{fig:code_batch}.
Specifically, we compute the 2D discrete Fourier transform \(\mathscr{F}(x)\), shift the zero-frequency component to the center, take the magnitude, and apply a logarithmic transformation:
\begin{equation}
\mathscr{F}_{\log}(x) = \log(1 + |\text{FFTShift}(\mathscr{F}(x))|)
\end{equation}
The resulting log-magnitude spectra from different methods are subsequently used to compute histogram distributions and overlap ratios $\rho_\text{ovl}$ for quantitative comparison, as follows:
\begin{equation}
    \rho_\text{ovl} = \frac{A_t\cap A_f}{A_f}\times 100\%,
\end{equation}
where $A_f$ denotes the area of histogram distributions of frame skipping methods, $A_t$ indicates the area of histogram distributions of testing methods. 
A higher 
 $\rho_\text{ovl}$
  means that the spectrum of the LFR frames generated by the testing method is more similar to that of compressed LFR frames from the original sequence.

\subsection{Theoretical Justification}
\label{subsec:theoretical}
In the preliminaries of our methodology, we stated that for any continuous random variable with a well-defined density (e.g., \( \mathbf{x}_H \sim p(\mathbf{x}_H \mid \mathbf{x}_L) \)), there exists a bijective transformation \( f_{\mathbf{x}_H} \) such that \( f_{\mathbf{x}_H}(\mathbf{x}_H) \sim \mathcal{N}(0, \mathbf{I}) \)~\cite{hyvarinen1999nonlinear}. This result provides the theoretical foundation for using invertible neural networks in image and video rescaling. 
Here, we briefly present the constructive proof of  Lemma \ref{lemma:invertible_transform}.

\begin{lemma}
Let \( \mathbf{x} = (x_1, x_2, \ldots, x_n) \) be a random vector with a joint density \( p(\mathbf{x}) \) that is continuous, strictly positive, and continuously differentiable almost everywhere.  
Then there exists a bijective transformation \( f \) such that
\(
f(\mathbf{x}) \sim \mathcal{N}(0, \mathbf{I})
\).
\label{lemma:invertible_transform}
\end{lemma}

\begin{proof}
The construction of $f$ involves two successive bijective transformations.  First, we apply the Rosenblatt transform~\cite{rosenblatt1952remarks} to map \( \mathbf{x} = (x_1, \ldots, x_n) \) to \( \mathbf{y} = (y_1, \ldots, y_n) \in [0,1]^n \) via
\begin{equation}
\begin{aligned}
y_1 &= F_1(x_1), \\
y_2 &= F_2(x_2 \mid x_1), \\
&\;\vdots \\
y_n &= F_n(x_n \mid x_1, \ldots, x_{n-1}),
\end{aligned}
\end{equation}
where each \( F_i \) is the conditional CDF of \( x_i \) given \( x_1, \ldots, x_{i-1} \).
Under mild regularity conditions, specifically, continuity and strict monotonicity in each \( x_i \), this transform is invertible with
\begin{equation}
\begin{aligned}
x_1 &= F_1^{-1}(y_1), \\
x_2 &= F_2^{-1}(y_2 \mid x_1), \\
&\;\vdots \\
x_n &= F_n^{-1}(y_n \mid x_1, \ldots, x_{n-1}).
\end{aligned}
\end{equation}
This yields \( \mathbf{y} = g(\mathbf{x}) \sim \mathcal{U}[0,1]^n \), with independent components.
Next, applying the inverse standard Gaussian CDF element-wise gives
\begin{equation}
\mathbf{s} = \Phi^{-1}(\mathbf{y}) \sim \mathcal{N}(0, \mathbf{I}).
\end{equation}
Combining the two steps, we define the overall transformation
\begin{equation}
f = g^{-1} \circ \Phi,
\end{equation}
such that \( \mathbf{x} = f(\mathbf{s}) \), where \( \mathbf{s} \sim \mathcal{N}(0, \mathbf{I}) \). Since each step is invertible, the full transformation is bijective.
\end{proof}

\subsection{Upper Bound on Gradient Estimation Error}
\label{subsec:upper_bound}

To justify the validity of the approximation of the proposed surrogate network, we provide a theoretical analysis of the gradient estimation error bound in Theorem~\ref{thm:gradient_bound}. 
This bound connects the output space discrepancy to the gradient estimation error, thereby justifying how controlling the output difference guarantees optimization robustness.
Notably, although practical lossy codecs contain discrete operations that preclude differentiability, the following derivation analyzes a locally smoothed distortion response of the codec to provide theoretical insight into the surrogate-gradient approximation.

\begin{theorem}[Gradient Error Bound via Output-Space Discrepancy]
\label{thm:gradient_bound}
Let $\eta(\cdot)$ denote a function (representing the true video codec), and let $\phi(\cdot)$ denote another function (representing the surrogate network). Assume that within a local neighborhood $\mathcal{B}_\epsilon(x)=\{x+\delta:\|\delta\|_2\le\epsilon\}$, both functions have Lipschitz continuous gradients with Lipschitz constants $L_\eta$ and $L_\phi$, respectively. Then, the difference between the true gradient $\nabla \eta(x)$ and the surrogate gradient $\nabla \phi(x)$ is bounded by
\begin{equation}
\begin{aligned}
    \|\nabla \eta(x)-\nabla \phi(x)\|_2
    \le
    \frac{2D}{\epsilon}+\frac{(L_\eta+L_\phi)\epsilon}{2} \\
    \text{where }\ D=\sup_{\|\delta\|_2\le\epsilon}\big|\eta(x+\delta)-\phi(x+\delta)\big|.
\end{aligned}
\end{equation}
\end{theorem}
\begin{proof}
For any perturbation $\delta$ within the local neighborhood $\mathcal{B}_\epsilon(x)$, the Taylor expansions of the true codec response $\eta(\cdot)$ and the surrogate response $\phi(\cdot)$ can be written as
\begin{align}
\eta(x+\delta) &= \eta(x) + \nabla \eta(x)^\top \delta + R_\eta(\delta), \quad |R_\eta(\delta)| \leq \frac{L_\eta}{2}\|\delta\|_2^2, \\
\phi(x+\delta) &= \phi(x) + \nabla \phi(x)^\top \delta + R_\phi(\delta), \quad |R_\phi(\delta)| \leq \frac{L_\phi}{2}\|\delta\|_2^2.
\end{align}
Subtracting the two expansions gives
\begin{equation}
\begin{aligned}
    [\nabla \eta(x)-\nabla \phi(x)]^\top \delta
    =
    [\eta(x+\delta)-\phi(x+\delta)] &
     - [\eta(x)-\phi(x)]\\
    &- [R_\eta(\delta) - R_\phi(\delta)].
\end{aligned}
\end{equation}
Since $\|\delta\|_2\le\epsilon$, both $|\eta(x+\delta)-\phi(x+\delta)|$ and $|\eta(x)-\phi(x)|$ are bounded by $D$. Therefore,
\begin{equation}
\begin{aligned}
    \left|[\nabla \eta(x)-\nabla \phi(x)]^\top \delta\right|
    &\le
    \left|\eta(x+\delta)-\phi(x+\delta)\right|
    + \left|\eta(x)-\phi(x)\right| \\
     & \quad \quad + |R_\eta(\delta)| + |R_\phi(\delta)| \\
    &\le
    2D + \frac{L_\eta+L_\phi}{2}\epsilon^2.
\label{eq:ineq_bound}
\end{aligned}
\end{equation}
The above inequality holds for any valid $\delta$. To recover the norm of the gradient difference, we construct $\delta^*$ to be aligned with $\nabla \eta(x)-\nabla \phi(x)$, namely
\begin{equation}
    \delta^*=
    \epsilon\frac{\nabla \eta(x)-\nabla \phi(x)}
    {\|\nabla \eta(x)-\nabla \phi(x)\|_2}.
\end{equation}
Substituting $\delta^*$ into Eq.~\eqref{eq:ineq_bound} yields
\begin{equation}
    \epsilon \|\nabla \eta(x)-\nabla \phi(x)\|_2
    \le
    2D+\frac{L_\eta+L_\phi}{2}\epsilon^2.
\end{equation}
Dividing both sides by $\epsilon$ gives
\begin{equation}
    \|\nabla \eta(x)-\nabla \phi(x)\|_2
    \le
    \frac{2D}{\epsilon}+\frac{(L_\eta+L_\phi)\epsilon}{2}.
\end{equation}
\end{proof}

Theorem \ref{thm:gradient_bound} provides a theoretical justification for optimization robustness. Since our goal is to ensure that the surrogate gradient $\nabla \phi(x)$ closely approximates the true gradient $\nabla \eta(x)$, this gradient difference is controlled by the output discrepancy $D$ and the local smoothness constants. By minimizing the $L_1$ loss between the true codec and the surrogate network during training, we explicitly minimize $D$. Consequently, as $D \to 0$, the upper bound of the gradient difference converges to a small value dependent on the Lipschitz constants and the perturbation radius $\epsilon$. This shows that minimizing the observational loss inherently reduces the upper bound on the local gradient mismatch, thereby supporting the use of the surrogate network as a practical gradient estimator.

\begin{figure}
    \centering
    \includegraphics[trim=0cm 0cm 0cm 0cm, width=0.95\linewidth]{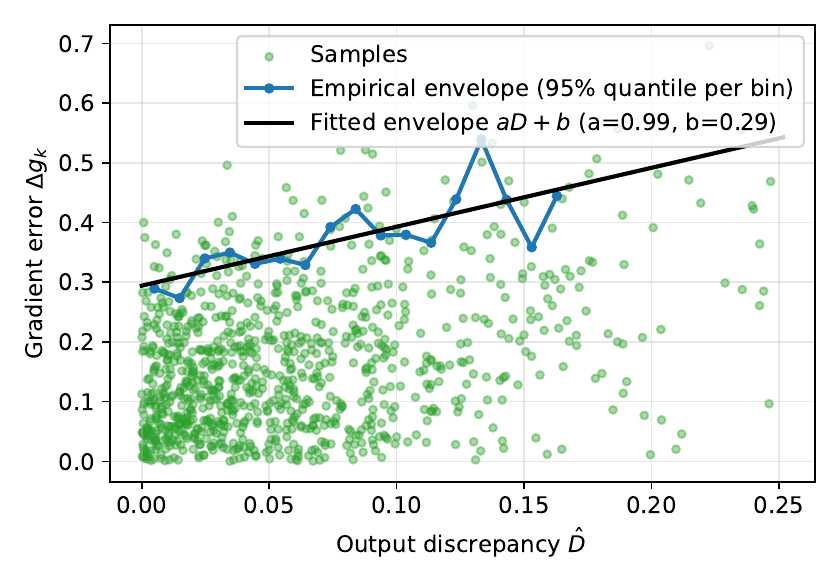}
    \caption{Empirical validation of Theorem~\ref{thm:gradient_bound}. Scatter plot of the gradient error \(\Delta g_k=\big|\nabla_{\text{num}}(\mathbf{x};\mathbf{u}_k) - \nabla_{\text{surr}}(\mathbf{x};\mathbf{u_k})\big|\) versus the output discrepancy $\hat{D}(\mathbf{x})$. The blue curve indicates 95\% quantile per bin, and the black curve shows the fitted upper envelope $0.99 {D} +0.29$.}
    \label{fig:exp2_empirical_envelope}
    \vspace{15pt} 
\end{figure}



To empirically validate Theorem~\ref{thm:gradient_bound}, we explore how the surrogate mismatch in the output space correlates with the gradient estimation error. 
For each video clip, we estimate the numerical gradient of $\eta(\cdot)$ using  finite differences. Specifically, we sample $K$ random unit directions $\{\mathbf{u_k}\}_{k=1}^{K}$ with $\|\mathbf{u}_k\|_2=1$ and compute
\begin{equation}
    \nabla_{\text{num}}(\mathbf{x};\mathbf{u}_k)=\frac{\eta(\mathbf{x}+\epsilon\mathbf{u}_k)-\eta(\mathbf{x}-\epsilon\mathbf{u}_k)}{2\epsilon}.
\end{equation}
Meanwhile, we compute the surrogate directional derivative via backpropagation:
\begin{equation}
    \nabla_{\text{surr}}(\mathbf{x};\mathbf{u_k})=\nabla_{\mathbf{x}}\phi(\mathbf{x})^\top \mathbf{u}_k,
\end{equation}
and measure the gradient error \(\Delta g_k=\big|\nabla_{\text{num}}(\mathbf{x};\mathbf{u}_k) - \nabla_{\text{surr}}(\mathbf{x};\mathbf{u_k})\big|\). To obtain an observable proxy of \(D\) in Theorem~\ref{thm:gradient_bound}, we estimate the local output discrepancy in the same neighborhood through random sampling:
\begin{equation}
    \widehat{D}(\mathbf{x})=\max_{k=1,\dots,K}\big|\eta(\mathbf{x}+\epsilon\mathbf{u}_k)-\phi(\mathbf{x}+\epsilon\mathbf{u}_k)\big|.
\end{equation}
We then plot \(\Delta g_k\) against \(\widehat{D}\) and fit an empirical upper envelope using the form \(a{\widehat{D}}+b\), which aligns with the \(\mathcal{O}({D})\) dependence predicted by Theorem~\ref{thm:gradient_bound}. As shown in Fig.~\ref{fig:exp2_empirical_envelope}, the fitted envelope basically follows this theoretical trend, supporting our interpretation that reducing the surrogate output mismatch also reduces the gradient mismatch.

Furthermore, to directly assess how surrogate-gradient quality affects the final RD performance, we conduct a controlled degradation study by injecting Gaussian noise into the surrogate gradients backpropagated to the downscaler during training. Specifically, let \(\mathbf{g}\) denote the gradient signal passed from the surrogate branch to the downscaler parameters. We perturb it as
\begin{equation}
    \widetilde{\mathbf{g}}=(1-\alpha)\mathbf{g}+\alpha\cdot \sigma(\mathbf{g})\cdot \boldsymbol{\xi}, \quad \boldsymbol{\xi}\sim\mathcal{N}(\mathbf{0},\mathbf{I}),
    \label{equ:bd_degradation}
\end{equation}
where \(\alpha\in[0,1]\) controls the mixing weight, and \(\sigma(\mathbf{g})=\mathrm{RMS}(\mathbf{g})\) (equivalently \(\|\mathbf{g}\|_2/\sqrt{|\mathbf{g}|}\)) normalizes the noise magnitude to match the scale of the original gradient. This design ensures that increasing \(\alpha\) progressively reduces the signal-to-noise ratio of the optimization direction, approaching a near-random gradient signal when \(\alpha\to 1\).
\begin{table}
    \centering
    \caption{BDBR results with different noise mixing weights ($\alpha$).}
    \label{tab:bdbr_degradation}
    \resizebox{0.5\linewidth}{!}{%
        \begin{tabular}{c|c|c}
            \hline
            $\alpha$ & BDBR (PSNR) & BDBR (SSIM) \\
            \hline \hline
            0.0 & -12.38 & -10.62 \\
            0.1 & -6.59 & -5.84 \\
            0.2 & -5.42 & -4.15 \\
            0.3 & -1.91 & -1.23 \\
            0.4 & 3.65 & 5.15 \\
            \hline
        \end{tabular}%
    }
    \vspace{-10pt}
\end{table}
As shown in Table~\ref{tab:bdbr_degradation}, the BDBR performance consistently degrades as \(\alpha\) increases, confirming that accurate surrogate gradients are crucial for end-to-end optimization under non-differentiable codecs and providing empirical support for the robustness analysis implied by Theorem~\ref{thm:gradient_bound}.

\section{Detailed Network Architecture}
\label{sec: supp_network}

\subsection{Transformation Modules}
\label{subsec: transformation_module}

Following prior designs \cite{huang2021video,ho2022video}, the coupling transformation modules $\varphi(\cdot)$, $\eta(\cdot)$, $\rho(\cdot)$, $\gamma(\cdot)$, and $\beta(\cdot)$ in MIMO-VRN, as well as the update module in MIMO-TWT, are implemented using DenseBlock \cite{DenseNet2017}, as shown in Fig.~\ref{fig:denseblock}. The DenseBlock utilizes direct connections from each layer to all subsequent layers. As a result, the $\ell$-th layer receives the feature maps from all preceding layers, $\f_0, \f_1, ..., \f_{\ell-1}$, as input:
\begin{equation}
    \f_\ell = H_\ell([\f_0, \f_1, ..., \f_{\ell-1}]),
\end{equation}
where $[\f_0, \f_1, ..., \f_{\ell-1}]$ denotes the concatenation of feature maps from layers $0$ to $\ell-1$.  
If each transformation module $H_\ell(\cdot)$ produces $k$ feature maps, then the $\ell$-th layer receives $k_0 + k \cdot (\ell-1)$ input channels, where $k_0$ is the number of channels in the input layer. Here, $k_0$ is set to $24$ in our work.

\begin{figure}[t]
\centering
\includegraphics[width=0.95\columnwidth]{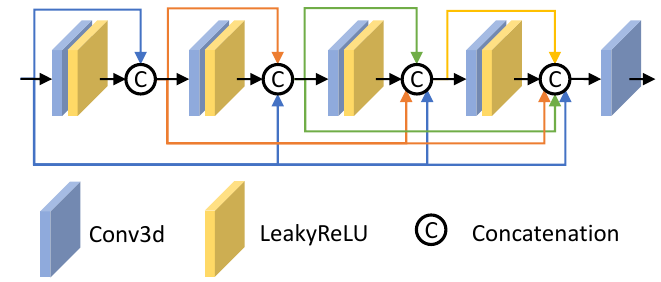}
\caption{\textbf{Detailed Structure of DenseBlock.}}
\label{fig:denseblock}
\end{figure}

\begin{figure}[t]
\centering
\includegraphics[width=0.95\columnwidth]{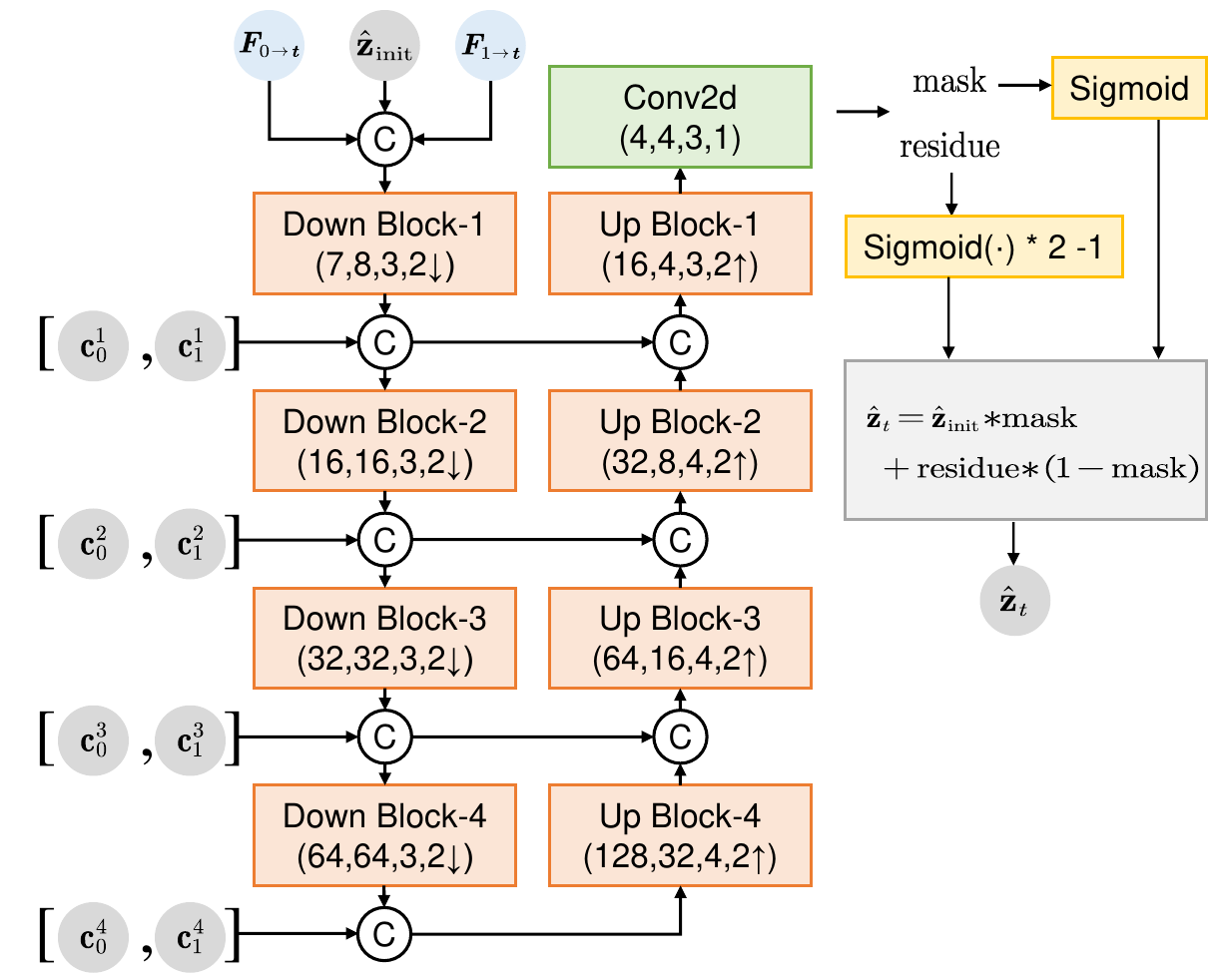}
\caption{\textbf{Detailed Structure of Context-aware U-Net.}}
\label{fig:unet}
\end{figure}

\begin{figure*}[t]
\centering
\includegraphics[width=1.6\columnwidth]{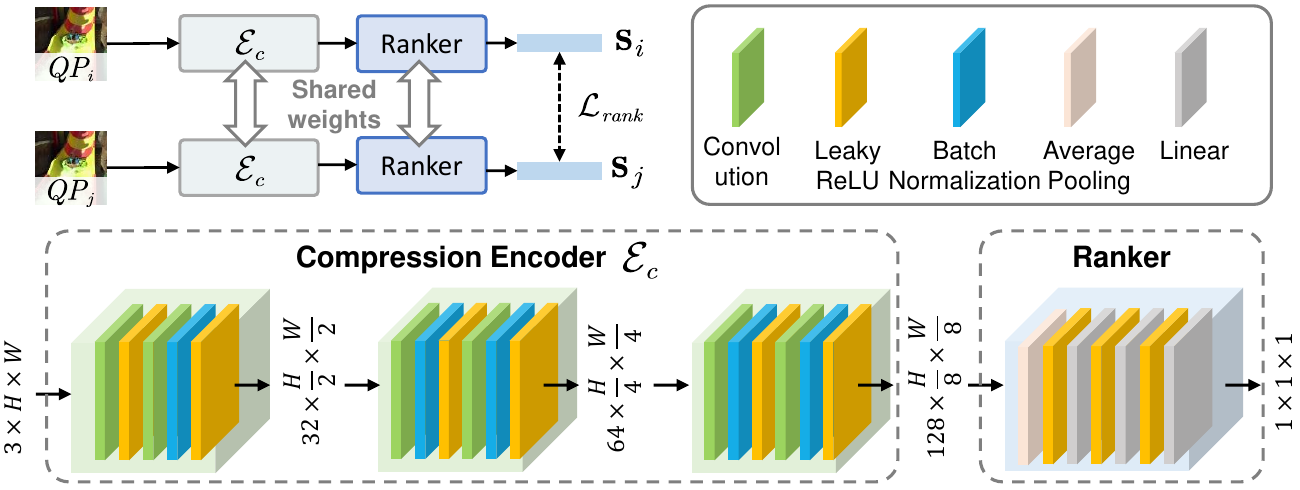}
\caption{\textbf{Detailed Structure of Compression Encoder and Ranker.}}
\label{fig:compression_encoder}
\end{figure*}

\subsection{Context-Aware U-Net}
\label{subsec: context_unet}

We design a context-aware U-Net architecture to hierarchically incorporate multi-scale contextual features into reconstructed high-frequency components, as shown in Fig. \ref{fig:unet}. The network consists of four downsampling blocks and four corresponding upsampling blocks, with skip connections between the encoder and decoder layers at the same resolution. 

Each downsampling block begins with a strided convolution for spatial reduction, followed by one or more convolutional layers with PReLU activations \cite{he2015delving}. As shown in Table \ref{tab: unet}, the first downsampling block employs multiple stacked convolutions to enhance low-level feature extraction, while subsequent blocks adopt a lighter structure with two convolutional layers.

The upsampling path mirrors the encoder, with each block consisting of a transposed convolution to restore spatial resolution, followed by PReLU activation. After upsampling, a single convolutional layers are applied to refine the output. Specifically, the first three channels of the final convolutional layer output are used as the residue \(\mathbf{R}\), and the last channel is used as the mask \(\mathbf{M}\). Both are normalized using the sigmoid function:
\begin{equation}
    \mathbf{M}' = \text{Sigmoid}(\mathbf{M}), \quad \mathbf{R}' = 2 \times \text{Sigmoid}(\mathbf{R}) - 1.
\end{equation}
The final reconstructed high-frequency components \(\hat{\mathbf{z}}_t\) are obtained via element-wise weighted summation:
\begin{equation}
    \hat{\mathbf{z}}_t = \hat{\mathbf{z}}_\text{init} \cdot \mathbf{M}' + \mathbf{R}' \cdot (1 - \mathbf{M}').
\end{equation}

\subsection{Lightweight Bi-directional Optical Flow Network}
\label{subsec: lightweight_op_net}

\setcounter{table}{5}
\begin{table}[t]
\centering
\caption{Configuration of the context-aware U-Net.\\ ``In" and ``Out" indicate the input and output channel dimensions, respectively. ``K" denotes the kernel size, ``S" the stride, and ``P" the padding.}
\renewcommand{\arraystretch}{1.2}
\label{tab: unet}
\resizebox{0.48\textwidth}{!}{%
\begin{tabular}{|l|l|c|c|c|c|c|}
\hline
\textbf{Module} & \textbf{Layer Type} & \textbf{In} & \textbf{Out} & \textbf{K} & \textbf{S} & \textbf{P} \\
\hline
\multirow{4}{*}{Down Block-1} 
 & Conv2d + PReLU & 7 & 8 & 3×3 & 2 & 1 \\
 & Conv2d + PReLU & 8 & 8 & 3×3 & 1 & 1 \\
 & Conv2d + PReLU & 8 & 8 & 3×3 & 1 & 1 \\
 & Conv2d + PReLU & 8 & 8 & 3×3 & 1 & 1 \\
\hline
\multirow{2}{*}{Down Block-2}
 & Conv2d + PReLU & 16 & 16 & 3×3 & 2 & 1 \\
 & Conv2d + PReLU & 16 & 16 & 3×3 & 1 & 1 \\
\hline
\multirow{2}{*}{Down Block-3}
 & Conv2d + PReLU & 32 & 32 & 3×3 & 2 & 1 \\
 & Conv2d + PReLU & 32 & 32 & 3×3 & 1 & 1 \\
\hline
\multirow{2}{*}{Down Block-4}
 & Conv2d + PReLU & 64 & 64 & 3×3 & 2 & 1 \\
 & Conv2d + PReLU & 64 & 64 & 3×3 & 1 & 1 \\
\hline
\multirow{3}{*}{Up Block-1}
 & ConvTrans2d + PReLU & 16 & 4 & 4×4 & 2 & 1 \\
 & Conv2d + PReLU & 4 & 4 & 3×3 & 1 & 1 \\
 & Conv2d + PReLU & 4 & 4 & 3×3 & 1 & 1 \\
\hline
Up Block-2 & ConvTrans2d + PReLU & 128 & 32 & 4×4 & 2 & 1 \\
\hline
Up Block-3 & ConvTrans2d + PReLU & 64 & 16 & 4×4 & 2 & 1 \\
\hline
Up Block-4 & ConvTrans2d + PReLU & 32 & 8 & 4×4 & 2 & 1 \\
\hline
\end{tabular}
}
\end{table}

\begin{table*}[]
\centering
\caption{\textbf{Quantitative ($\text{PSNR}_\text{avg. mse}$/SSIM) comparisons of various methods on
UCF101~\cite{soomro2012ucf101}, Vimeo90K~\cite{xue2019video} and
SNU-FILM~\cite{choi2020channel} benchmarks using the rounding-based quantization as compression. }\\
 {\textbf{BOLD}}: best
performance, {\underline{UNDERLINE}}: second best performance.}
\begin{tabular}{{|l|l||c|c|c|c|c|c|}}
\hline
\multirow{2}{*}{Downscaler} & \multirow{2}{*}{Upscaler} & \multirow{2}{*}{UCF101} &
\multirow{2}{*}{Vimeo90K} & \multicolumn{4}{c|}{SNU-FILM}   \\
\cline{5-8}
& & &  &  easy & medium & hard & extreme \\
\hline
\hline 
 \multirow{13}{*}{\begin{tabular}[c]{@{}c@{}}Frame\\      Skipping\end{tabular}}
 & AdaCoF \cite{lee2020adacof} & 35.34/0.9815 & 33.96/0.9705 & 36.25/0.9839 & 35.74/0.9778 & 33.77/0.9782 & 31.02/0.9531
                           \\
 & BMBC \cite{park2020bmbc} & 35.58/0.9814 & 34.50/0.9711 & 36.72/0.9837 & 36.01/0.9785 & 33.63/0.9691 & 30.63/0.9528
                       \\
 & CDFI \cite{ding2021cdfi} & 35.65/0.9754 & 34.66/0.9789 & 36.91/0.9817 & 36.20/0.9785 & 34.03/0.9693 & 31.24/0.9543 
                       \\
 & ABME \cite{park2021asymmetric}& 35.82/0.9820 & 35.67/0.9820 & 37.41/0.9840 & 36.46/0.9790 & 34.89/0.9722 & 32.24/0.9543
                              \\
 & XVFI$_v$ \cite{sim2021xvfi} & 35.65/0.9673 & 34.56/0.9788 & 36.98/0.9800 & 36.06/0.9740 & 34.21/0.9579 & 31.43/0.9312   \\
 & EBME-H* \cite{jin2022enhanced} & 35.89/0.9820 & 35.58/0.9822 & 37.48/0.9825 & 36.92/0.9816 & 35.02/0.9746 & 32.11/0.9592
                          \\
 & VFIformer \cite{lu2022video} & 40.20/{\textbf{0.9900}}
                  & 41.27/0.9939
                                   & 44.90/0.9969 & 40.86/0.9933
                                   & 35.44/0.9793
                                   & 30.20/0.9548 
                               \\
 & IFRNet \cite{kong2022ifrnet}
                                   & 40.19/{\underline{0.9899}}
                                   & 40.97/0.9936 & 44.87/0.9969
                                   & 40.89/0.9932
                                   & 35.40/0.9789
                                   & 30.04/0.9536  \\
 & UPR-Net \cite{jin2023unified} & 40.18/{\underline{0.9899}} & 40.80/0.9934
         & 45.14/0.9970 &40.93/0.9932 &35.44/0.9788 &30.26/0.9542
                 \\
 & UPR-Net l \cite{jin2023unified} &  40.20/{\textbf{0.9900}} 
               & 41.05/0.9937 &  45.19/0.9970
                  & 41.01/0.9933
                  & 35.58/0.9790
                  & 30.35/0.9545
                \\
 & UPR-Net L \cite{jin2023unified} &  40.24/{\textbf{0.9900}} 
                  & 41.19/0.9938
                  &  {\underline{45.21}}/{\underline{0.9970}}
                  &  {\underline{41.06}}/0.9934
                  &  35.63/0.9792
                  &  30.40/0.9547
                \\
 & EMA-VFI \cite{zhang2023extracting} &  {{40.24}}/{\textbf{0.9900}}
                  & {\underline{41.41}}/{\underline{0.9939}}
                  &  44.75/0.9940
                  &  40.86/{\underline{0.9970}}
                  &  {{35.71}}/{\underline{0.9934}}
                  &  {{30.46}}/{{0.9554}}
             \\
 & GIMM-VFI \cite{guo2024generalizable} & \underline{40.25}/\textbf{0.9900} & 40.59/0.9930 & 44.96/0.9969 & 40.93/0.9933 & \underline{35.76}/0.9802 & \underline{30.72}/\underline{0.9576}
 \\ \hline

\multicolumn{2}{|c|}{TVRN}     &   {\textbf{49.05}}/{\textbf{0.9900}}
                  & {\textbf{47.50}}/{\textbf{0.9977}} 
                  &  {\textbf{51.01}}/{\textbf{0.9989}} 
                  &  {\textbf{48.97}}/{\textbf{0.9983}} 
                  &  {\textbf{45.57}}/{\textbf{0.9966}} 
                  &  {\textbf{41.36}}/{\textbf{0.9905}} 
           \\
\hline
\end{tabular}

\label{tab:sota_lossless}
\end{table*}

\begin{figure*}[]
\centering
\includegraphics[width=0.85\textwidth]{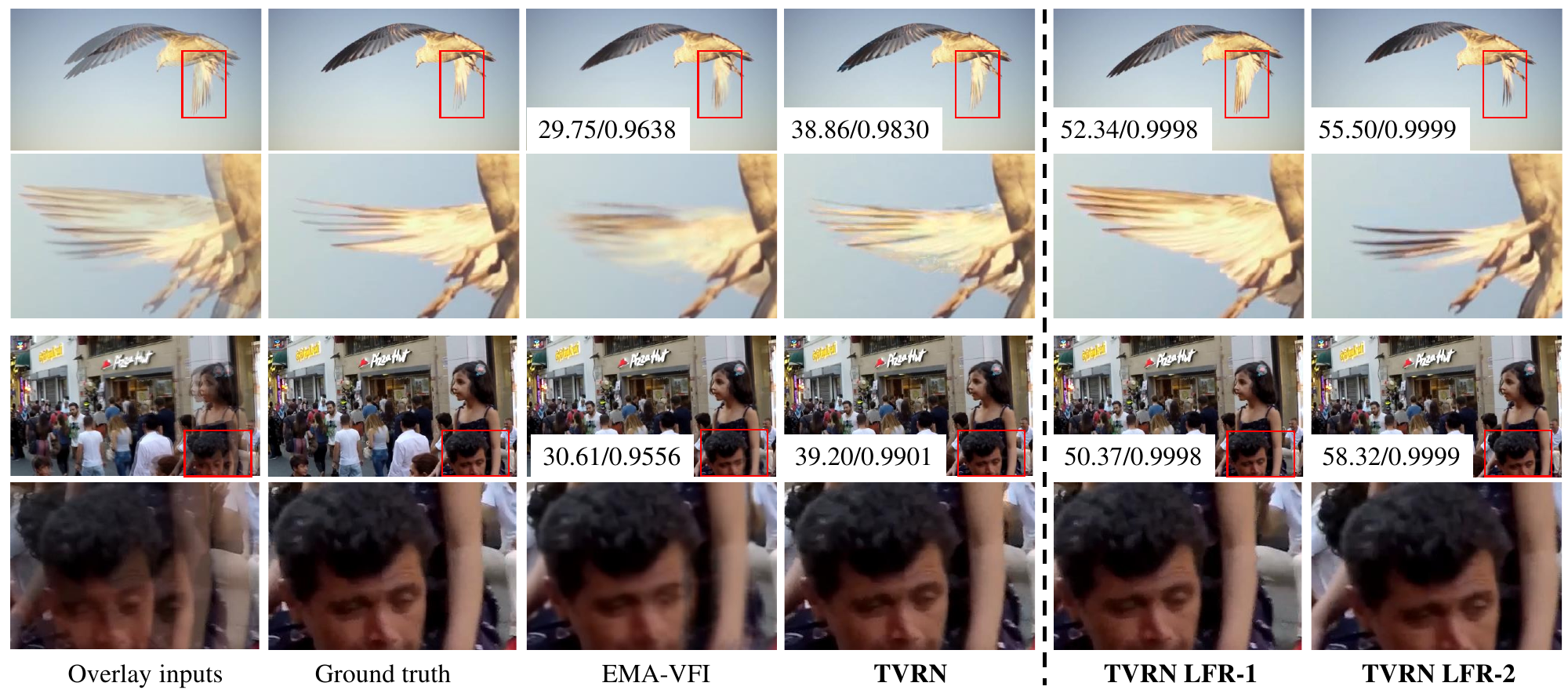}
\caption{
    \textbf{Qualitative comparisons on SNU-FILM datasets~\cite{choi2020channel} with the rounding-based quantization.}
}
\label{fig:vis-snufilm_lossless}
\end{figure*}

We adopt the lightweight Bi-directional Motion Estimation module proposed in~\cite{jin2022enhanced}, which estimates optical flow between two input images \(\mathbf{I}_0, \mathbf{I}_1 \in \mathbb{R}^{3 \times H \times W}\) via a coarse-to-fine strategy with \(L\) pyramid levels. At each level \(\ell\), the input images are downsampled by a factor of \(s_\ell = 2^{-\ell}\), yielding \(\mathbf{I}_0^{(\ell)}\) and \(\mathbf{I}_1^{(\ell)}\).
Specifically, a two-stage feature pyramid extracts multi-scale representations \(\mathbf{C}_1^{(\ell)}, \mathbf{C}_2^{(\ell)} \in \mathbb{R}^{2c \times \frac{H}{4 s_\ell} \times \frac{W}{4 s_\ell}}\) from each warped image, where \(c=24\). To initialize estimation at level \(\ell\), the flow \(\mathbf{F}^{(\ell+1)}\) and hidden features \(\mathbf{H}^{(\ell+1)}\) from the coarser level are upsampled and scaled:
\begin{equation}
    \tilde{\mathbf{F}}^{(\ell)} = 2 \cdot \text{Upsample}(\mathbf{F}^{(\ell+1)}), \quad 
    \tilde{\mathbf{H}}^{(\ell)} = \text{Upsample}(\mathbf{H}^{(\ell+1)}).
\end{equation}
Afterwards, motion compensation is performed using a middle-oriented forward warping:
\begin{equation}
    \hat{\mathbf{I}}_0^{(\ell)} = \mathcal{W}\big(\mathbf{I}_0^{(\ell)}, t \cdot \tilde{\mathbf{F}}^{(\ell)}_{1:2}\big), \quad 
    \hat{\mathbf{I}}_1^{(\ell)} = \mathcal{W}\big(\mathbf{I}_1^{(\ell)}, (1-t) \cdot \tilde{\mathbf{F}}^{(\ell)}_{3:4}\big).
\end{equation}
Then, the warped images are passed through the feature extractor to obtain \(\mathbf{C}_1^{(\ell)}, \mathbf{C}_2^{(\ell)}\), which, together with \(\tilde{\mathbf{F}}^{(\ell)}\) and \(\tilde{\mathbf{H}}^{(\ell)}\), are used by the flow estimator \(\mathcal{G}_{op}(\cdot)\) to produce a residual flow and updated hidden state:
\begin{equation}
    \Delta \mathbf{F}^{(\ell)}, \mathbf{H}^{(\ell)} = \mathcal{G}_{op}\big(\mathbf{C}_2^{(\ell)}, \mathbf{C}_2^{(\ell)}, \tilde{\mathbf{H}}^{(\ell)}, \tilde{\mathbf{F}}^{(\ell)}\big).
\end{equation}
The flow is then refined via residual connection:
\begin{equation}
    \mathbf{F}^{(\ell)} = \tilde{\mathbf{F}}^{(\ell)} + \Delta \mathbf{F}^{(\ell)}.
\end{equation}
Starting from zero initialization at the coarsest level, the network iteratively refines flow estimates across all pyramid levels. Finally, the final flow \(\mathbf{F}^{(0)}\) is upsampled by a factor of 4 to match the input resolution.

\begin{figure*}[]
\centering
\includegraphics[width=0.85\textwidth]{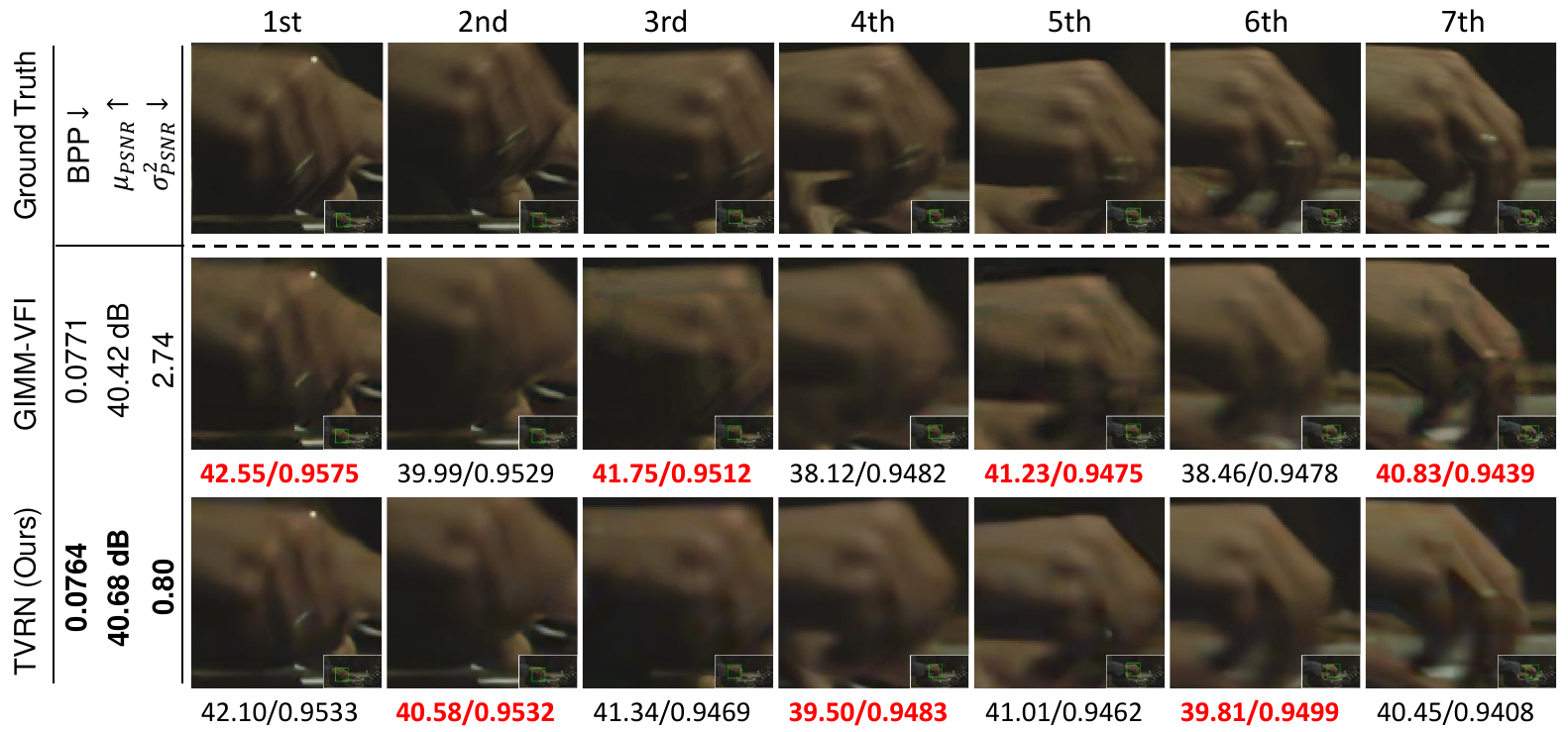}
\caption{
    \textbf{Qualitative comparisons of upscaled HFR videos on vimeo90K Septuplet test dataset \cite{xue2019video} with the lossy H.265 codec.}
}
\label{fig:vis-upscaled}
\end{figure*}

\begin{figure*}[]
\centering
\includegraphics[width=0.85\textwidth]{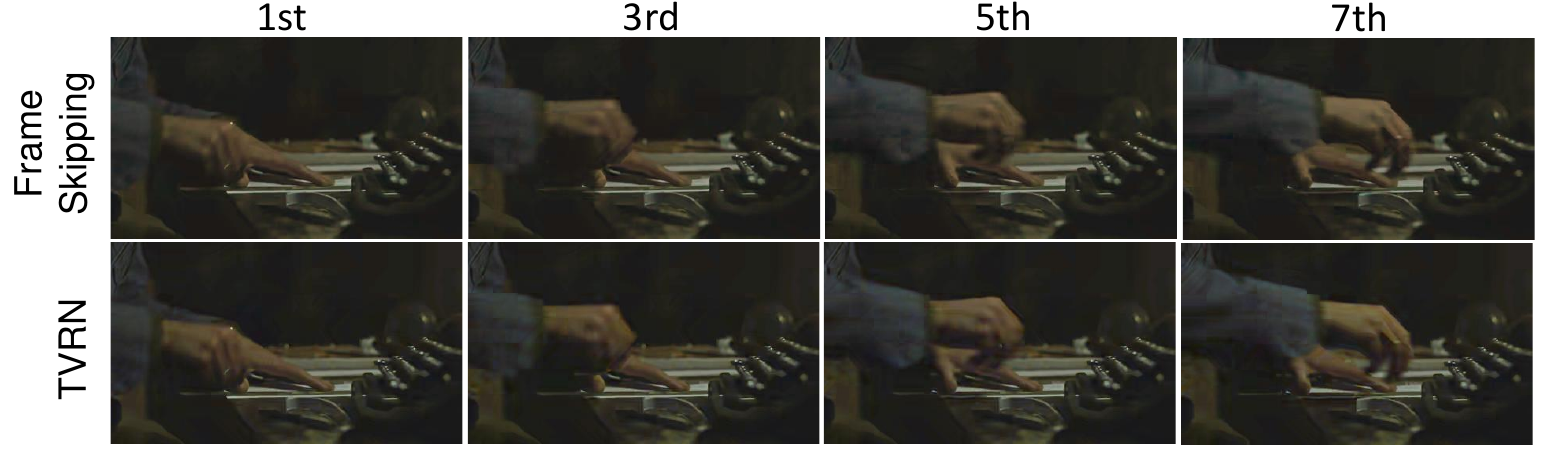}
\caption{
    \textbf{Qualitative comparisons of downscaled LFR videos on vimeo90K Septuplet test dataset \cite{xue2019video} with the lossy H.265 codec.}
}
\label{fig:vis-downscaled}
\end{figure*}

\subsection{Compression Encoder and Ranker}
\label{subsec: compression_encoder}

In Fig. \textcolor{red}{6}, we extract multi-level compression-aware features from compressed input frames using a pretrained compression encoder. Fig.~\ref{fig:viz_feature_tsne} further visualizes these features under different loss functions.
Here, we provide details on the learning-to-rank framework, as well as the architecture of the compression encoder and the ranker. As shown in Fig.~\ref{fig:compression_encoder}, we adopt a Siamese network composed of two identical branches, each containing a compression encoder and a ranker with shared parameters. Given a compressed frame $\mathbf{x} \in \mathbb{R}^{3 \times H \times W}$ at a specific QP, the encoder $\mathcal{E}_c$ outputs a scalar score $\mathbf{z} \in \mathbb{R}$ that reflects its compression level.
The compression encoder consists of three residual blocks, producing multi-level compression-aware features $\{\mathbf{f}_c^{(\ell)} \in \mathbb{R}^{(32\ell) \times \frac{H}{2^{\ell}} \times \frac{W}{2^{\ell}}}\}_{\ell=1}^{3}$. These features are later fused with aligned contextual features $\mathbf{f}_d\in\mathbb{R}^{C\times H\times W}$ from the enhancement module via a context-aware U-Net to generate the weight map $\mathbf{w}_q$ in Fig. \textcolor{red}{6}. During pretraining, the feature $\mathbf{f}c^{(2)}$ is passed to the ranker, which includes global average pooling followed by three fully connected layers, to yield a scalar score $\mathbf{s}$. By maximizing the score difference between frames with varying compression levels using the ranking loss $\mathcal{L}_{\text{rank}}$, the compression encoder is encouraged to capture subtle variations in compression quality, as evidenced by Fig. \ref{fig:viz_compression_aware_results}. In our implementation, the quantization parameter is randomly sampled from the list $[17, 22, 27, 32, 37]$.

\section{Experiments with Rounding-Based Quantization} 
\label{sec: lossless_performance}

In this section, we present additional comparison results with prior methods that adopt rounding-based quantization to approximate the non-differentiable encoder, following the default settings commonly used in rescaling approaches~\cite{xiao2020invertible, huang2021video, tian2021self, yang2023self, tian2023clsa}. Although this compression strategy is not used in actual deployment, previous works often simplify the compression process using this setup. Therefore, we include these experimental results in the appendix for reference.

\subsection{Experimental Setup}

\label{subsec: experiment_setup}
In this study, we use rounding-based quantization as a substitute for the surrogate network and present the reconstruction results in Table \ref{tab:sota_lossless} and the visual effects in Fig. \ref{fig:vis-snufilm_lossless}. The results indicate that employing a rounding-based quantization as a compression method significantly enhances reconstruction quality, with minimal impact on LFR videos.
It is important to note that since the reference frame undergoes only a rounding-based quantization for frame-skipping methods,  the commonly used logarithmic average PSNR metric:
\begin{equation}
    \mathrm{PSNR}_{\text{avg.log}}(\mathbf{x}, \hat{\mathbf{x}}) = \frac{1}{n} \sum_{i=0}^{n-1} 10\log_{10} \frac{\text{MAX}^2}{\text{MSE}(\mathbf{x}_i, \hat{\mathbf{x}}_i)},
\end{equation}
yields theoretically infinite values for non-interpolated frames where \(MSE(\mathbf{x}_i, \hat{\mathbf{x}}_i) = 0\). To avoid this issue, we instead compute PSNR based on the average MSE across all frames:
\begin{equation}
    \mathrm{PSNR}_{\text{avg.mse}}(\mathbf{x}, \hat{\mathbf{x}}) = 10\log_{10} \frac{\text{MAX}^2}{\frac{1}{n} \sum_{i=0}^{n-1} \text{MSE}(\mathbf{x}_i, \hat{\mathbf{x}}_i)},
\end{equation}
where \(n\) is the number of frames, and \(\mathbf{x}\) and \(\hat{\mathbf{x}}\) denote the ground truth and predicted frame sequences, respectively.
For consistency with the VFI evaluation configurations, we adopt a 3-frame setting. In this case, reported PSNR values from existing VFI methods, which are typically measured only on the interpolated frame, can be aligned with \(\mathrm{PSNR}_{\text{avg.log}}\) by adding a constant \(10\log_{10}3 \approx 4.77\) dB.

\subsection{Experimental Results}
\label{subsec: experimental_results}

The quantitative results are presented in Table~\ref{tab:sota_lossless}, and the visual comparisons with frame-skipping methods are shown in Fig.~\ref{fig:vis-snufilm_lossless}. From these results, we draw two key observations:

\begin{itemize}
    \item \textbf{Generative models, \textit{e.g.}, GIMM-VFI~\cite{guo2024generalizable}, basically unperforms under rounding-based quantization in terms of objective reconstruction quality, compared to VFI methods optimized for distortion metrics.} This is expected, as GIMM-VFI is designed for generalizable implicit motion modeling rather than precise reconstruction. However, under realistic compression scenarios, GIMM-VFI significantly outperforms state-of-the-art VFI methods~\cite{zhang2023extracting} in terms of PSNR and SSIM, as shown in the main text. Given the lack of prior research focused on video frame generation under compressed conditions, we believe this phenomenon is worthwhile to investigate further, specifically, {why generative models demonstrate better robustness to complex compression conditions compared to distortion-optimized methods}.

    \item \textbf{Our method successfully hides high-quality motion cues within LFR frames with minimal visual degradation under rounding-based quantization.} As shown in Fig.~\ref{fig:vis-snufilm_lossless}, the two LFR frames achieve PSNR scores above 50~dB, indicating that the distortion introduced by downscaling is visually negligible. Nevertheless, this subtle motion embedding yields significant benefits for final reconstruction, particularly in preserving structural details such as feather edges and contours.
\end{itemize}


\section{More Visualization Results}
\label{sec: more_viz}


We first present a qualitative comparison of a specific sequence from the Vimeo90K septuplet test dataset \cite{xue2019video}. Fig.~\ref{fig:vis-upscaled} and Fig.~\ref{fig:vis-downscaled} show the upscaled HFR reconstructions and the corresponding downscaled LFR references, respectively. Compared to frame-skipping methods based on GIMM-VFI \cite{guo2024generalizable}, our proposed approach sacrifices some reference frame quality to improve interpolation performance. This trade-off leads to not only a higher overall reconstruction quality (\textbf{40.68}~dB vs. 40.42~dB in frame-wise PSNR average), but also significantly reduced temporal quality fluctuation in the HFR reconstruction (\textbf{0.80} vs. 2.74 in frame-wise PSNR variance). In addition, the downscaled LFR frames show negligible perceptual differences from the original frames, supporting the effectiveness of our motion steganography design.

To supplement the visual results presented in the main text, we further provide more qualitative comparisons. On one hand, Fig.~\ref{fig:comp_hfr_selected} showcases the reconstruction results of interpolated frames from other representative sequences. On the other hand, Fig.~\ref{fig:vimeo_test_1}, \ref{fig:vimeo_test_2}, \ref{fig:vimeo_test_3}, and \ref{fig:vimeo_test_4} present the reconstruction HFR frames and downscaled LFR frames of the full sequences.

\begin{figure*}[!t]
\setlength{\tabcolsep}{0.5pt}
\centering
\label{fig: vimeo_test_1}

\resizebox{1.0\textwidth}{!}{
\scriptsize
\begin{tabular}{ccccccc>{\columncolor[HTML]{FFEEED}}c}
Ground Truth &  UPR-Net L \cite{jin2023unified} & IFRNet \cite{kong2022ifrnet} & EMA-VFI \cite{zhang2023extracting} & GIMM-VFI \cite{guo2024generalizable} & STAA \cite{xiang2022learning} & CSTVR* \cite{zhang2025continuous} & \textbf{TVRN (Ours)} \\
\includegraphics[width=0.09\textwidth]{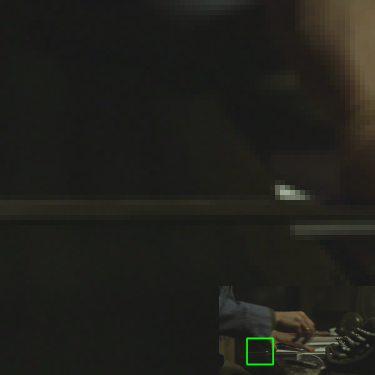}
&\includegraphics[width=0.09\textwidth]{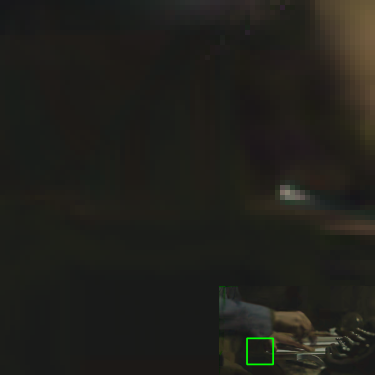}
&\includegraphics[width=0.09\textwidth]{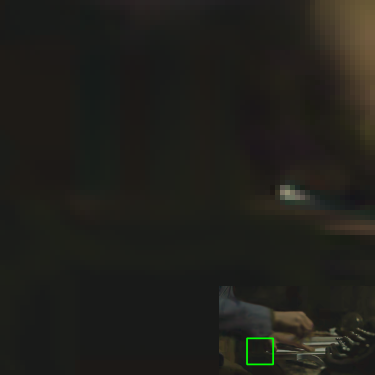}
&\includegraphics[width=0.09\textwidth]{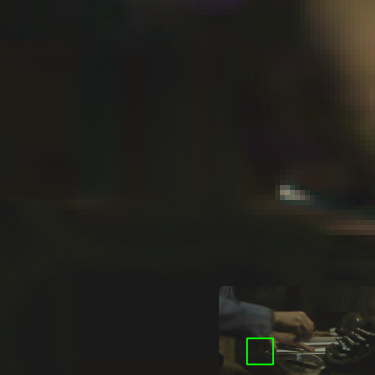}
&\includegraphics[width=0.09\textwidth]{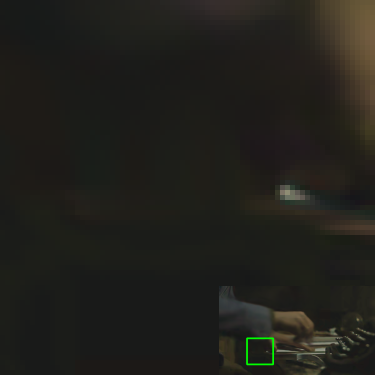}
&\includegraphics[width=0.09\textwidth]{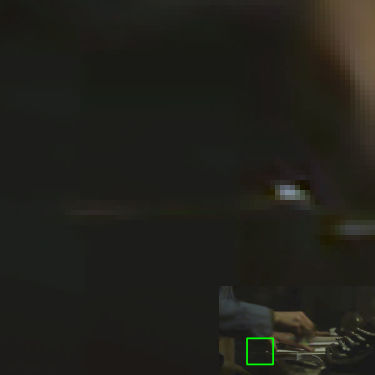}
&\includegraphics[width=0.09\textwidth]{Fig/vimeo_comparison/00025_0034/00025_0034_frame5_qp26_CVRS_finetuned_psnr37.02_ssim0.9476_bpp0.0891.png}
&\includegraphics[width=0.09\textwidth]{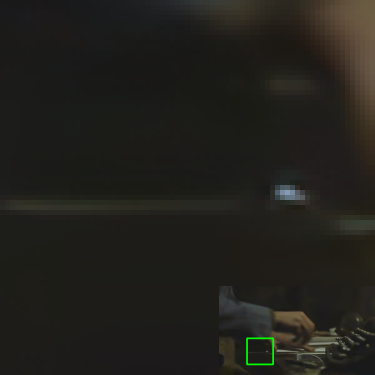}\\
BPP/PSNR & 0.0771/35.34 & {0.0771}/{38.61} & \secondbest{0.0771}/\secondbest{39.03} & 0.0771/38.46 & 0.0839/38.40 & 0.0891/37.02 & \textbf{0.0764}/\textbf{39.79}
\end{tabular}}

\resizebox{1.0\textwidth}{!}{
\scriptsize
\begin{tabular}{ccccccc>{\columncolor[HTML]{FFEEED}}c}
\includegraphics[width=0.09\textwidth]{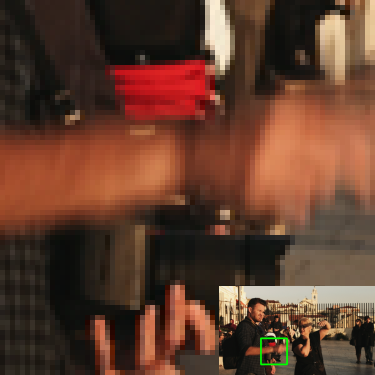}
&\includegraphics[width=0.09\textwidth]{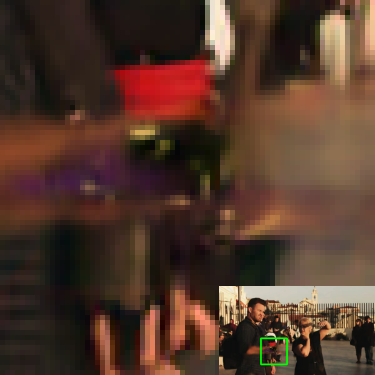}
&\includegraphics[width=0.09\textwidth]{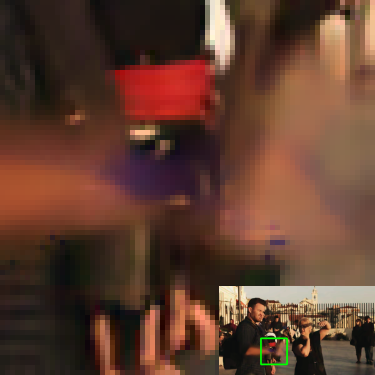}
&\includegraphics[width=0.09\textwidth]{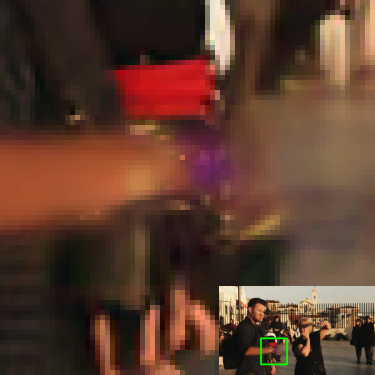}
&\includegraphics[width=0.09\textwidth]{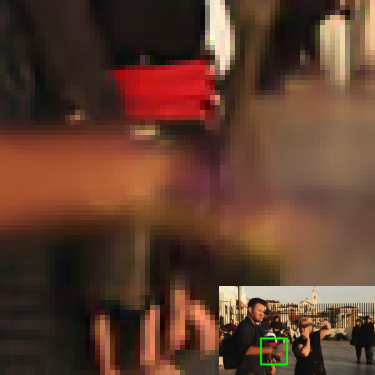}
&\includegraphics[width=0.09\textwidth]{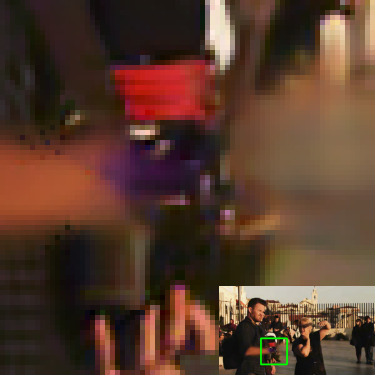}
&\includegraphics[width=0.09\textwidth]{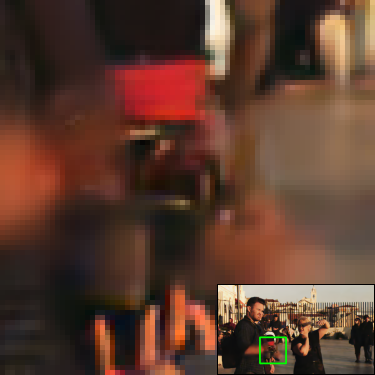}
&\includegraphics[width=0.09\textwidth]{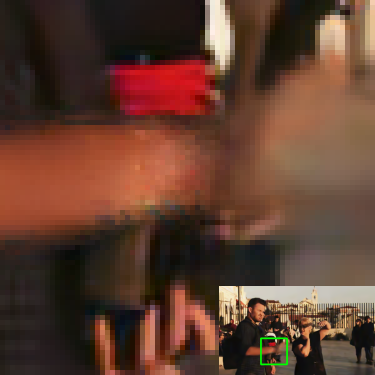}\\
BPP/PSNR & 0.3481/30.26 & 0.3481/30.59 & \secondbest{0.3481}/\secondbest{31.06} & {0.3481}/{30.80} & 0.3613/30.40 & 0.3911/30.61 & \textbf{0.3411}/\textbf{31.07}
\end{tabular}}

\resizebox{1.0\textwidth}{!}{
\scriptsize
\begin{tabular}{ccccccc>{\columncolor[HTML]{FFEEED}}c}
\includegraphics[width=0.09\textwidth]{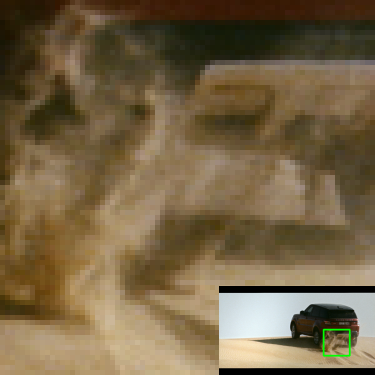}
&\includegraphics[width=0.09\textwidth]{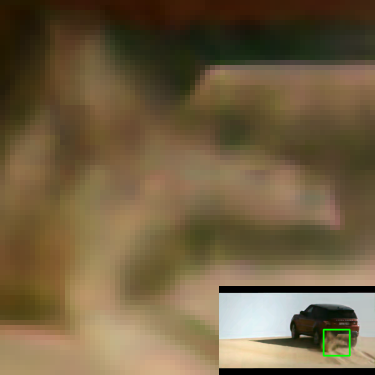}
&\includegraphics[width=0.09\textwidth]{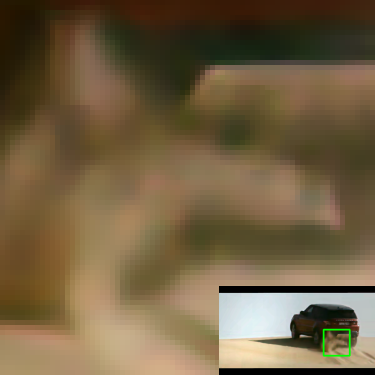}
&\includegraphics[width=0.09\textwidth]{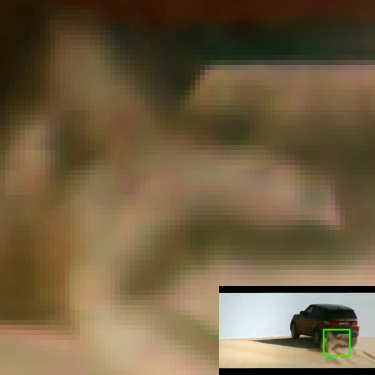}
&\includegraphics[width=0.09\textwidth]{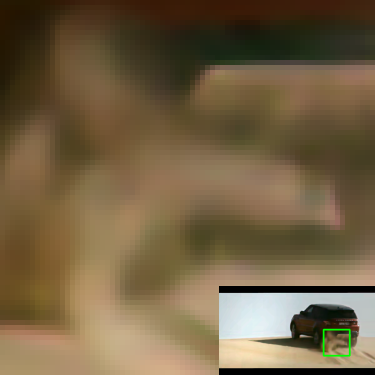}
&\includegraphics[width=0.09\textwidth]{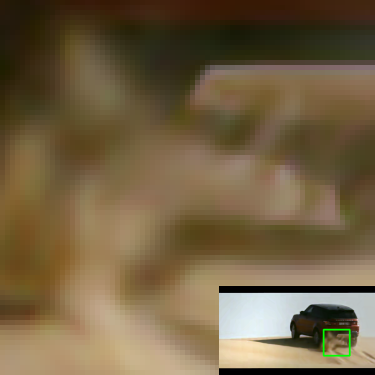}
&\includegraphics[width=0.09\textwidth]{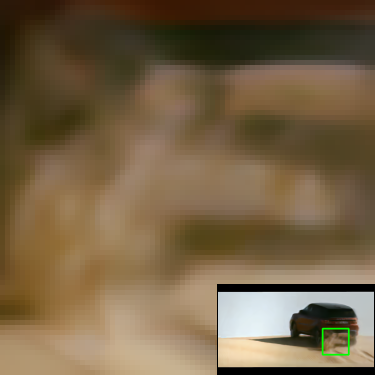}
&\includegraphics[width=0.09\textwidth]{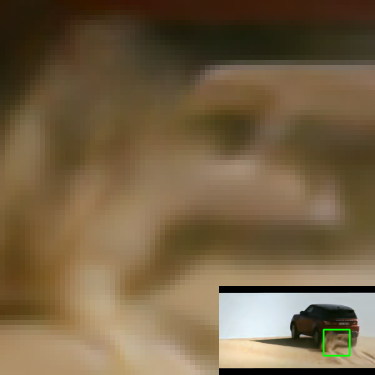}\\
BPP/PSNR & 0.0913/35.22 & 0.0913/35.28 & 0.0913/35.31 & \secondbest{0.0913}/35.40 & 0.0951/\secondbest{35.45} & 0.1066/34.02 & \textbf{0.0859}/\textbf{35.68}
\end{tabular}}

\resizebox{1.0\textwidth}{!}{
\scriptsize
\begin{tabular}{ccccccc>{\columncolor[HTML]{FFEEED}}c}
\includegraphics[width=0.09\textwidth]{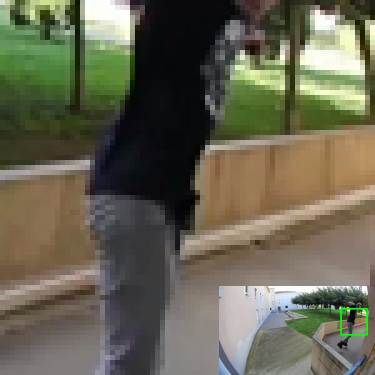}
&\includegraphics[width=0.09\textwidth]{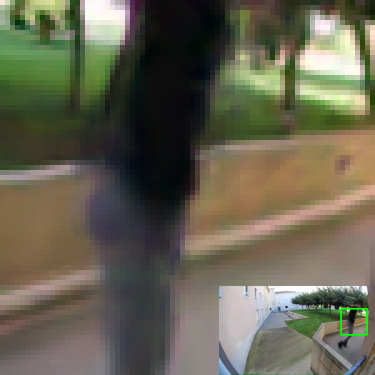}
&\includegraphics[width=0.09\textwidth]{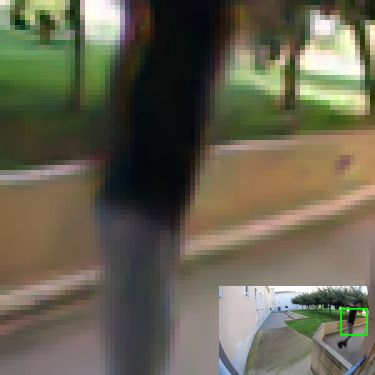}
&\includegraphics[width=0.09\textwidth]{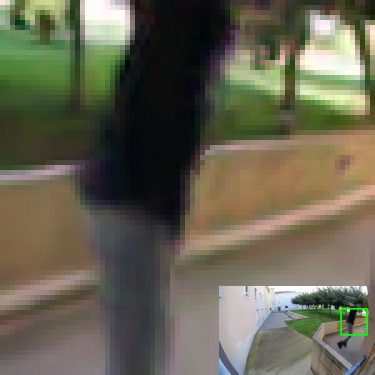}
&\includegraphics[width=0.09\textwidth]{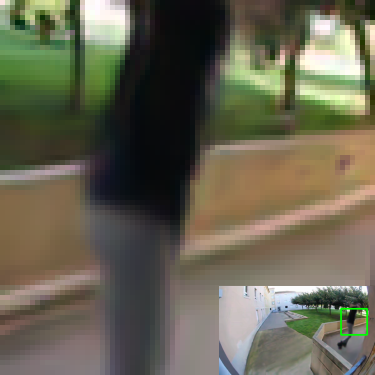}
&\includegraphics[width=0.09\textwidth]{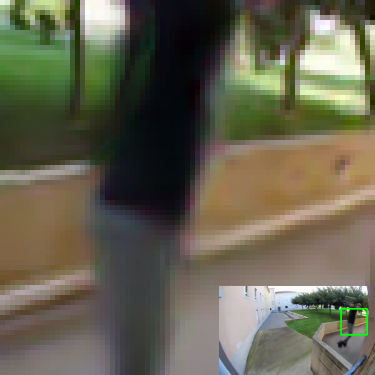}
&\includegraphics[width=0.09\textwidth]{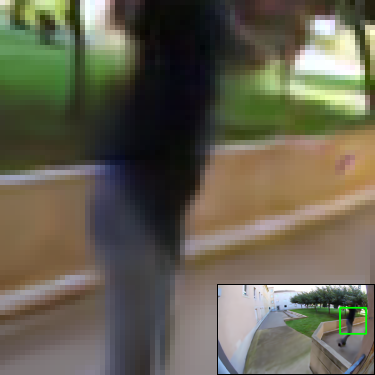}
&\includegraphics[width=0.09\textwidth]{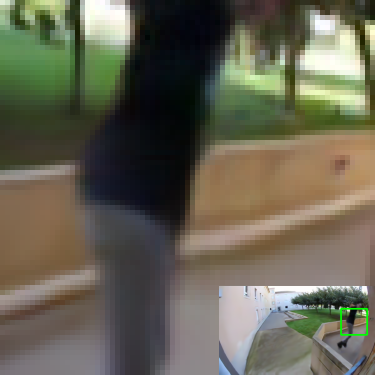}\\
BPP/PSNR & 0.2515/31.76 & 0.2515/31.98 & 0.2515/\secondbest{32.08} & 0.2515/32.07 & \secondbest{0.2488}/31.85 & 0.2367/31.06 & \textbf{0.2302}/\textbf{32.37} \\
\end{tabular}}

\resizebox{1.0\textwidth}{!}{
\scriptsize
\begin{tabular}{ccccccc>{\columncolor[HTML]{FFEEED}}c}
\includegraphics[width=0.09\textwidth]{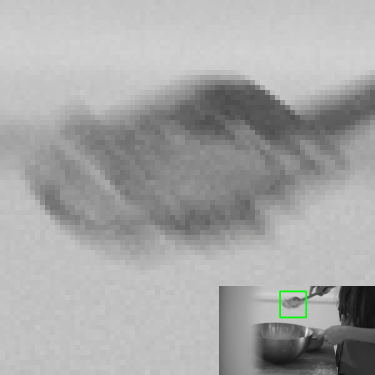}
&\includegraphics[width=0.09\textwidth]{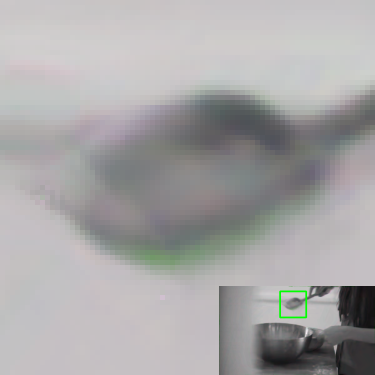}
&\includegraphics[width=0.09\textwidth]{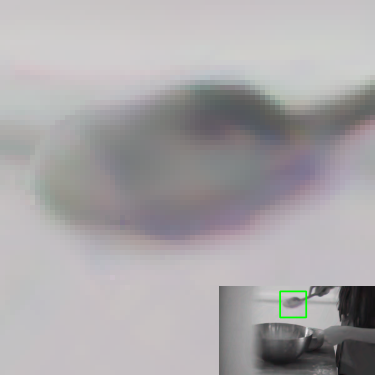}
&\includegraphics[width=0.09\textwidth]{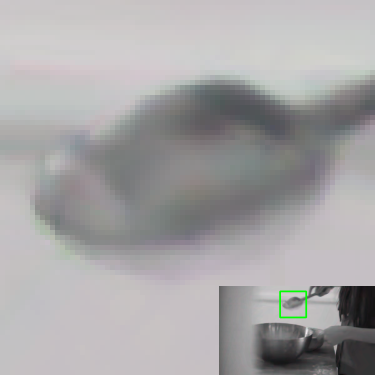}
&\includegraphics[width=0.09\textwidth]{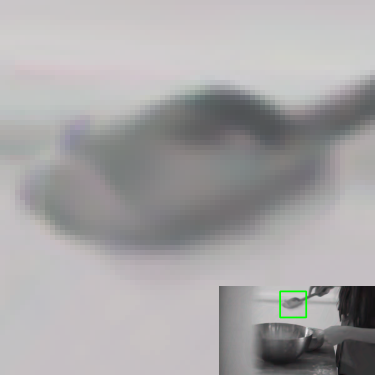}
&\includegraphics[width=0.09\textwidth]{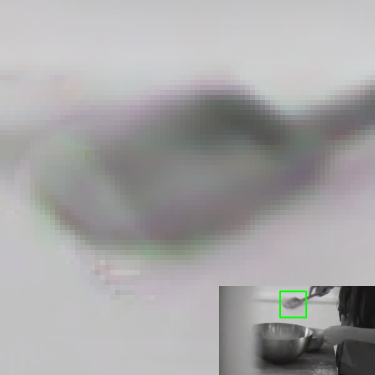}
&\includegraphics[width=0.09\textwidth]{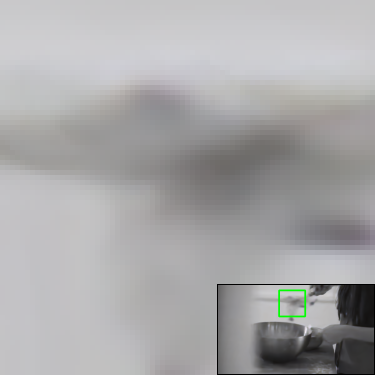}
&\includegraphics[width=0.09\textwidth]{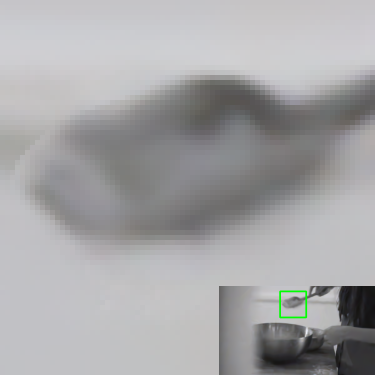}\\
BPP/PSNR & 0.0960/36.19 & {0.0960}/\secondbest{38.15} & 0.0960/38.12 & 0.0960/38.01 & 0.0968/37.27 & \secondbest{0.0948}/33.57 & \textbf{0.0833}/\textbf{39.44}
\end{tabular}}

\resizebox{1.0\textwidth}{!}{
\scriptsize
\begin{tabular}{ccccccc>{\columncolor[HTML]{FFEEED}}c}
\includegraphics[width=0.09\textwidth]{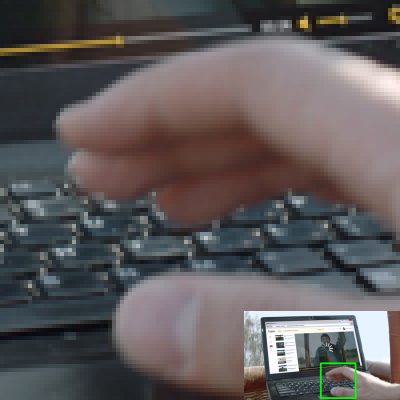}
&\includegraphics[width=0.09\textwidth]{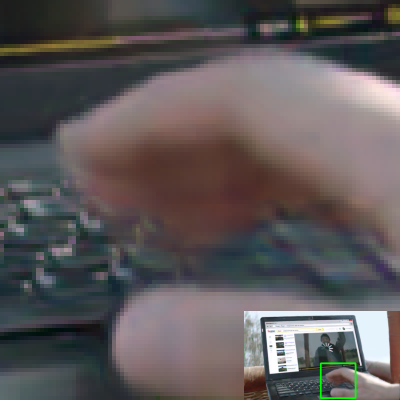}
&\includegraphics[width=0.09\textwidth]{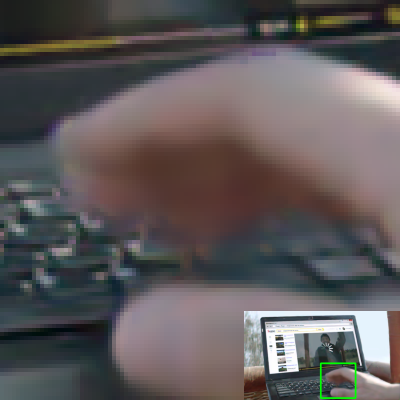}
&\includegraphics[width=0.09\textwidth]{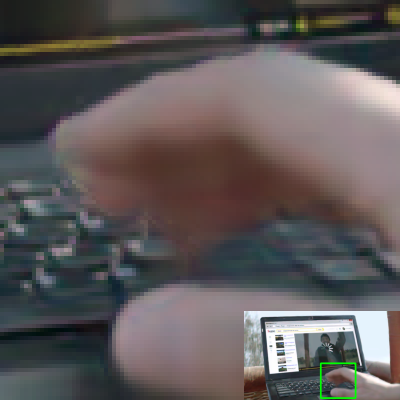}
&\includegraphics[width=0.09\textwidth]{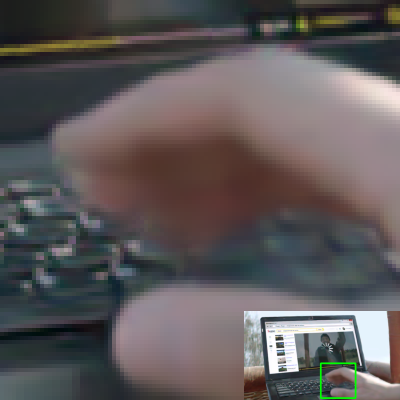}
&\includegraphics[width=0.09\textwidth]{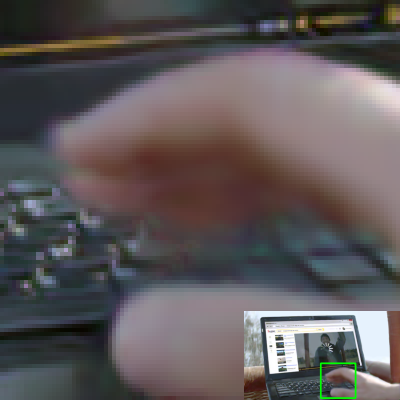}
&\includegraphics[width=0.09\textwidth]{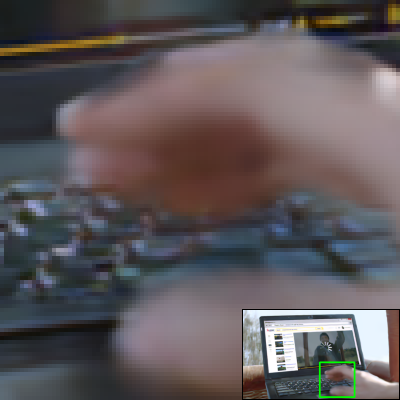}
&\includegraphics[width=0.09\textwidth]{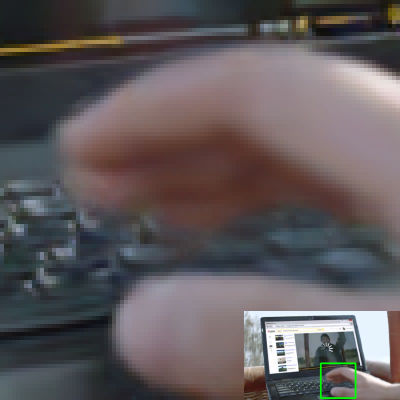} \\
BPP/PSNR& 0.2514/29.86 & 0.2514/30.13 & 0.2514/29.92 & 0.2514/30.14 & \secondbest{0.2443}/{30.18} & 0.2349/\secondbest{30.81} & \textbf{0.2323}/\textbf{32.23} \\
\end{tabular}
}

\caption{\textbf{Qualitative comparison of our method and competitive methods on the vimeo septuplet dataset~\cite{xue2019video}.} 
We crop the frames for easier comparison and visualize the interpolated frames at the bottom-right. 
 Error-prone regions are highlighted with red boxes, best viewed by zooming in.}
\label{fig:comp_hfr_selected}
\end{figure*}

\begin{figure*}[!t]
\setlength{\tabcolsep}{0.5pt}
\centering

\resizebox{1.0\textwidth}{!}{
\scriptsize
\begin{tabular}{cccccc>{\columncolor[HTML]{FFEEED}}cc}
Ground Truth &  UPR-Net L \cite{jin2023unified}  & EMA-VFI \cite{zhang2023extracting} & GIMM-VFI \cite{guo2024generalizable} & STAA \cite{xiang2022learning} & CSTVR* \cite{zhang2025continuous} & \textbf{TVRN (Ours)} & \textbf{TVRN$\downarrow$} \\
\includegraphics[width=0.09\textwidth]{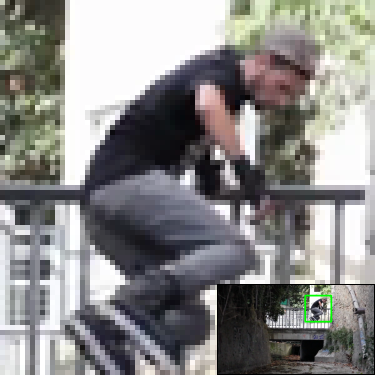}
&\includegraphics[width=0.09\textwidth]{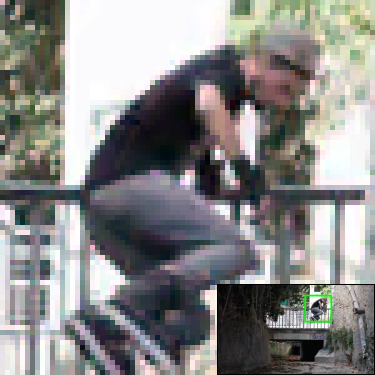}
&\includegraphics[width=0.09\textwidth]{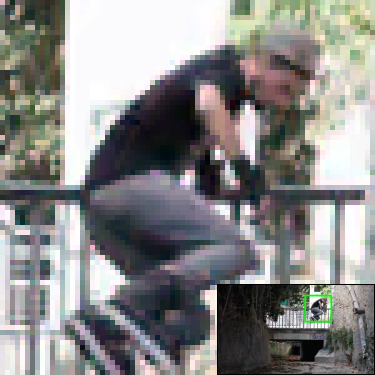}
&\includegraphics[width=0.09\textwidth]{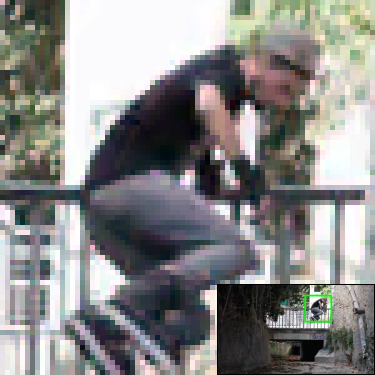}
&\includegraphics[width=0.09\textwidth]{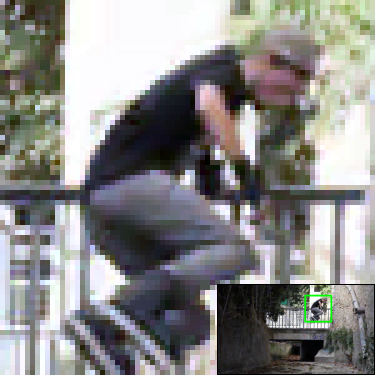}
&\includegraphics[width=0.09\textwidth]{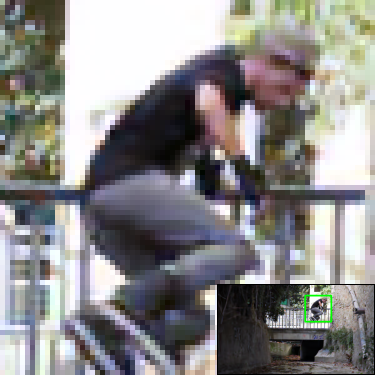}
&\includegraphics[width=0.09\textwidth]{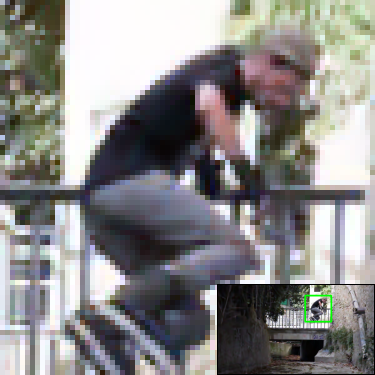}
&\includegraphics[width=0.09\textwidth]{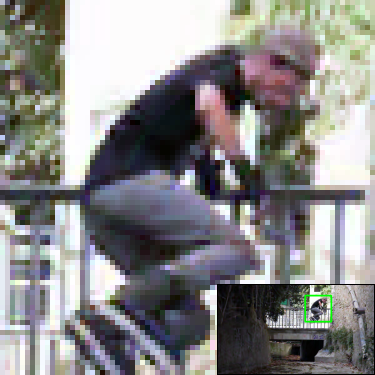}\\
\includegraphics[width=0.09\textwidth]{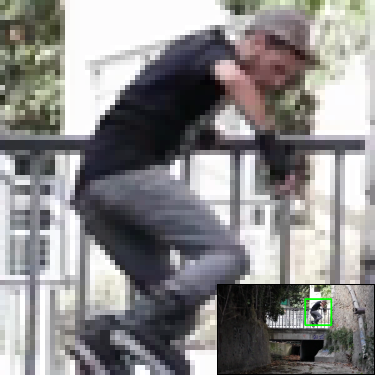}
&\includegraphics[width=0.09\textwidth]{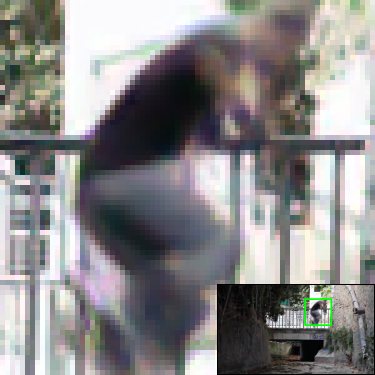}
&\includegraphics[width=0.09\textwidth]{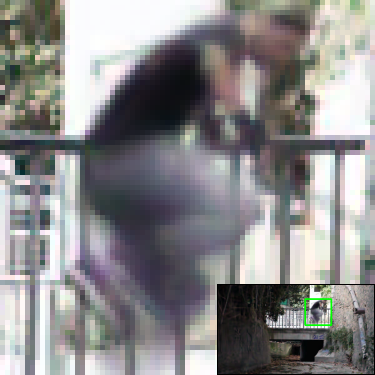}
&\includegraphics[width=0.09\textwidth]{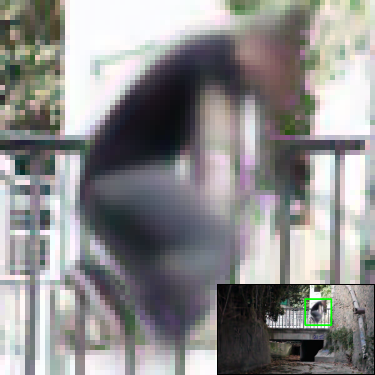}
&\includegraphics[width=0.09\textwidth]{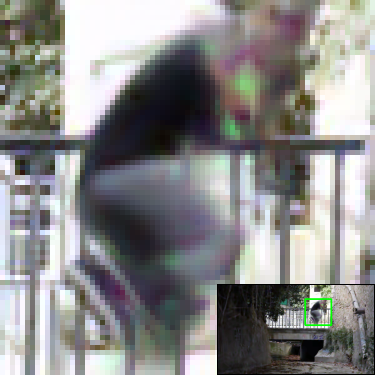}
&\includegraphics[width=0.09\textwidth]{Fig/vimeo_comparison/00026_0036/frame1_qp26_CVRS_finetuned_psnr30.84_ssim0.8630_bpp0.3496.png}
&\includegraphics[width=0.09\textwidth]{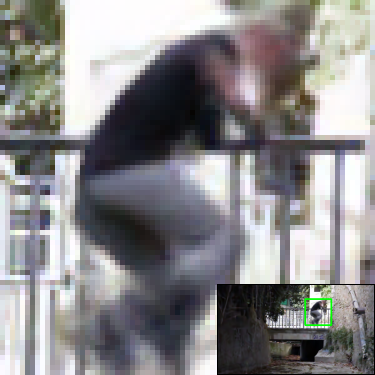}
&\includegraphics[width=0.09\textwidth]{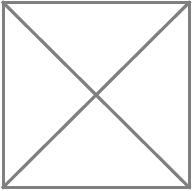}\\
\includegraphics[width=0.09\textwidth]{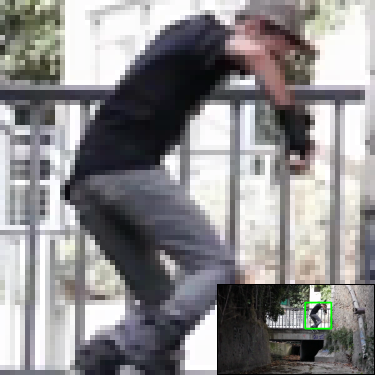}
&\includegraphics[width=0.09\textwidth]{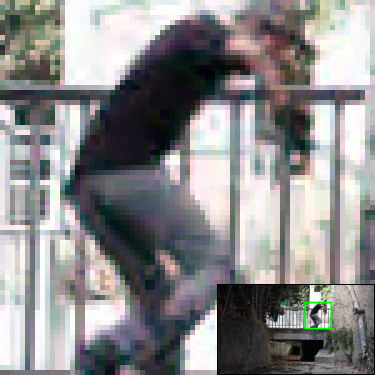}
&\includegraphics[width=0.09\textwidth]{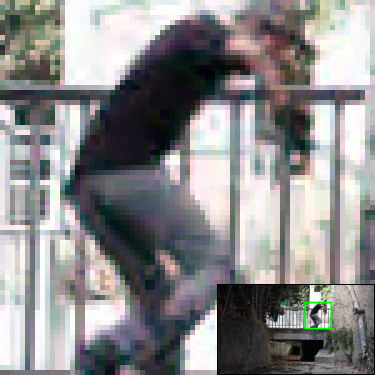}
&\includegraphics[width=0.09\textwidth]{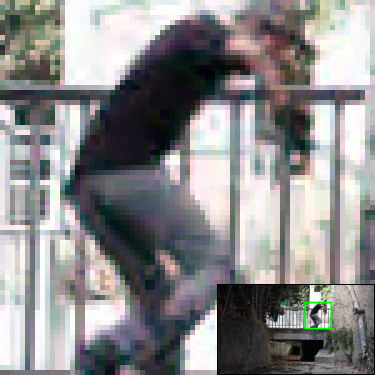}
&\includegraphics[width=0.09\textwidth]{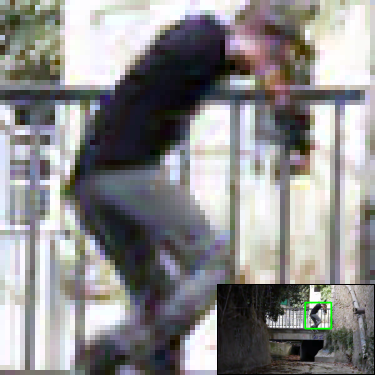}
&\includegraphics[width=0.09\textwidth]{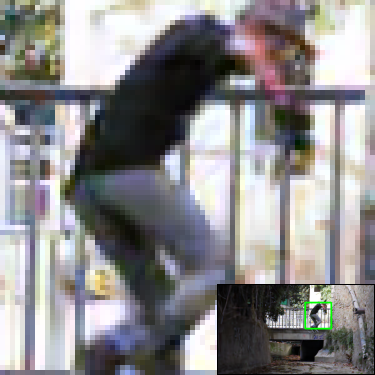}
&\includegraphics[width=0.09\textwidth]{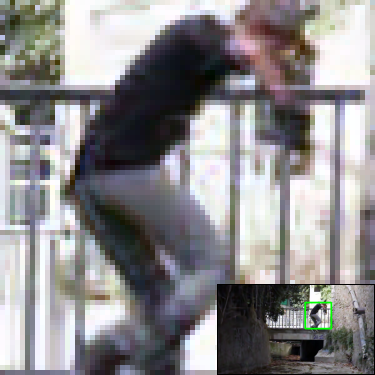}
&\includegraphics[width=0.09\textwidth]{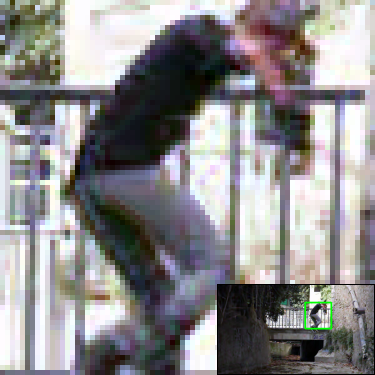}\\
\includegraphics[width=0.09\textwidth]{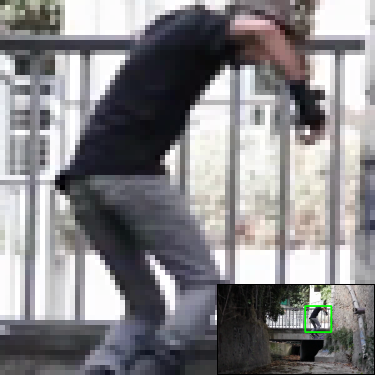}
&\includegraphics[width=0.09\textwidth]{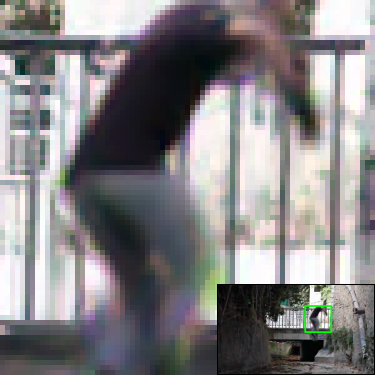}
&\includegraphics[width=0.09\textwidth]{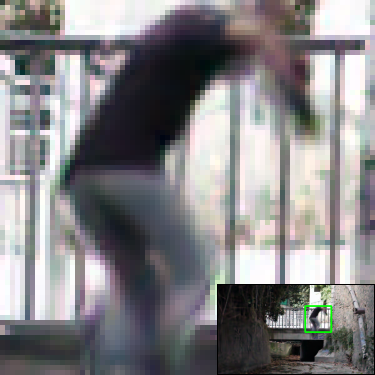}
&\includegraphics[width=0.09\textwidth]{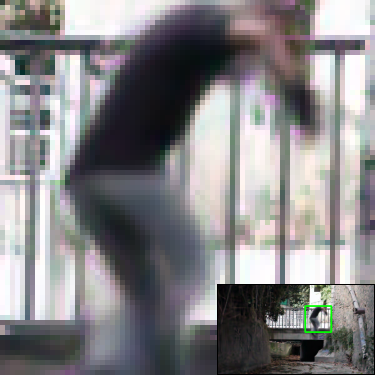}
&\includegraphics[width=0.09\textwidth]{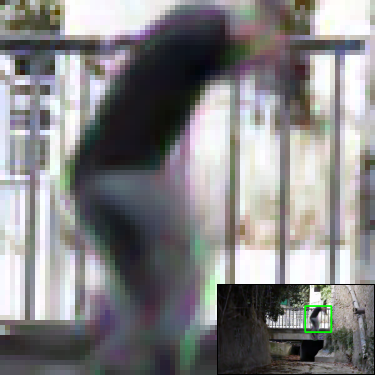}
&\includegraphics[width=0.09\textwidth]{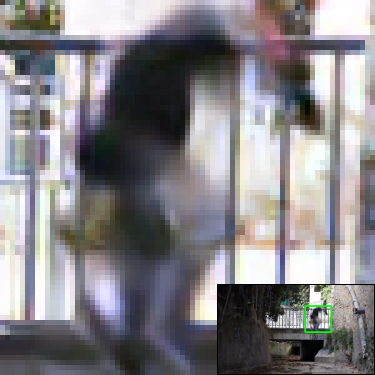}
&\includegraphics[width=0.09\textwidth]{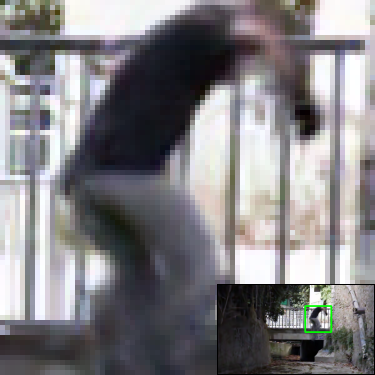}
&\includegraphics[width=0.09\textwidth]{Fig/cross_filled_square.pdf}\\
\includegraphics[width=0.09\textwidth]{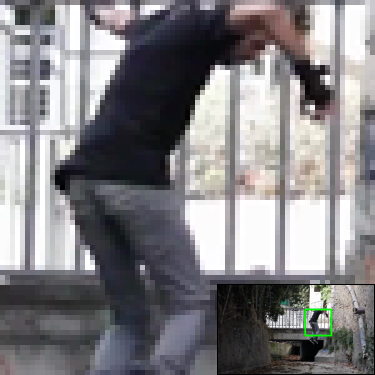}
&\includegraphics[width=0.09\textwidth]{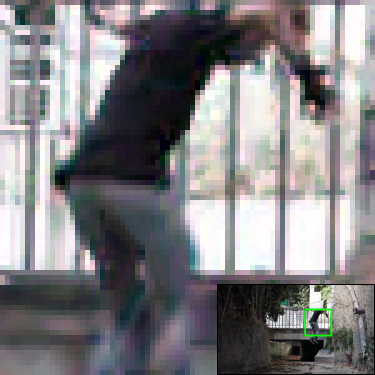}
&\includegraphics[width=0.09\textwidth]{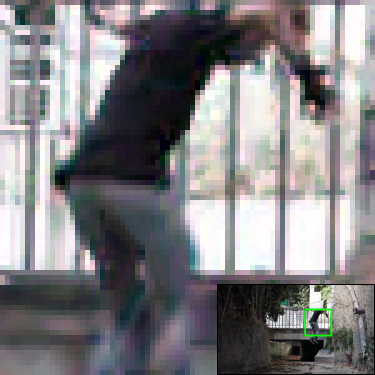}
&\includegraphics[width=0.09\textwidth]{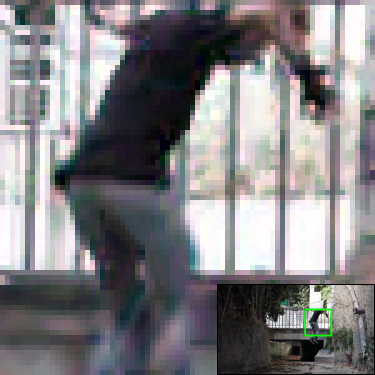}
&\includegraphics[width=0.09\textwidth]{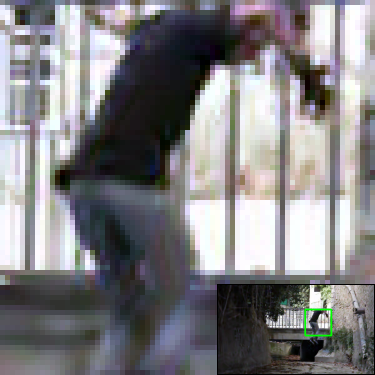}
&\includegraphics[width=0.09\textwidth]{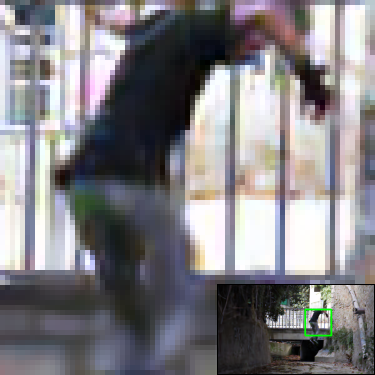}
&\includegraphics[width=0.09\textwidth]{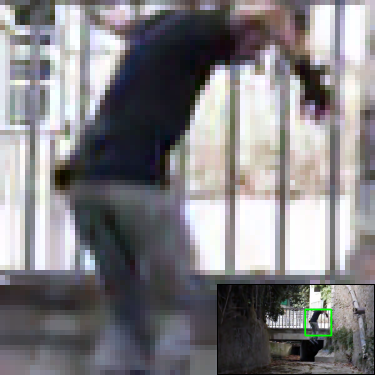}
&\includegraphics[width=0.09\textwidth]{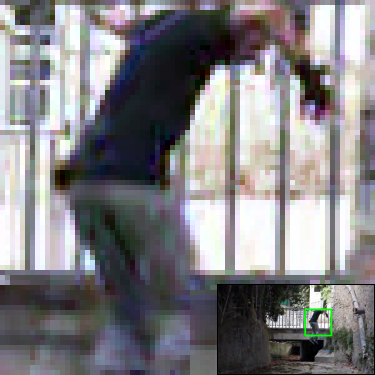}\\
\includegraphics[width=0.09\textwidth]{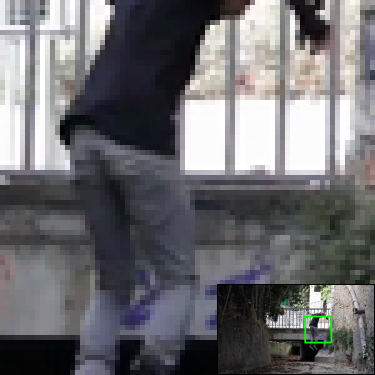}
&\includegraphics[width=0.09\textwidth]{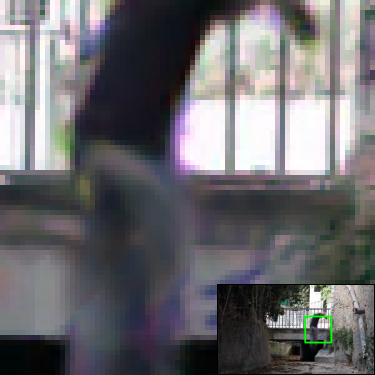}
&\includegraphics[width=0.09\textwidth]{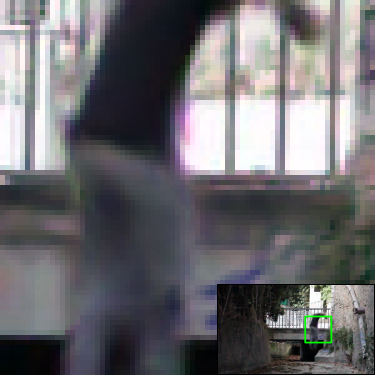}
&\includegraphics[width=0.09\textwidth]{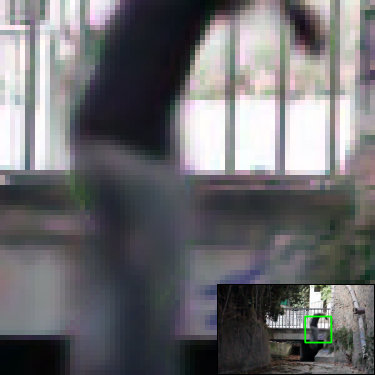}
&\includegraphics[width=0.09\textwidth]{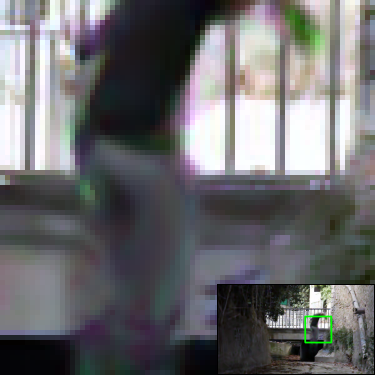}
&\includegraphics[width=0.09\textwidth]{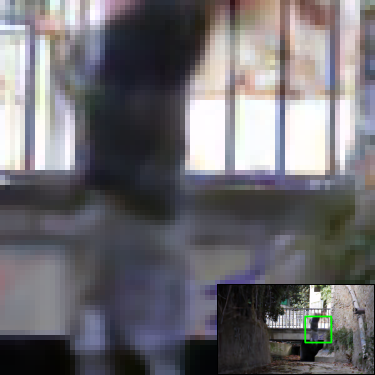}
&\includegraphics[width=0.09\textwidth]{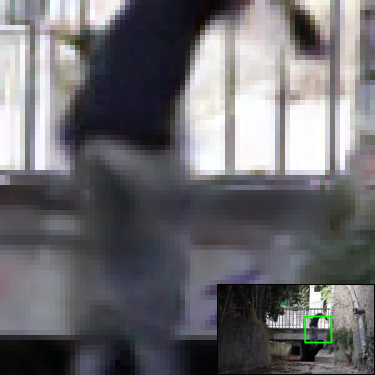}
&\includegraphics[width=0.09\textwidth]{Fig/cross_filled_square.pdf}\\
\includegraphics[width=0.09\textwidth]{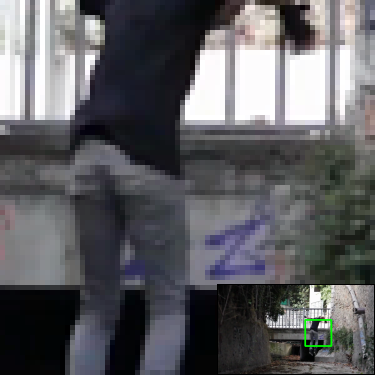}
&\includegraphics[width=0.09\textwidth]{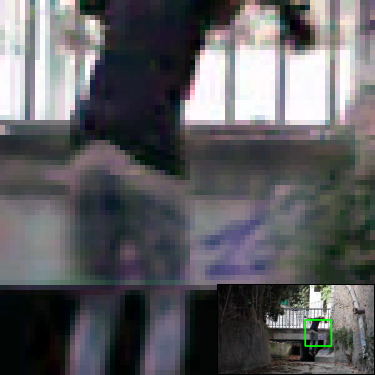}
&\includegraphics[width=0.09\textwidth]{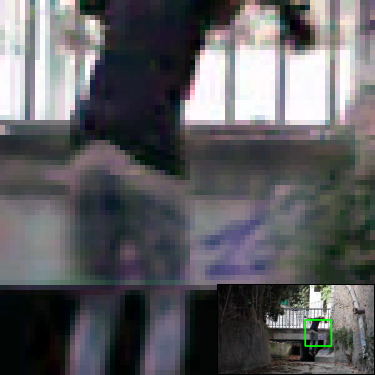}
&\includegraphics[width=0.09\textwidth]{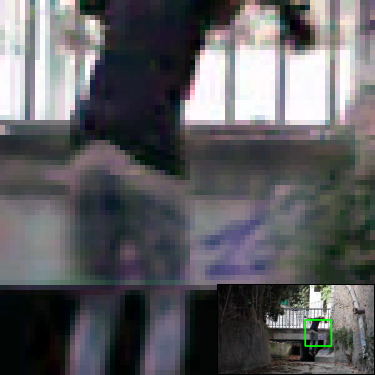}
&\includegraphics[width=0.09\textwidth]{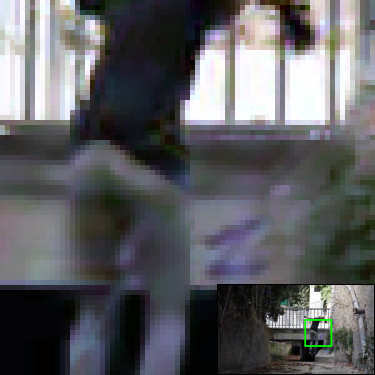}
&\includegraphics[width=0.09\textwidth]{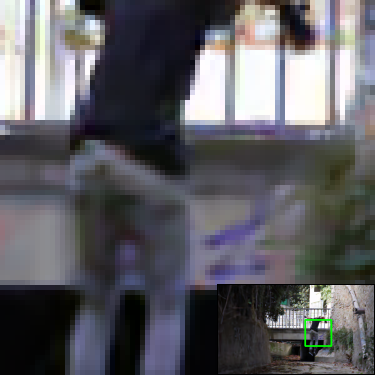}
&\includegraphics[width=0.09\textwidth]{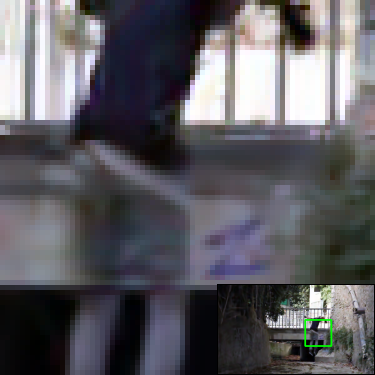}
&\includegraphics[width=0.09\textwidth]{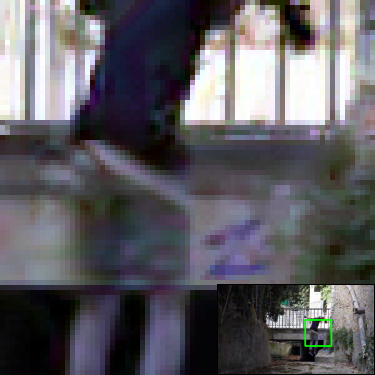}\\
\end{tabular}}

\caption{\textbf{Qualitative comparison of our method and competitive methods on the full-length sequence \textit{00026\_0036} from the Vimeo septuplet dataset~\cite{xue2019video}.} 
We crop the frames for easier comparison and visualize the interpolated frames at the bottom right. TVRN$\downarrow$ donates the downscaled low-frame-rate video.}
\label{fig:vimeo_test_1}
\end{figure*}

\begin{figure*}[!t]
\setlength{\tabcolsep}{0.5pt}
\centering
\label{fig: vimeo_test_2}

\resizebox{1.0\textwidth}{!}{
\scriptsize
\begin{tabular}{cccccc>{\columncolor[HTML]{FFEEED}}cc}
Ground Truth &  UPR-Net L \cite{jin2023unified}  & EMA-VFI \cite{zhang2023extracting} & GIMM-VFI \cite{guo2024generalizable} & STAA \cite{xiang2022learning} & CSTVR* \cite{zhang2025continuous} & \textbf{TVRN (Ours)} & \textbf{TVRN$\downarrow$} \\
\includegraphics[width=0.09\textwidth]{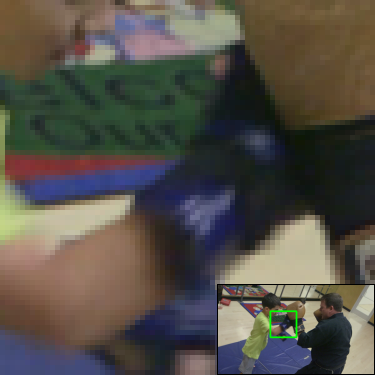}
&\includegraphics[width=0.09\textwidth]{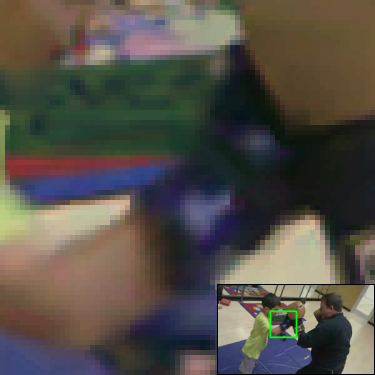}
&\includegraphics[width=0.09\textwidth]{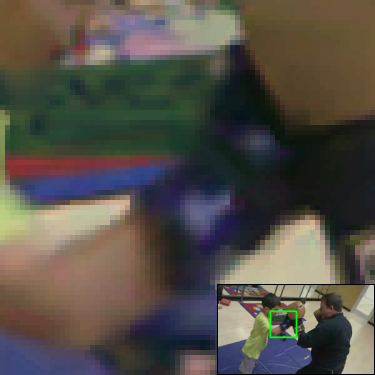}
&\includegraphics[width=0.09\textwidth]{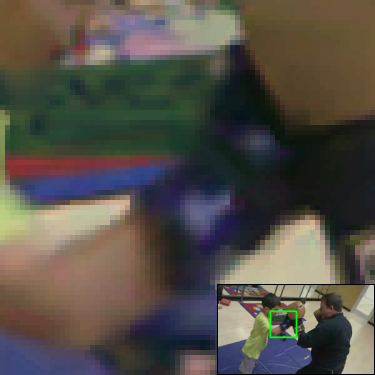}
&\includegraphics[width=0.09\textwidth]{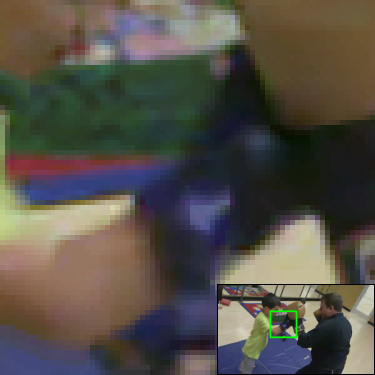}
&\includegraphics[width=0.09\textwidth]{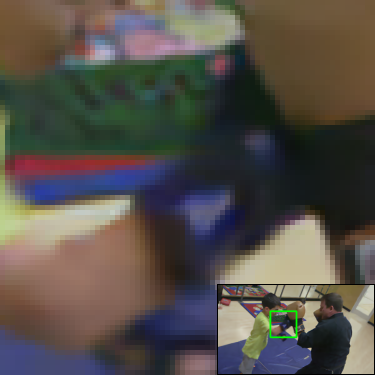}
&\includegraphics[width=0.09\textwidth]{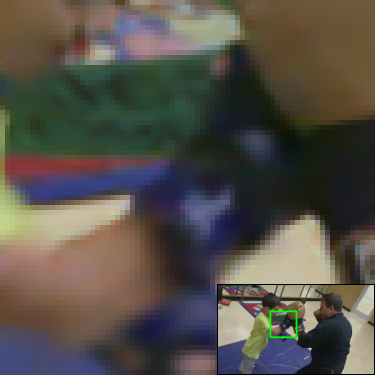}
&\includegraphics[width=0.09\textwidth]{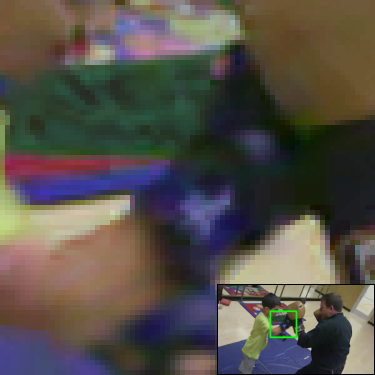}\\
\includegraphics[width=0.09\textwidth]{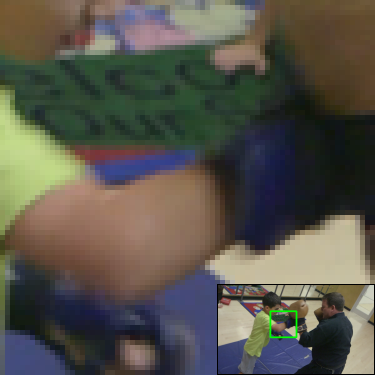}
&\includegraphics[width=0.09\textwidth]{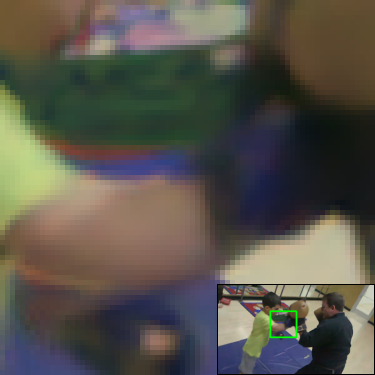}
&\includegraphics[width=0.09\textwidth]{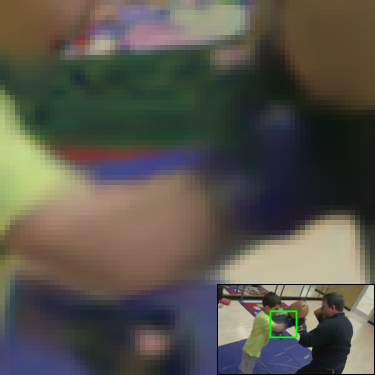}
&\includegraphics[width=0.09\textwidth]{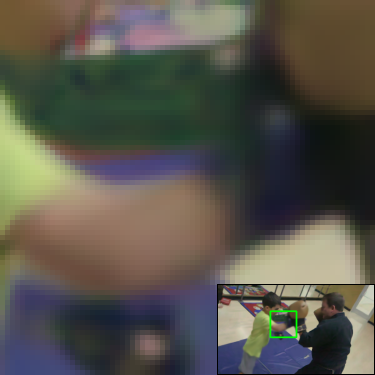}
&\includegraphics[width=0.09\textwidth]{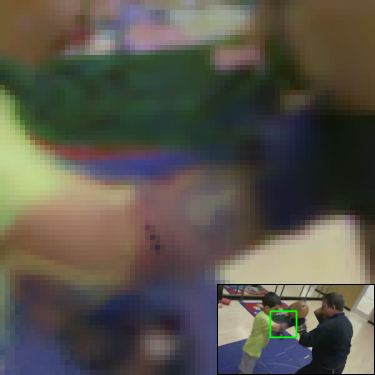}
&\includegraphics[width=0.09\textwidth]{Fig/vimeo_comparison/00023_0059/frame1_qp26_CVRS_finetuned_psnr33.86_ssim0.9337_bpp0.2081.png}
&\includegraphics[width=0.09\textwidth]{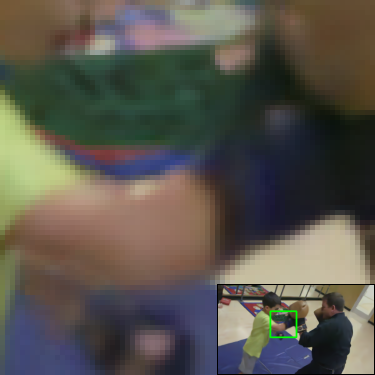}
&\includegraphics[width=0.09\textwidth]{Fig/cross_filled_square.pdf}\\
\includegraphics[width=0.09\textwidth]{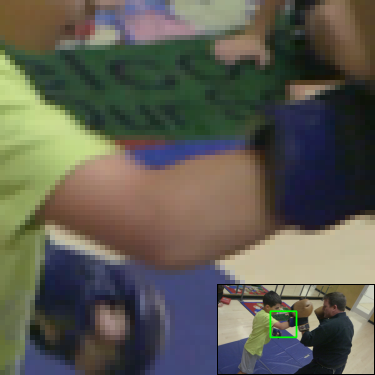}
&\includegraphics[width=0.09\textwidth]{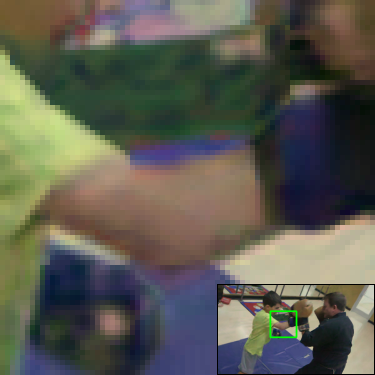}
&\includegraphics[width=0.09\textwidth]{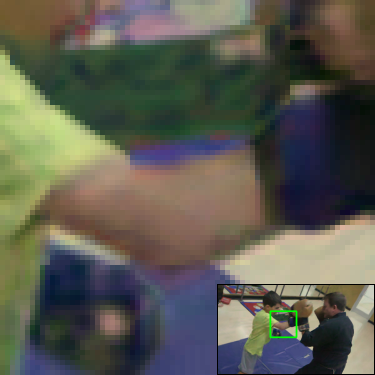}
&\includegraphics[width=0.09\textwidth]{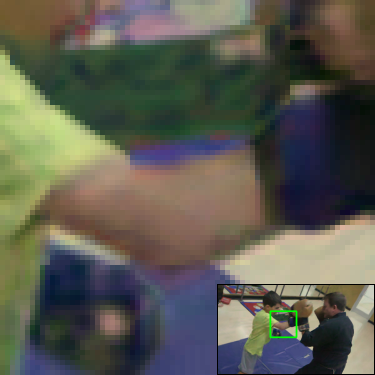}
&\includegraphics[width=0.09\textwidth]{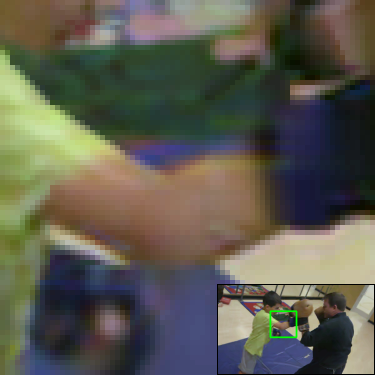}
&\includegraphics[width=0.09\textwidth]{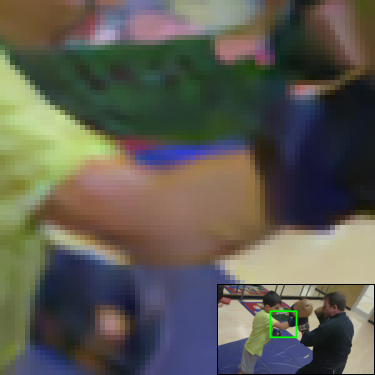}
&\includegraphics[width=0.09\textwidth]{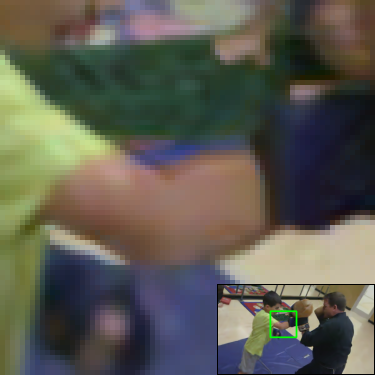}
&\includegraphics[width=0.09\textwidth]{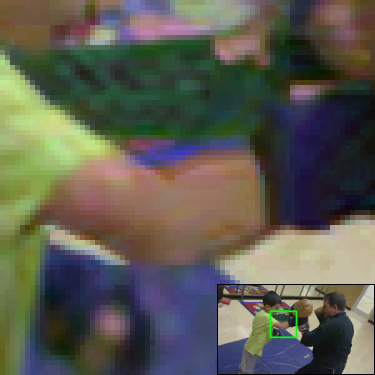}\\
\includegraphics[width=0.09\textwidth]{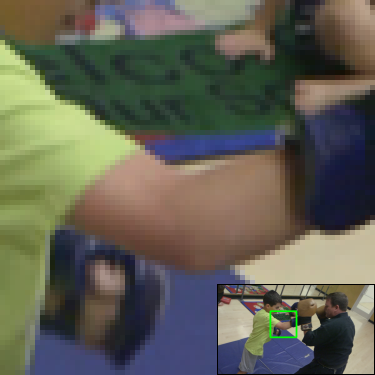}
&\includegraphics[width=0.09\textwidth]{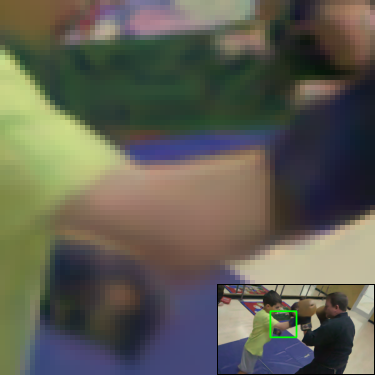}
&\includegraphics[width=0.09\textwidth]{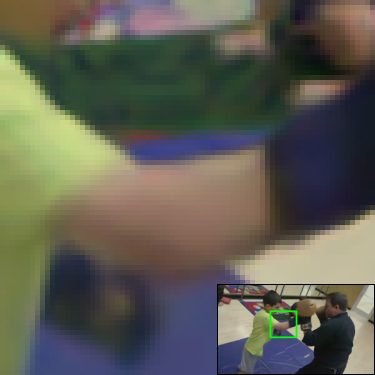}
&\includegraphics[width=0.09\textwidth]{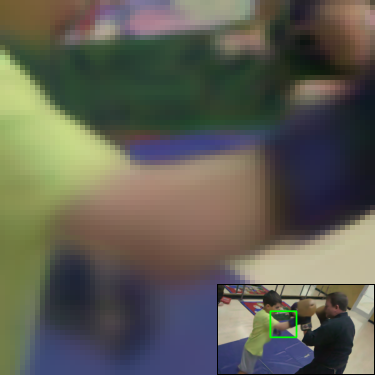}
&\includegraphics[width=0.09\textwidth]{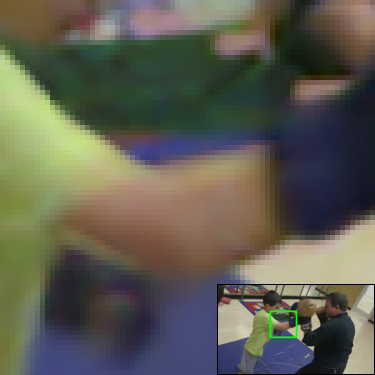}
&\includegraphics[width=0.09\textwidth]{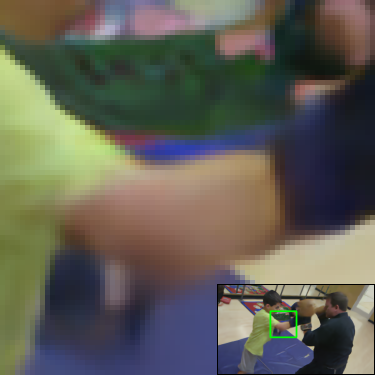}
&\includegraphics[width=0.09\textwidth]{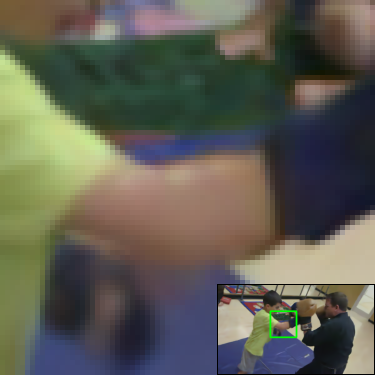}
&\includegraphics[width=0.09\textwidth]{Fig/cross_filled_square.pdf}\\
\includegraphics[width=0.09\textwidth]{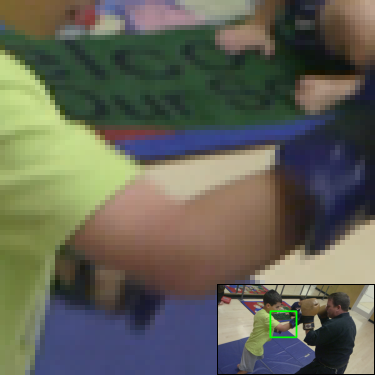}
&\includegraphics[width=0.09\textwidth]{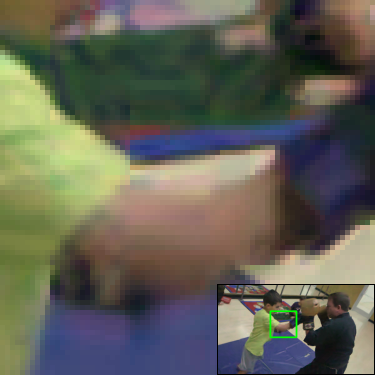}
&\includegraphics[width=0.09\textwidth]{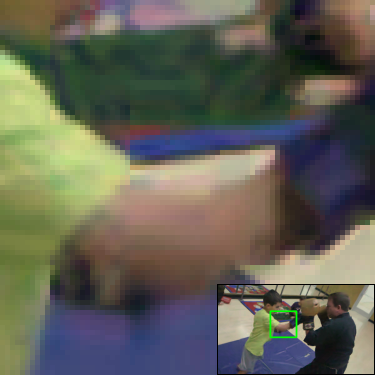}
&\includegraphics[width=0.09\textwidth]{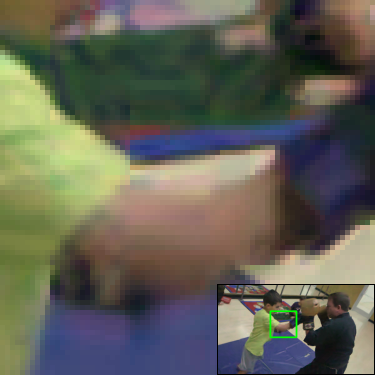}
&\includegraphics[width=0.09\textwidth]{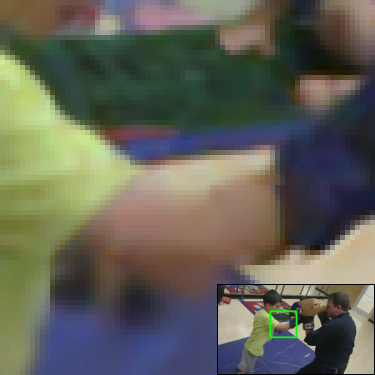}
&\includegraphics[width=0.09\textwidth]{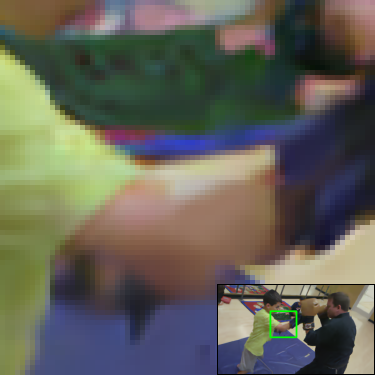}
&\includegraphics[width=0.09\textwidth]{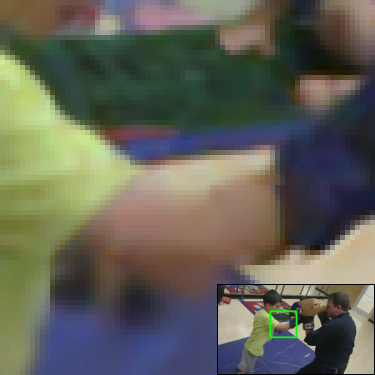}
&\includegraphics[width=0.09\textwidth]{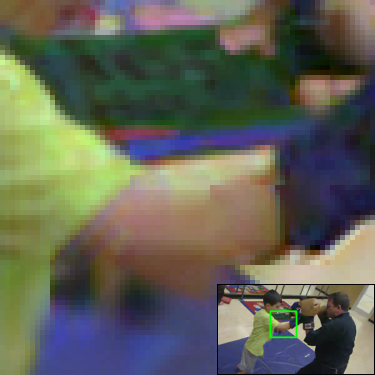}\\
\includegraphics[width=0.09\textwidth]{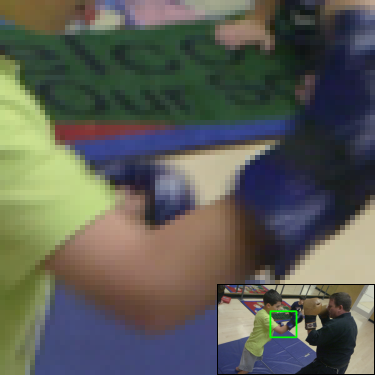}
&\includegraphics[width=0.09\textwidth]{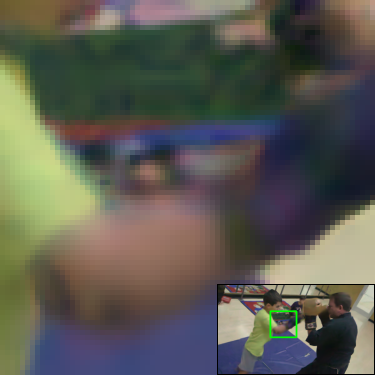}
&\includegraphics[width=0.09\textwidth]{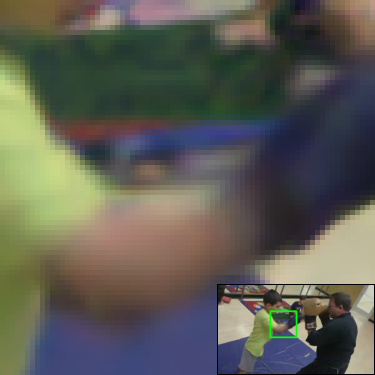}
&\includegraphics[width=0.09\textwidth]{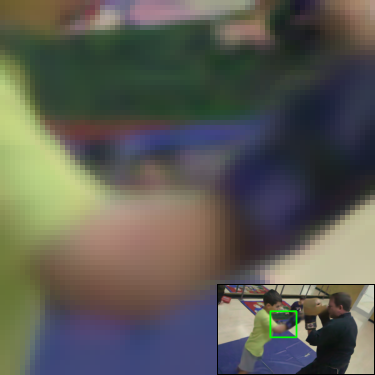}
&\includegraphics[width=0.09\textwidth]{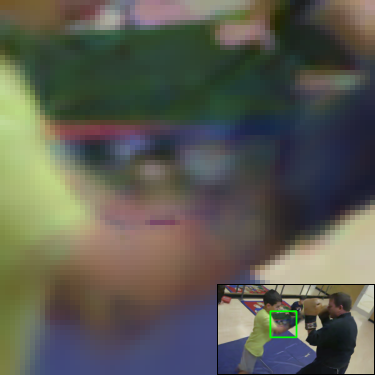}
&\includegraphics[width=0.09\textwidth]{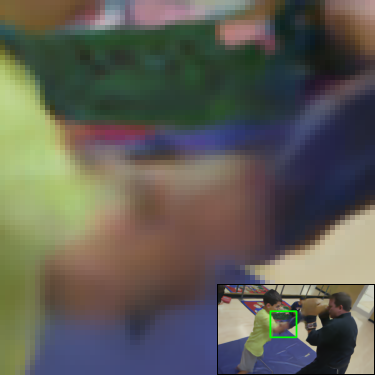}
&\includegraphics[width=0.09\textwidth]{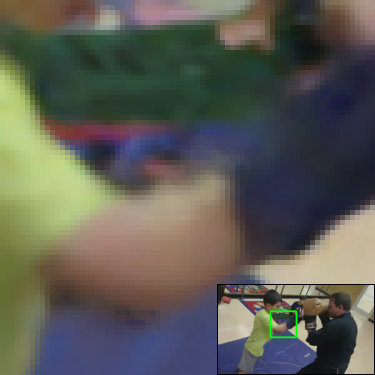}
&\includegraphics[width=0.09\textwidth]{Fig/cross_filled_square.pdf}\\
\includegraphics[width=0.09\textwidth]{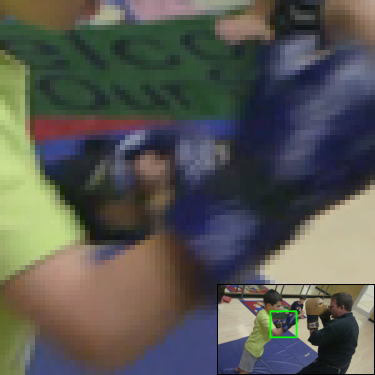}
&\includegraphics[width=0.09\textwidth]{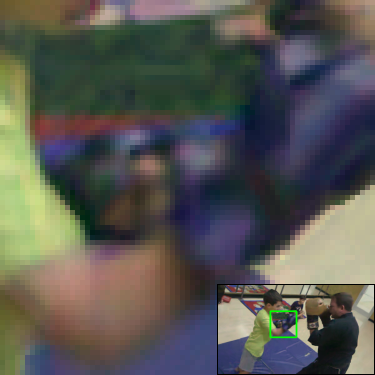}
&\includegraphics[width=0.09\textwidth]{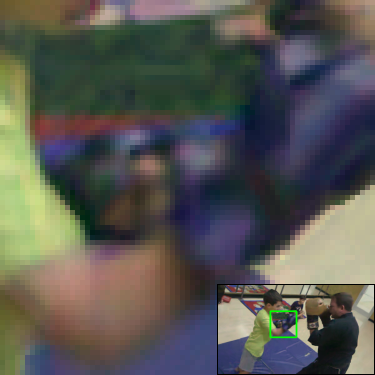}
&\includegraphics[width=0.09\textwidth]{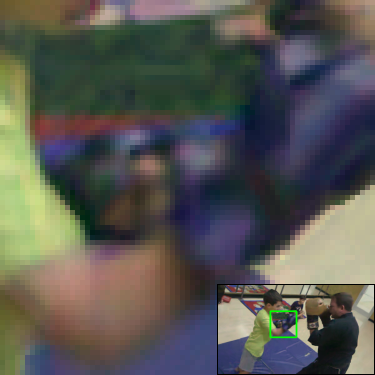}
&\includegraphics[width=0.09\textwidth]{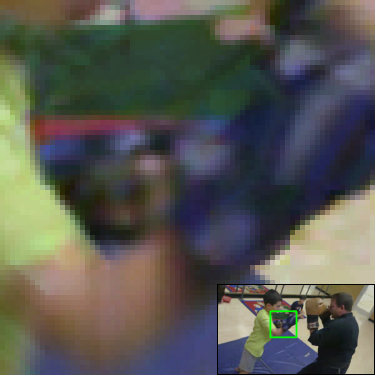}
&\includegraphics[width=0.09\textwidth]{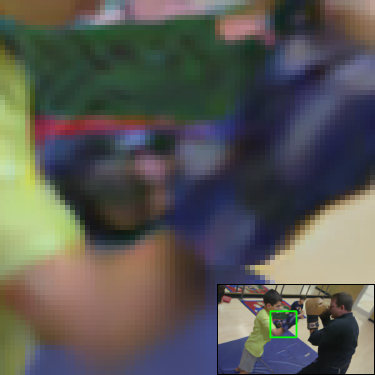}
&\includegraphics[width=0.09\textwidth]{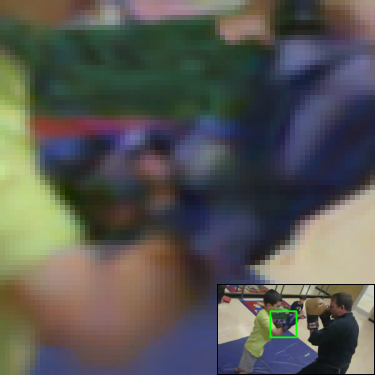}
&\includegraphics[width=0.09\textwidth]{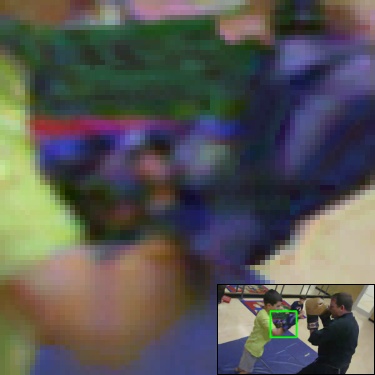}\\
\end{tabular}}

\caption{\textbf{Qualitative comparison of our method and competitive methods on the full-length sequence \textit{00023\_0059} from the Vimeo septuplet dataset~\cite{xue2019video}.} 
We crop the frames for easier comparison and visualize the interpolated frames at the bottom right. TVRN$\downarrow$ donates the downscaled low-frame-rate video.}
\label{fig:vimeo_test_2}
\end{figure*}

\begin{figure*}[!t]
\setlength{\tabcolsep}{0.5pt}
\centering

\label{fig: vimeo_test_3}

\resizebox{1.0\textwidth}{!}{
\scriptsize
\begin{tabular}{cccccc>{\columncolor[HTML]{FFEEED}}cc}
Ground Truth &  UPR-Net L \cite{jin2023unified}  & EMA-VFI \cite{zhang2023extracting} & GIMM-VFI \cite{guo2024generalizable} & STAA \cite{xiang2022learning} & CSTVR* \cite{zhang2025continuous} & \textbf{TVRN (Ours)} & \textbf{TVRN$\downarrow$} \\
\includegraphics[width=0.09\textwidth]{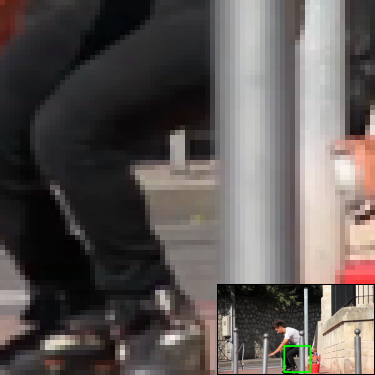}
&\includegraphics[width=0.09\textwidth]{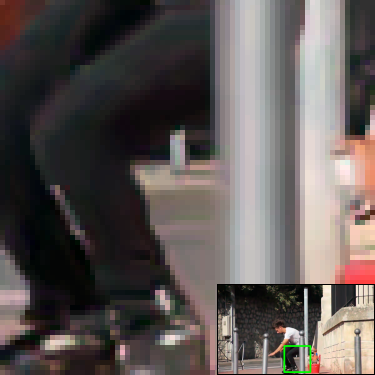}
&\includegraphics[width=0.09\textwidth]{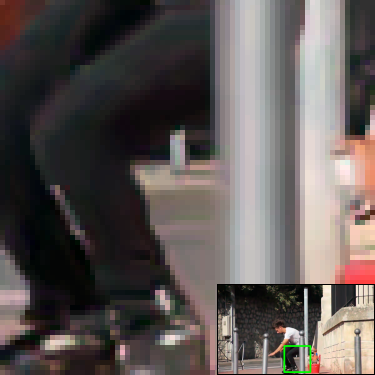}
&\includegraphics[width=0.09\textwidth]{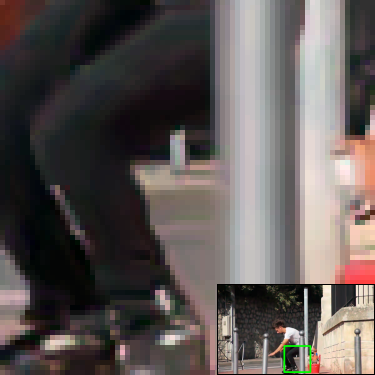}
&\includegraphics[width=0.09\textwidth]{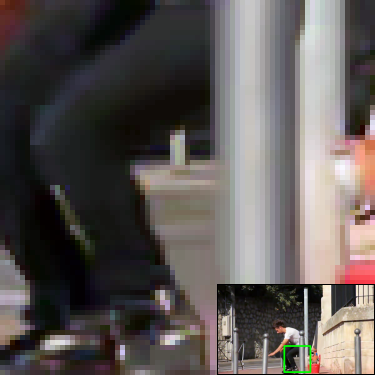}
&\includegraphics[width=0.09\textwidth]{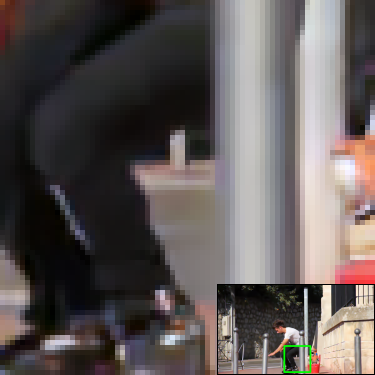}
&\includegraphics[width=0.09\textwidth]{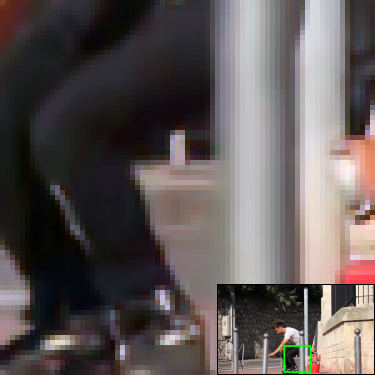}
&\includegraphics[width=0.09\textwidth]{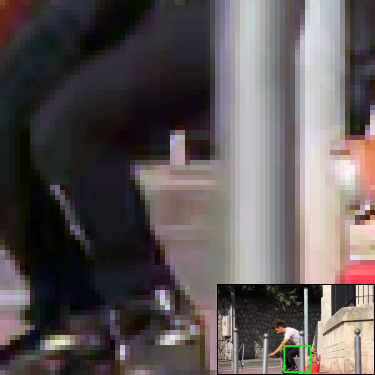}\\
\includegraphics[width=0.09\textwidth]{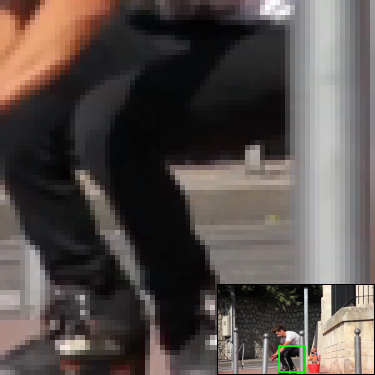}
&\includegraphics[width=0.09\textwidth]{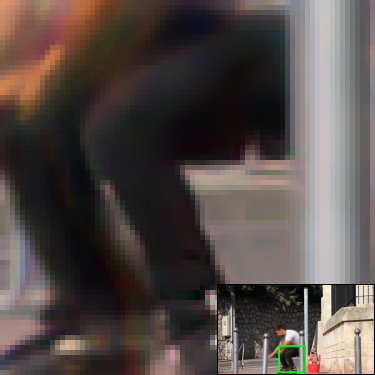}
&\includegraphics[width=0.09\textwidth]{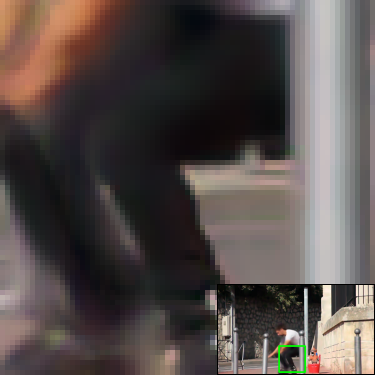}
&\includegraphics[width=0.09\textwidth]{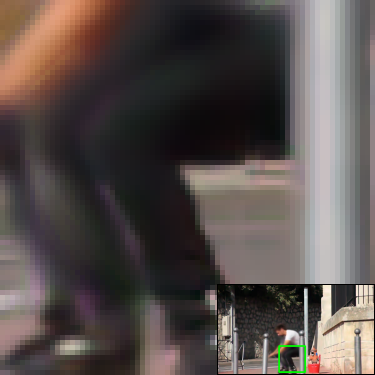}
&\includegraphics[width=0.09\textwidth]{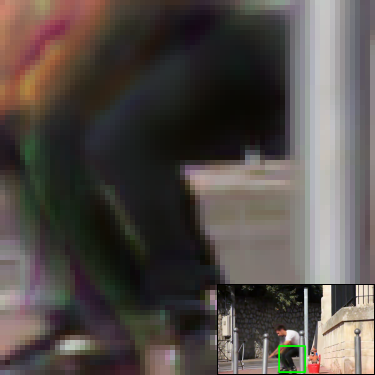}
&\includegraphics[width=0.09\textwidth]{Fig/vimeo_comparison/00026_0003/frame1_qp26_CVRS_finetuned_psnr30.40_ssim0.8835_bpp0.3167.png}
&\includegraphics[width=0.09\textwidth]{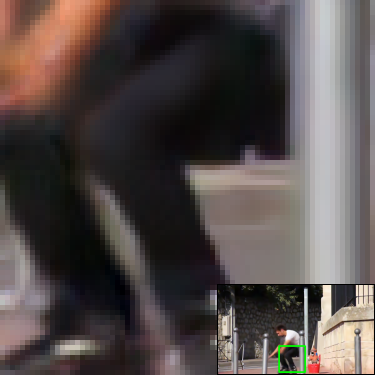}
&\includegraphics[width=0.09\textwidth]{Fig/cross_filled_square.pdf}\\
\includegraphics[width=0.09\textwidth]{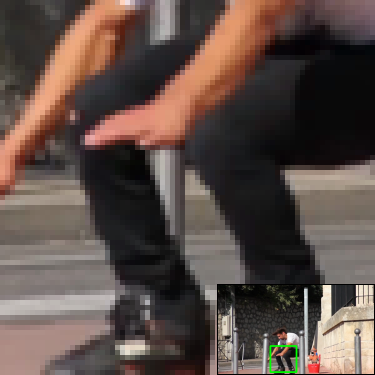}
&\includegraphics[width=0.09\textwidth]{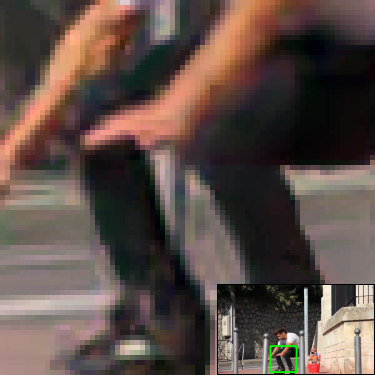}
&\includegraphics[width=0.09\textwidth]{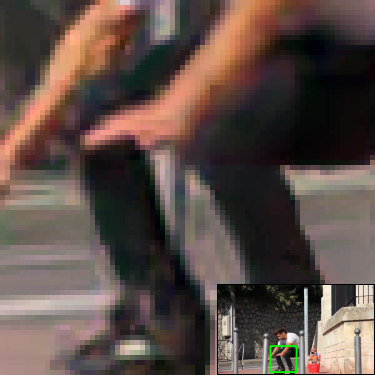}
&\includegraphics[width=0.09\textwidth]{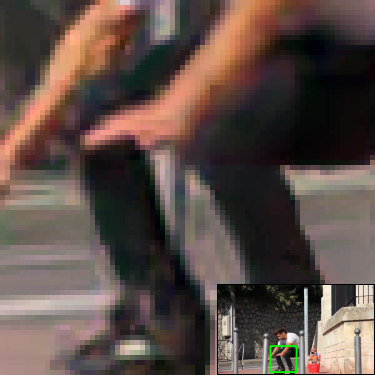}
&\includegraphics[width=0.09\textwidth]{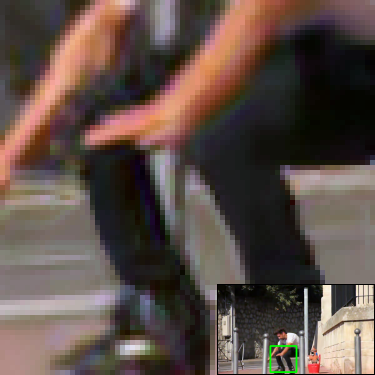}
&\includegraphics[width=0.09\textwidth]{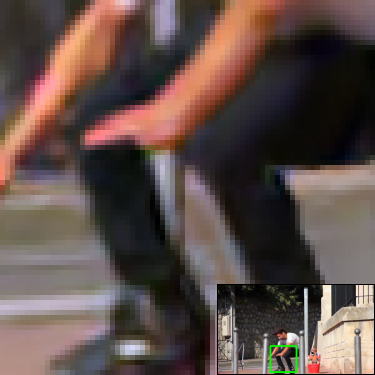}
&\includegraphics[width=0.09\textwidth]{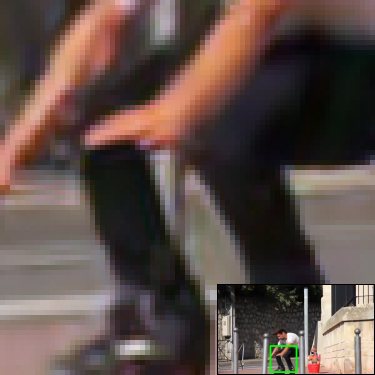}
&\includegraphics[width=0.09\textwidth]{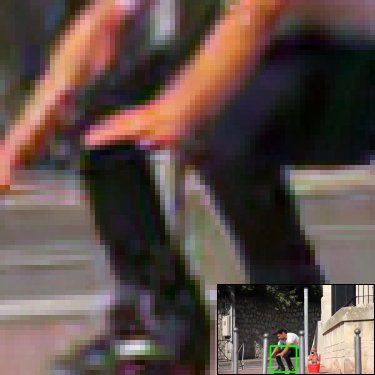}\\
\includegraphics[width=0.09\textwidth]{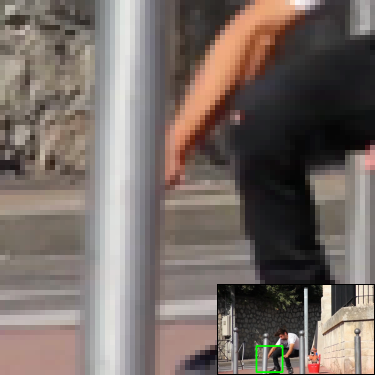}
&\includegraphics[width=0.09\textwidth]{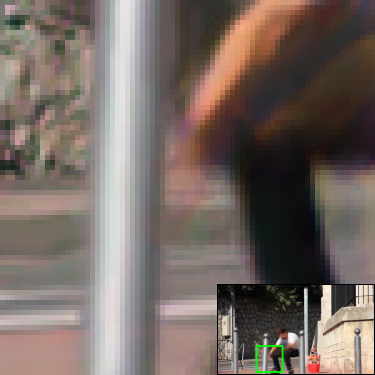}
&\includegraphics[width=0.09\textwidth]{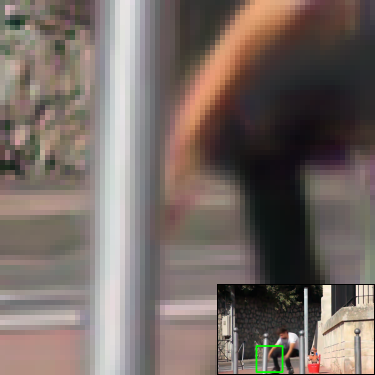}
&\includegraphics[width=0.09\textwidth]{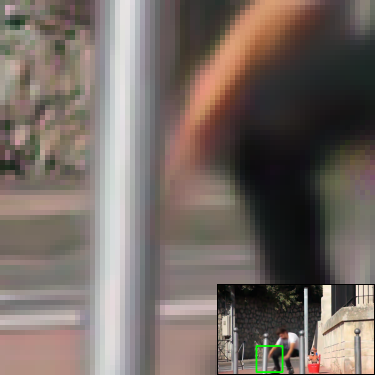}
&\includegraphics[width=0.09\textwidth]{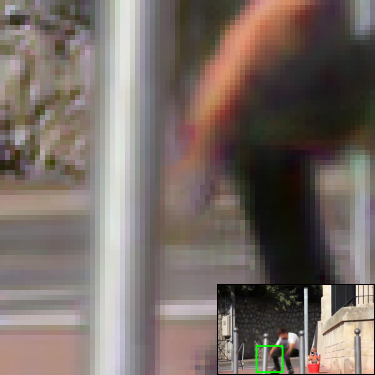}
&\includegraphics[width=0.09\textwidth]{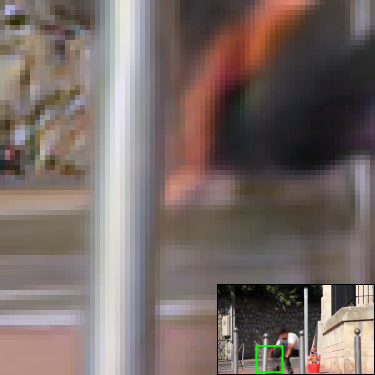}
&\includegraphics[width=0.09\textwidth]{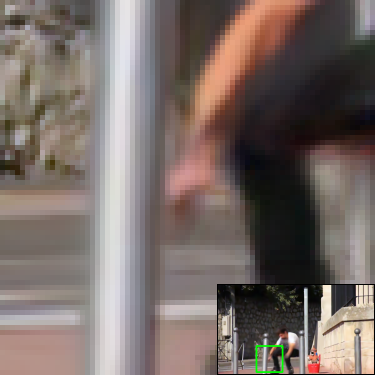}
&\includegraphics[width=0.09\textwidth]{Fig/cross_filled_square.pdf}\\
\includegraphics[width=0.09\textwidth]{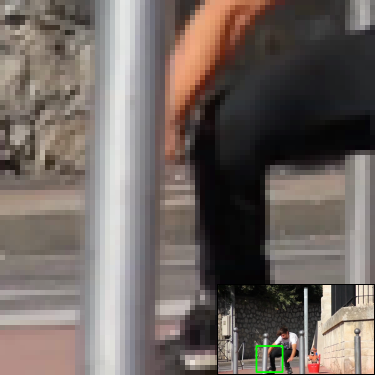}
&\includegraphics[width=0.09\textwidth]{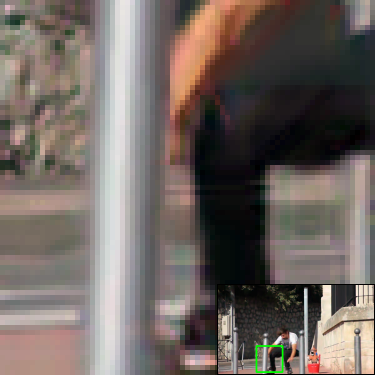}
&\includegraphics[width=0.09\textwidth]{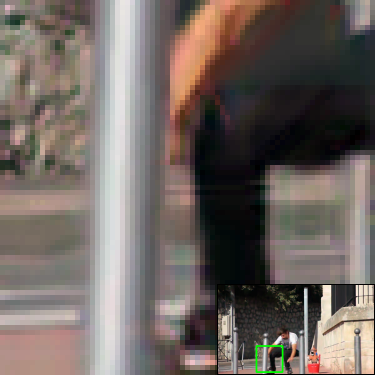}
&\includegraphics[width=0.09\textwidth]{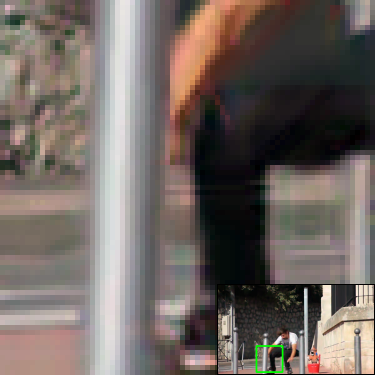}
&\includegraphics[width=0.09\textwidth]{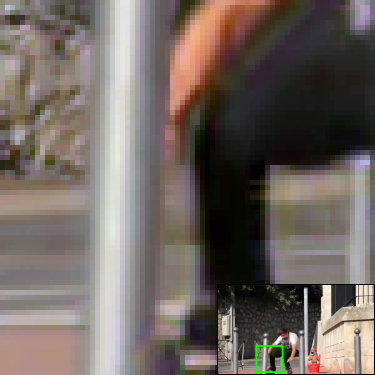}
&\includegraphics[width=0.09\textwidth]{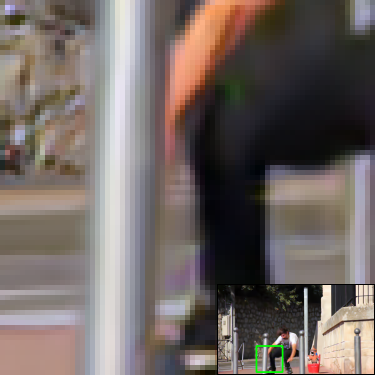}
&\includegraphics[width=0.09\textwidth]{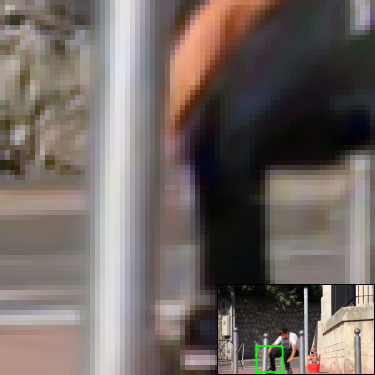}
&\includegraphics[width=0.09\textwidth]{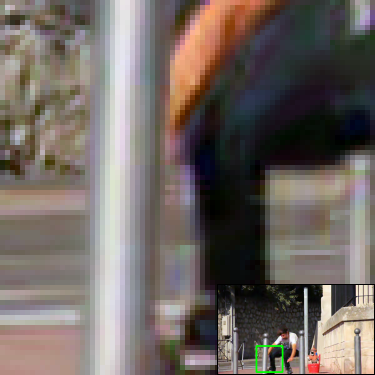}\\
\includegraphics[width=0.09\textwidth]{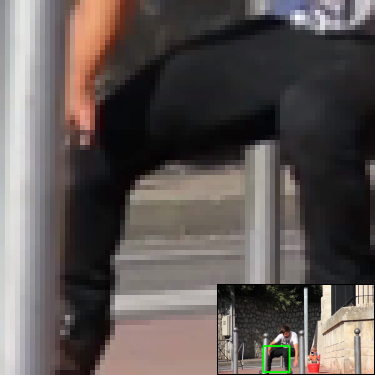}
&\includegraphics[width=0.09\textwidth]{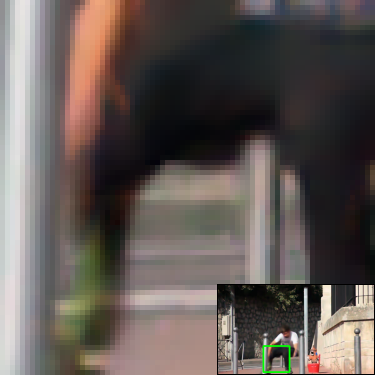}
&\includegraphics[width=0.09\textwidth]{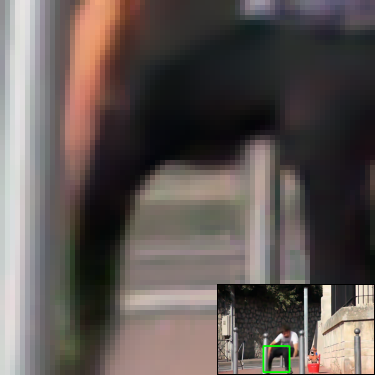}
&\includegraphics[width=0.09\textwidth]{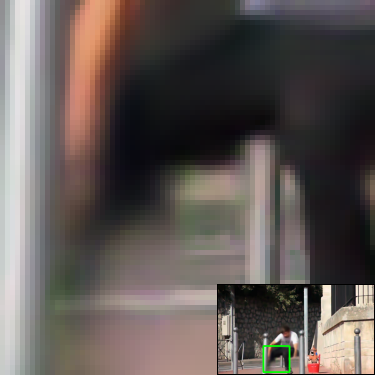}
&\includegraphics[width=0.09\textwidth]{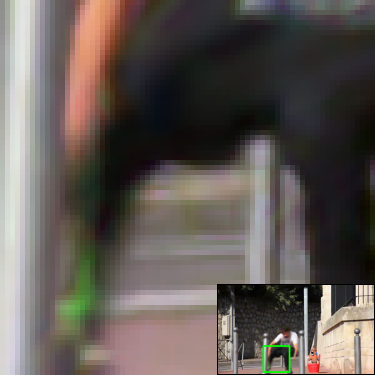}
&\includegraphics[width=0.09\textwidth]{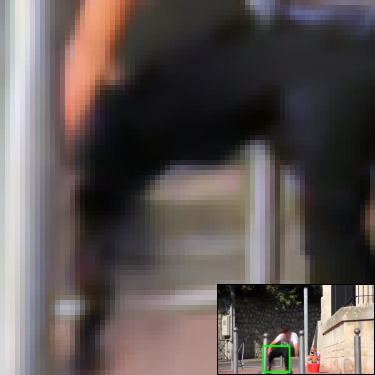}
&\includegraphics[width=0.09\textwidth]{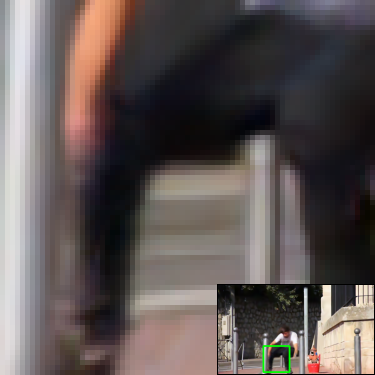}
&\includegraphics[width=0.09\textwidth]{Fig/cross_filled_square.pdf}\\
\includegraphics[width=0.09\textwidth]{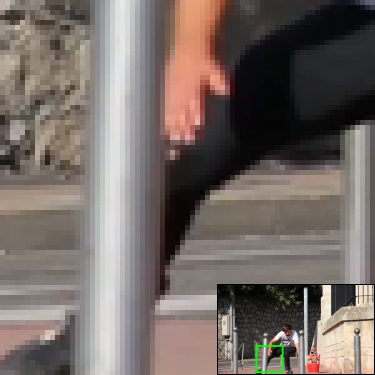}
&\includegraphics[width=0.09\textwidth]{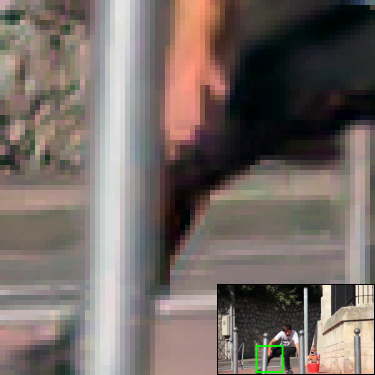}
&\includegraphics[width=0.09\textwidth]{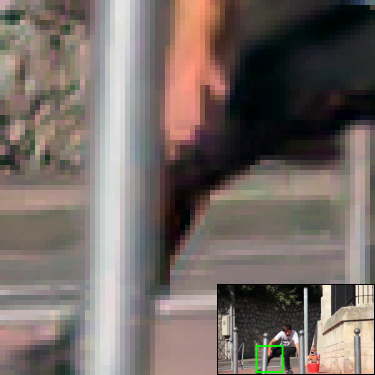}
&\includegraphics[width=0.09\textwidth]{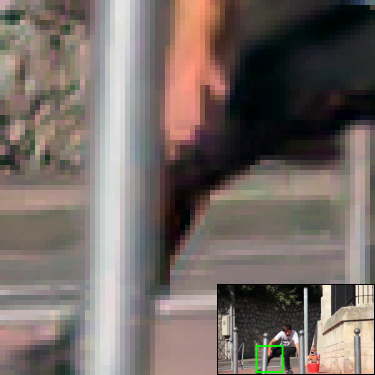}
&\includegraphics[width=0.09\textwidth]{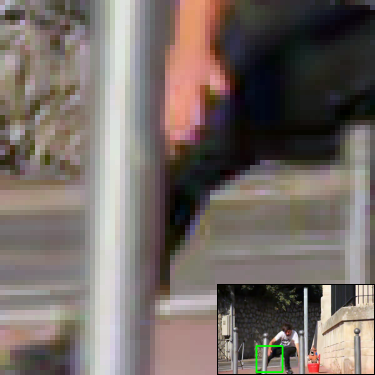}
&\includegraphics[width=0.09\textwidth]{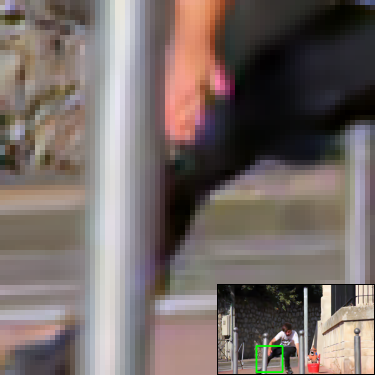}
&\includegraphics[width=0.09\textwidth]{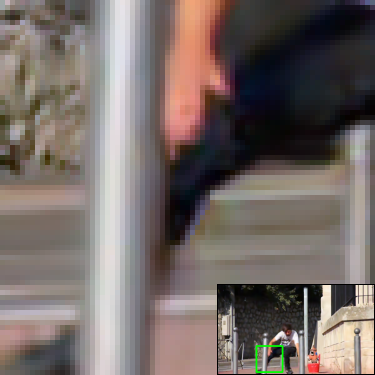}
&\includegraphics[width=0.09\textwidth]{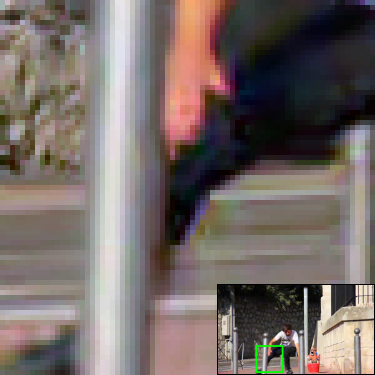}\\
\end{tabular}}

\caption{\textbf{Qualitative comparison of our method and competitive methods on the full-length sequence \textit{00026\_0003} from the Vimeo septuplet dataset~\cite{xue2019video}.} 
We crop the frames for easier comparison and visualize the interpolated frames at the bottom right. TVRN$\downarrow$ donates the downscaled low-frame-rate video.}
\label{fig:vimeo_test_3}
\end{figure*}

\begin{figure*}[!t]
\setlength{\tabcolsep}{0.5pt}
\centering
\label{fig: vimeo_test_4}

\resizebox{1.0\textwidth}{!}{
\scriptsize
\begin{tabular}{cccccc>{\columncolor[HTML]{FFEEED}}cc}
Ground Truth &  UPR-Net L \cite{jin2023unified}  & EMA-VFI \cite{zhang2023extracting} & GIMM-VFI \cite{guo2024generalizable} & STAA \cite{xiang2022learning} & CSTVR* \cite{zhang2025continuous} & \textbf{TVRN (Ours)} & \textbf{TVRN$\downarrow$} \\
\includegraphics[width=0.09\textwidth]{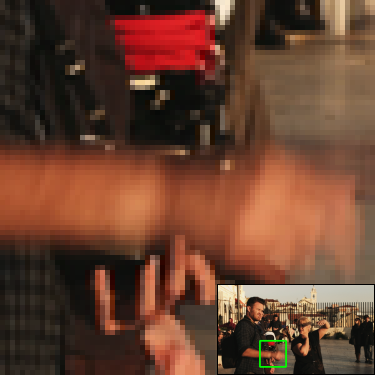}
&\includegraphics[width=0.09\textwidth]{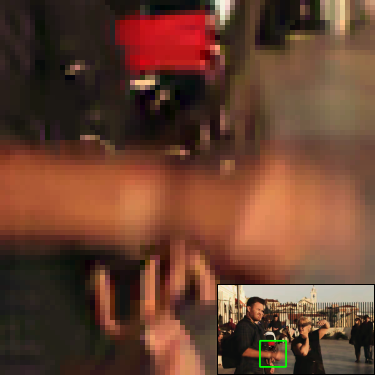}
&\includegraphics[width=0.09\textwidth]{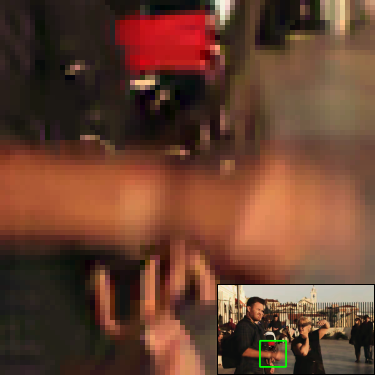}
&\includegraphics[width=0.09\textwidth]{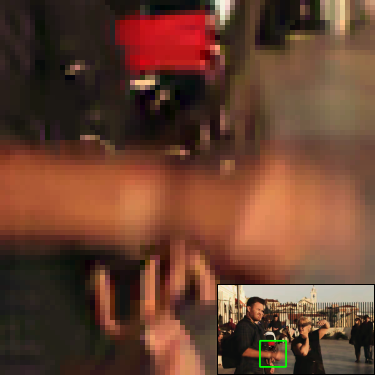}
&\includegraphics[width=0.09\textwidth]{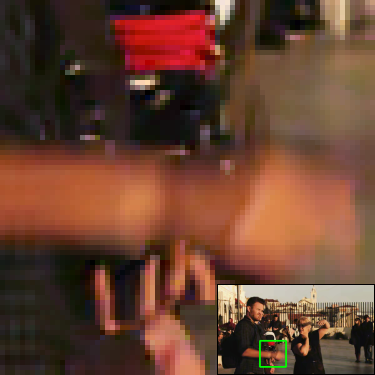}
&\includegraphics[width=0.09\textwidth]{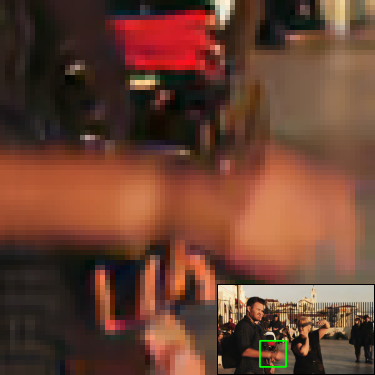}
&\includegraphics[width=0.09\textwidth]{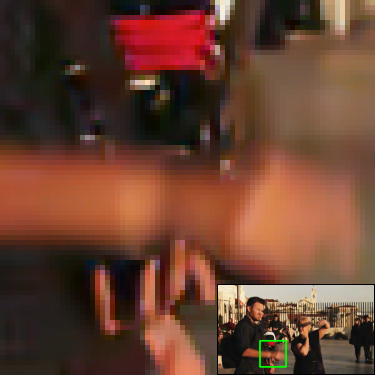}
&\includegraphics[width=0.09\textwidth]{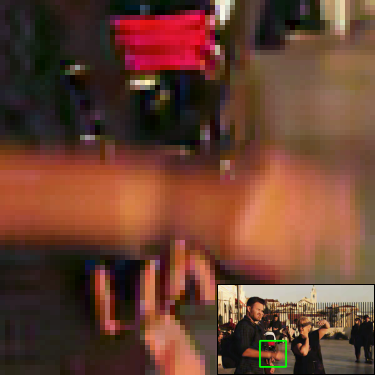}\\
\includegraphics[width=0.09\textwidth]{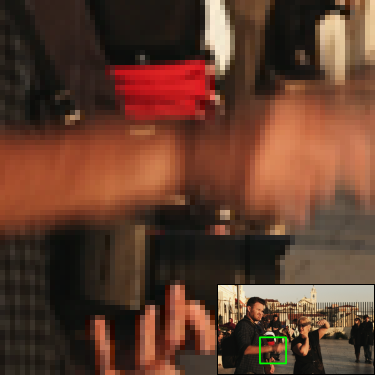}
&\includegraphics[width=0.09\textwidth]{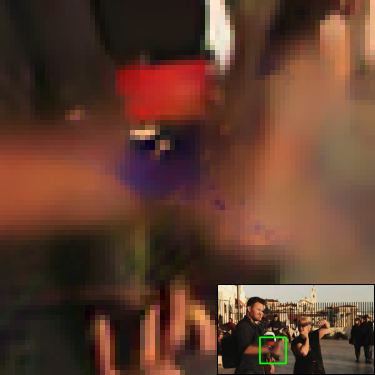}
&\includegraphics[width=0.09\textwidth]{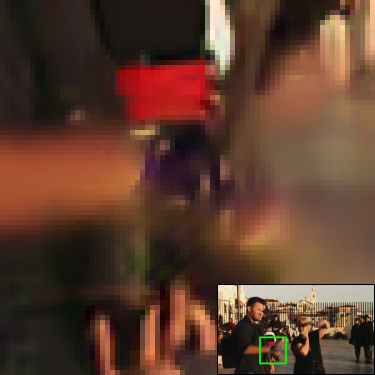}
&\includegraphics[width=0.09\textwidth]{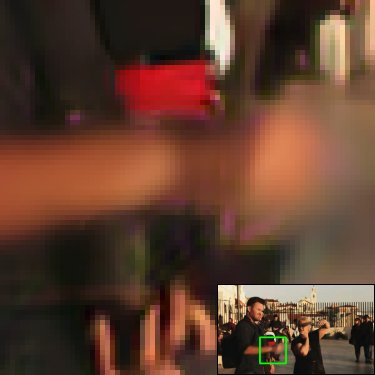}
&\includegraphics[width=0.09\textwidth]{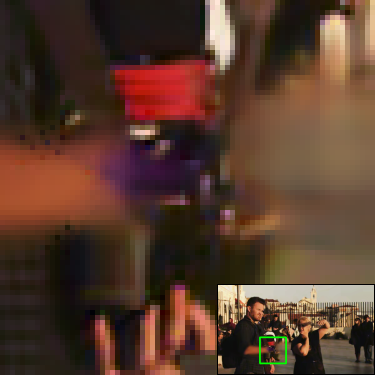}
&\includegraphics[width=0.09\textwidth]{Fig/vimeo_comparison/00002_0238/frame1_qp26_CVRS_finetuned_psnr30.61_ssim0.9206_bpp0.3911.png}
&\includegraphics[width=0.09\textwidth]{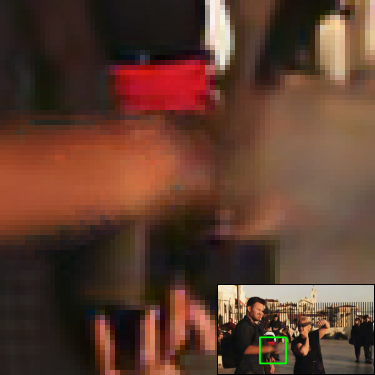}
&\includegraphics[width=0.09\textwidth]{Fig/cross_filled_square.pdf}\\
\includegraphics[width=0.09\textwidth]{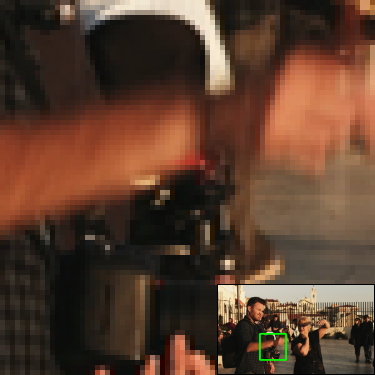}
&\includegraphics[width=0.09\textwidth]{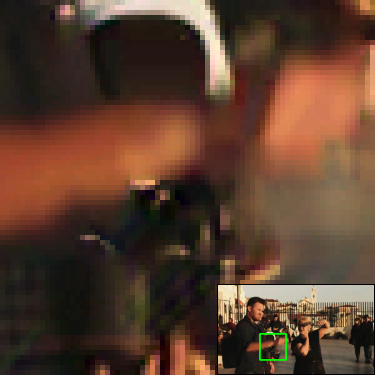}
&\includegraphics[width=0.09\textwidth]{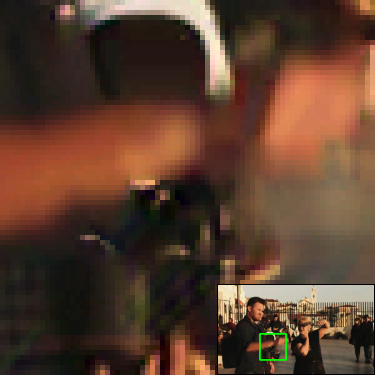}
&\includegraphics[width=0.09\textwidth]{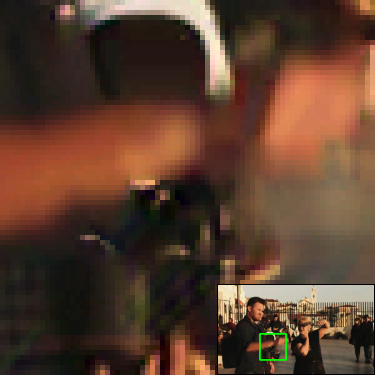}
&\includegraphics[width=0.09\textwidth]{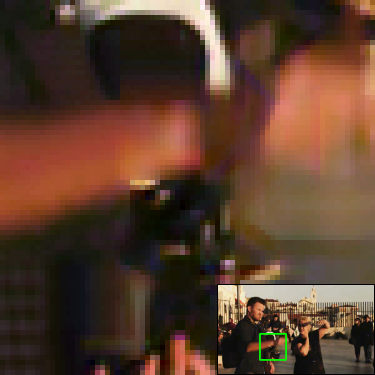}
&\includegraphics[width=0.09\textwidth]{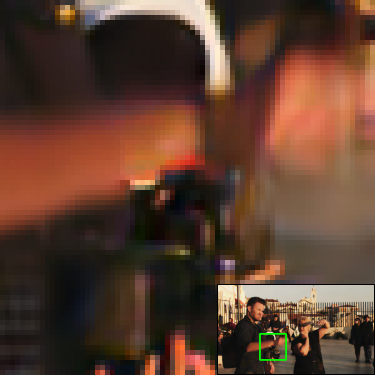}
&\includegraphics[width=0.09\textwidth]{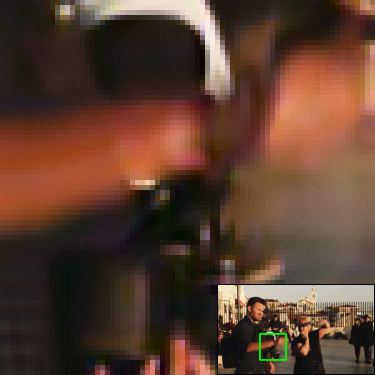}
&\includegraphics[width=0.09\textwidth]{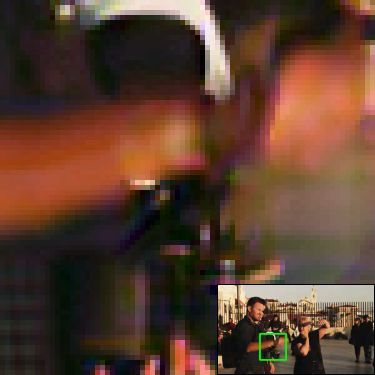}\\
\includegraphics[width=0.09\textwidth]{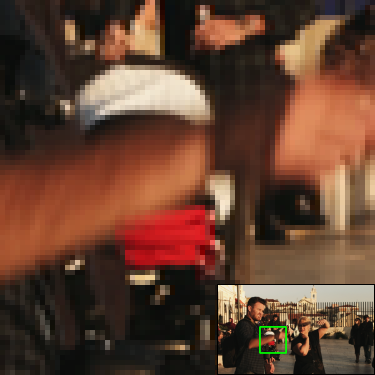}
&\includegraphics[width=0.09\textwidth]{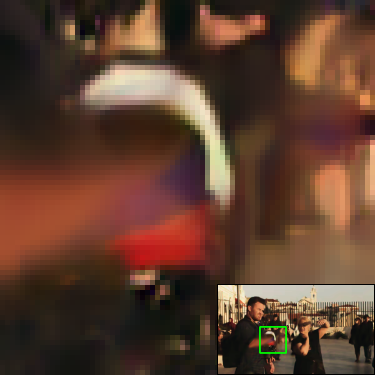}
&\includegraphics[width=0.09\textwidth]{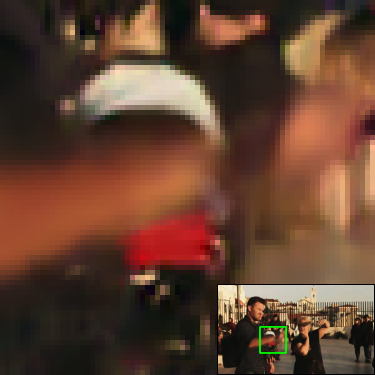}
&\includegraphics[width=0.09\textwidth]{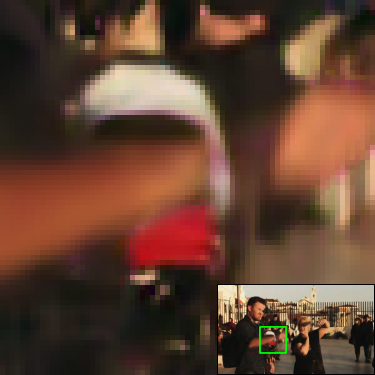}
&\includegraphics[width=0.09\textwidth]{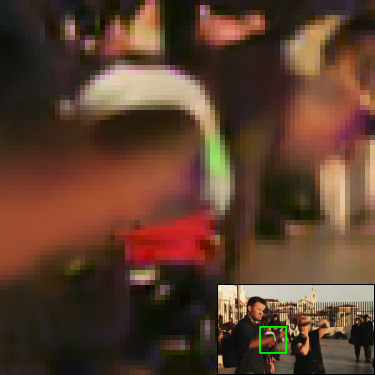}
&\includegraphics[width=0.09\textwidth]{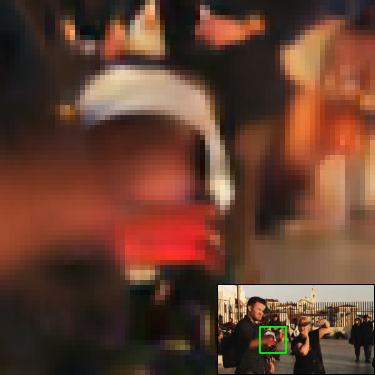}
&\includegraphics[width=0.09\textwidth]{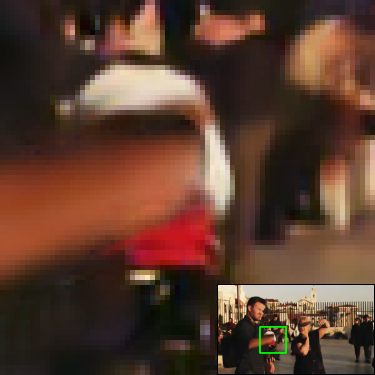}
&\includegraphics[width=0.09\textwidth]{Fig/cross_filled_square.pdf}\\
\includegraphics[width=0.09\textwidth]{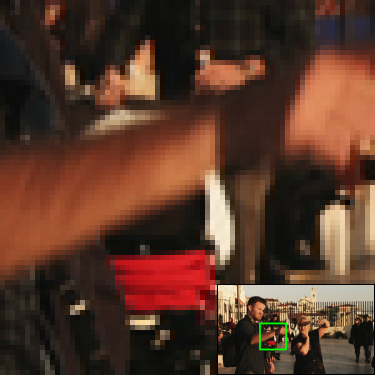}
&\includegraphics[width=0.09\textwidth]{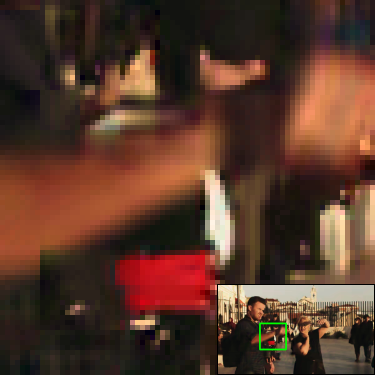}
&\includegraphics[width=0.09\textwidth]{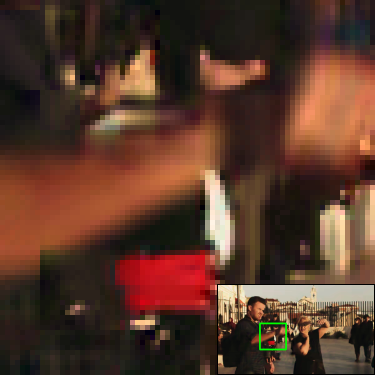}
&\includegraphics[width=0.09\textwidth]{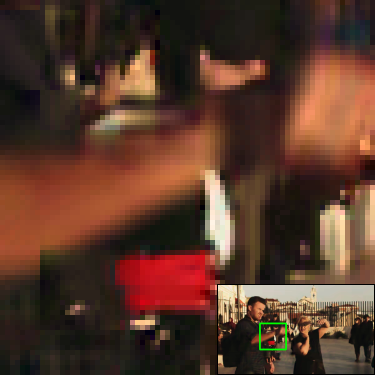}
&\includegraphics[width=0.09\textwidth]{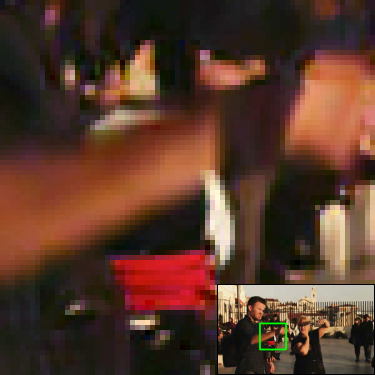}
&\includegraphics[width=0.09\textwidth]{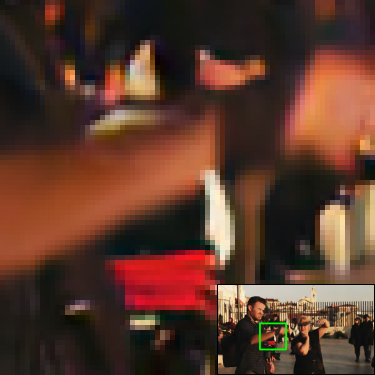}
&\includegraphics[width=0.09\textwidth]{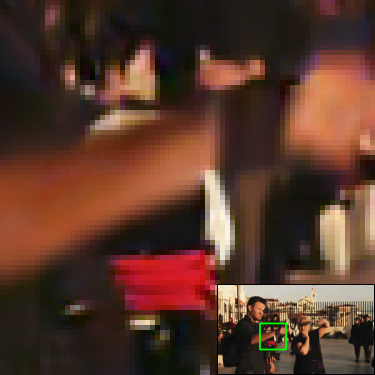}
&\includegraphics[width=0.09\textwidth]{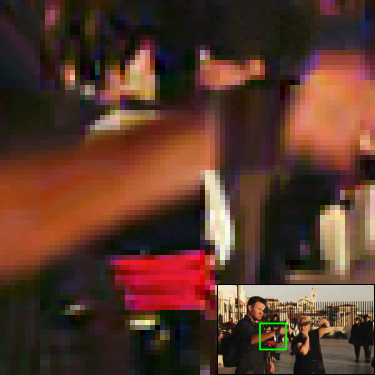}\\
\includegraphics[width=0.09\textwidth]{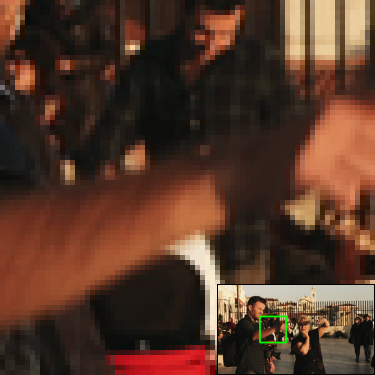}
&\includegraphics[width=0.09\textwidth]{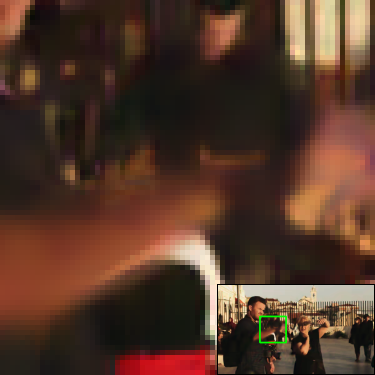}
&\includegraphics[width=0.09\textwidth]{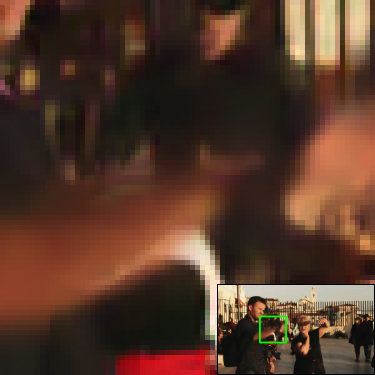}
&\includegraphics[width=0.09\textwidth]{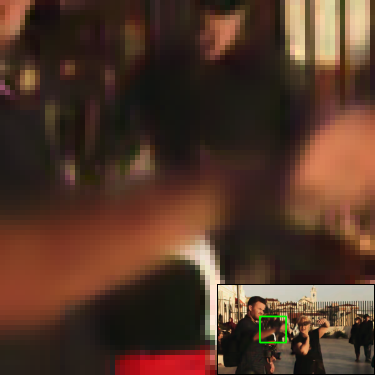}
&\includegraphics[width=0.09\textwidth]{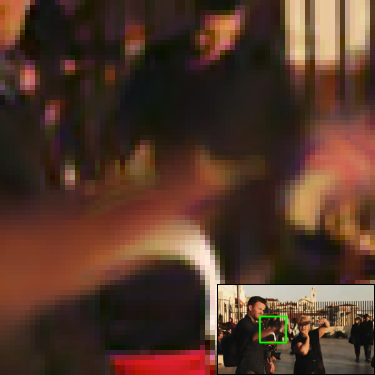}
&\includegraphics[width=0.09\textwidth]{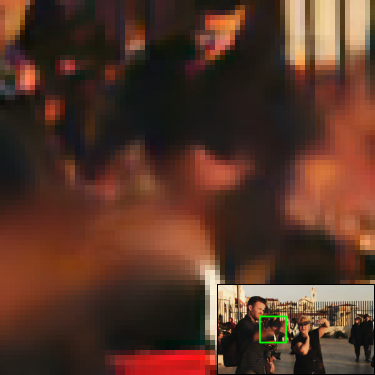}
&\includegraphics[width=0.09\textwidth]{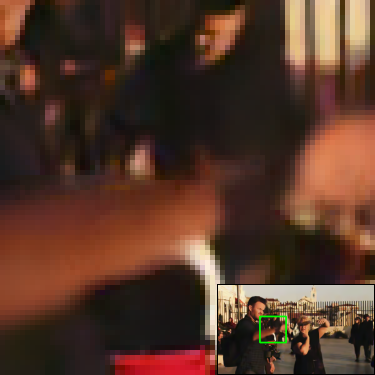}
&\includegraphics[width=0.09\textwidth]{Fig/cross_filled_square.pdf}\\
\includegraphics[width=0.09\textwidth]{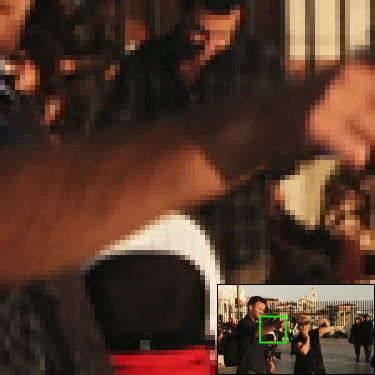}
&\includegraphics[width=0.09\textwidth]{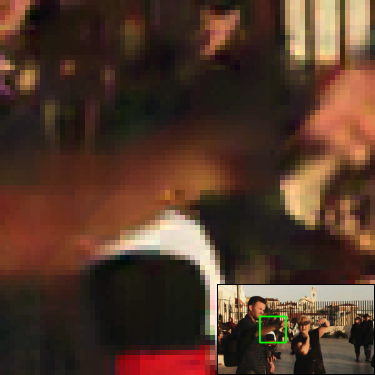}
&\includegraphics[width=0.09\textwidth]{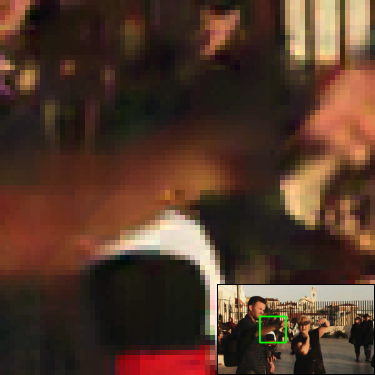}
&\includegraphics[width=0.09\textwidth]{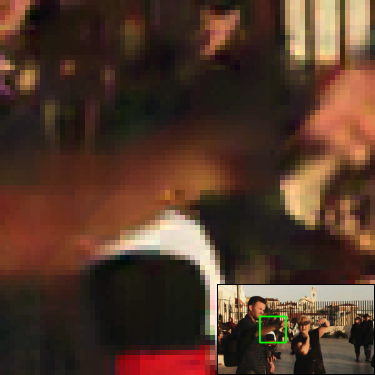}
&\includegraphics[width=0.09\textwidth]{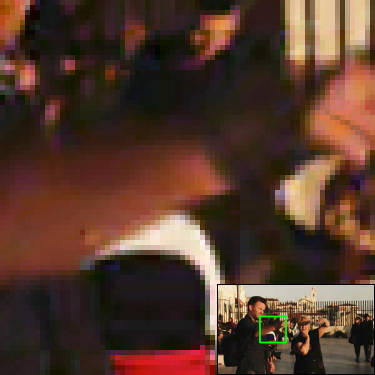}
&\includegraphics[width=0.09\textwidth]{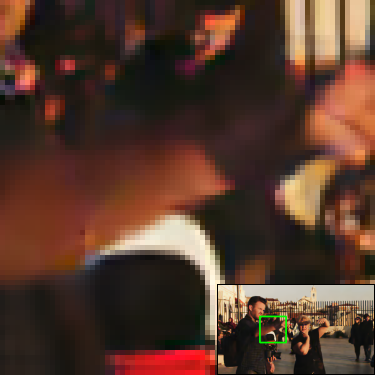}
&\includegraphics[width=0.09\textwidth]{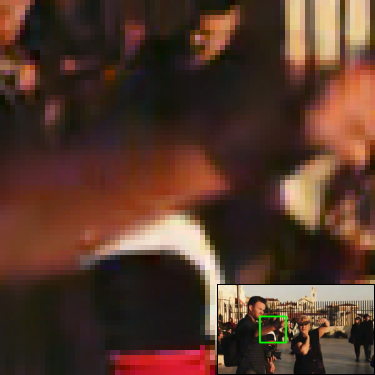}
&\includegraphics[width=0.09\textwidth]{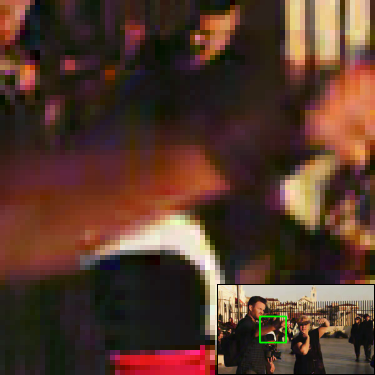}\\
\end{tabular}}

\caption{\textbf{Qualitative comparison of our method and competitive methods on the full-length sequence \textit{00002\_0238} from the Vimeo septuplet dataset~\cite{xue2019video}.} 
We crop the frames for easier comparison and visualize the interpolated frames at the bottom right. TVRN$\downarrow$ donates the downscaled low-frame-rate video.}
\label{fig:vimeo_test_4}
\end{figure*}

\end{document}